%% file: QPT_arxiv_v4.tex
\documentclass[a4paper,onecolumn,superscriptaddress,11pt, accepted=2019-07-16]{quantumarticle}

\pdfoutput=1

\usepackage[T1]{fontenc} 
\usepackage[english]{babel} 
\usepackage[utf8]{inputenc}

\usepackage[numbers,sort&compress]{natbib}

\usepackage[table,usenames,dvipsnames]{xcolor}
\usepackage[colorlinks=true,
			hypertexnames=false
			]{hyperref}
\PassOptionsToPackage{linktocpage}{hyperref} 
\usepackage{amsthm} 
\usepackage{paralist}
\usepackage{array}

\makeatletter 
\newcommand{\thickhline}{%
	\noalign {\ifnum 0=`}\fi \hrule height 1pt
	\futurelet \reserved@a \@xhline
}
\newcolumntype{"}{@{\hskip\tabcolsep\vrule width 1pt\hskip\tabcolsep}}
\makeatother

\newcommand{\myparagraph}[1]{\paragraph*{#1}\hspace{-.8em}:\ }

\usepackage[caption=false, justification=justified]{subfig}

\newcommand{\mySubfloatCaption}[1]{%
	\noindent\refstepcounter{subfigure}{\footnotesize (\thesubfigure)\ \hspace{+.0em}#1}%
	}

\usepackage{pgf,tikz}
\usetikzlibrary{
  positioning,
  fadings,
  decorations.markings
  }

\newcommand{\capstr}[1]{\textbf{#1}}

\usepackage{amssymb} 
\usepackage{mathtools} 
\usepackage{amsmath} 
\usepackage{amsfonts}
\usepackage{mathrsfs} 
\usepackage{dsfont}
\usepackage[vcentermath]{youngtab}

\usepackage{silence}
\input{mymath}
\newcommand{\Jamiolkowski}{Jamio{\l}kowski}

 \makeatletter 
  \hypersetup{pdftitle = {Guaranteed recovery of quantum processes from few measurements},
	     pdfauthor = {Martin Kliesch, Richard Kueng, Jens Eisert, David Gross},
 	     pdfsubject = {Compressed sensing, quantum physics},
 	     pdfkeywords = {Compressive sensing, 
			    quantum process tomography, 
			    quantum channel tomography, 
			    tensor recovery, 
			    convex relaxation, 
			    low-rank matrix recovery, 
			    diamond norm,
			    randomized benchmarking, 
			    quantum gate, 
			    null space property, 
			    descent cone, 
			    bowling scheme, 
			    minimum conic singular value, 
			    recovery guarantee,
			    tensor network
			    }
 	    }
 \makeatother

\newcommand{\fu}{Dahlem Center for Complex Quantum Systems, Freie Universit\"{a}t Berlin, Germany}
\newcommand{\futwo}{Department of Mathematics and Computer Science, Freie Universit\"{a}t Berlin, Germany}
\newcommand{\ug}{Institute of Theoretical Physics and Astrophysics, University of Gda\'{n}sk, Poland}
\newcommand{\hhu}{Institute for Theoretical Physics,
	Heinrich Heine University D{\"u}sseldorf, 
	Germany
}
\newcommand{\uc}{Institute for Theoretical Physics, University of Cologne, Germany}
\newcommand{\syd}{Centre for Engineered Quantum Systems, School of Physics,
The University of Sydney, Australia}

\newcommand{\caltech}{Institute for Quantum Information and Matter, California Institute of Technology, USA}

\newcommand{\hzb}{Helmholtz-Zentrum Berlin f{\"u}r Materialien und Energie, Germany}

\begin{document}
\title{Guaranteed recovery of quantum processes from few measurements} 

\author[MK]{M.\ Kliesch}
\affiliation{\hhu}
\affiliation{\ug}
\email{science@mkliesch.eu}
\homepage{www.mkliesch.eu}
\orcid{0000-0002-8009-0549}

\author{R.\ Kueng}
\affiliation{\uc}
\affiliation{\fu}
\affiliation{\caltech}

\author{J.\ Eisert}
\affiliation{\fu}
\affiliation{\futwo}
\affiliation{\hzb}

\author{D.\ Gross}
\affiliation{\uc}
\affiliation{\syd}

\begin{abstract}
	Quantum process tomography is the task of reconstructing unknown quantum channels from measured data.
	In this work, we introduce compressed sensing-based methods that facilitate the reconstruction of quantum channels of low Kraus rank. 
	Our main contribution is the analysis of a natural measurement model for this task: 
	We assume that data is obtained by sending pure states into the channel and measuring expectation values on the output.
	Neither ancillary systems nor coherent operations across multiple channel uses are required.
	Most previous results on compressed process reconstruction reduce the problem to quantum state tomography on the channel's Choi matrix. 
	While this ansatz yields recovery guarantees from an essentially minimal number of measurements, physical implementations of such schemes would typically involve ancillary systems.
	A priori, it is unclear whether a measurement model tailored directly to quantum process tomography might require more measurements.
	We establish that this is not the case.
	
	Technically, we prove recovery guarantees for three different reconstruction algorithms. 
	The reconstructions are based on a trace, diamond, and $\ell_2$-norm minimization, respectively. 
	Our recovery guarantees are uniform in the sense that with one random choice of measurement settings all quantum channels can be recovered equally well. 
	Moreover, stability against arbitrary measurement noise and robustness against violations of the low-rank assumption is guaranteed. 
	Numerical studies demonstrate the feasibility of the approach.
	\\
\end{abstract}

\maketitle

{
\hypersetup{linkcolor=black} 
\tableofcontents
}

\newpage

\section{Introduction}
Recent years have seen significant advances in the precise control of quantum systems. 
Complex quantum states of systems with an increasing number of degrees of freedom can be prepared 
and manipulated with high accuracy. 
In this development, it is important to have tools at hand that allow for a complete characterization of state or process that are actually being realized in a given experimental setup.
The task of reconstructing quantum states or process \cite{QuantumProcessTomography,PhysRevA.77.032322} from experimental data is variously called quantum state or quantum process \emph{tomography}, \emph{estimation}, or \emph{recovery}.

The precise characterization of processes is highly relevant, e.g., in the quest for scalable quantum computers. 
The stringent requirements of the error correction threshold -- and the adverse scaling of the error correction overhead in terms of the noise -- make it necessary that implementations of quantum gates match their specification extremely closely.
We note that full quantum process tomography is distinct from certification protocols or coarser characterization schemes like randomized benchmarking (which reports only a single number: a certain \emph{error rate}).
While this makes the latter type of protocols much cheaper to implement, only process tomography allows one to understand in precisely which way a quantum gate deviates from its specification. 

The task of quantum process tomography -- important as it is -- comes at a high price: This is an 
unfavourable scaling of the necessary effort with the system size. 
To learn an unknown unstructured process acting on an $n$-dimensional
quantum systems, $m \sim n^4$ expectation values are required. 
However, common processes exhibit additional structure. Most importantly, quantum gates correspond to unitary processes, which are quantum channels with Kraus rank equal to one.
Compressed sensing techniques allow one to estimate channels with (approximately) low Kraus rank from 
significantly fewer measurements than naively required. It is imperative to make use of this structure.

\subsection{Motivation of our measurement model}
The most general quantum process tomography setting allows for ``coherent measurements'' of an unknown channel $T$. 
Here, $m$ identical copies of $T$ are available simultaneously. At the same time,
only part of the input state may be sent through the copies of $T$, while another part potentially
entangled with the former is left unchanged. 
This in turn allows for performing measurements on 
$T^{\otimes m}\otimes \id_{n^m}$ --- where $\id_{n^m}$ denotes the 
identity map on a $n^m$-dimensional additional ancillary Hilbert space that may be available -- rather than ``sequential'' measurements on single copies of $T$.
This includes the possibility of choosing global, entangled input states $\rho_{\textrm{global}}$, and likewise performing correlated measurements on the respective output states 
$(T^{\otimes m}\otimes \id_{n^m}) (\rho_{\textrm{global}})$. 

While potentially powerful, coherent measurements are arguably impractical. 
This is in particular true when attempting to diagnose errors in implementations of a quantum gates, as these are exactly the building blocks of circuits that would be used to prepare the entangled input state in the first place.

Avoiding this drawback, we adopt an experimentally feasible ``sequential'' measurement model. What is more, we resort to the situation
where the local input states $\rho_{\textrm{loc}}$ are transmitted via individual copies of the unknown channel $T$
in their entirety. 
This setting is often referred to as ``ancilla-free'' or ``direct'' process tomography.
Data acquisition is then concluded by performing measurements on the individual output states $T(\rho_{\textrm{loc}})$.

We employ the following natural measurement setup:  the unknown quantum channel $T$ is applied 
to pure input states 
$\ketbra{\psi_i}{\psi_i}$. Subsequently, the expectation value of an observable $A_i$ on the output is measured.
Those expectation values are estimated by means of a suitably frequent repetition of the prescription choosing the
same measurement setting.  
Repeating this procedure $m$ times for different input states $\ketbra{\psi_i}{\psi_i}$ and observables $A_i$ results in a \emph{measurement vector} of the form
\begin{equation}\label{eq:measurement_intro}
	y_i \approx \Tr \bigl[ A_i T (\ketbra{\psi_i}{\psi_i}) \bigr] ,\quad 1 \leq i \leq m. 
\end{equation}
For our theoretical recovery guarantees, the input states and observables are assumed to be sampled independently from certain ensembles -- see below for details.
We restrict attention to pure input states, as they are expected to yield the most information about $T$, and also because the preparation of pure input states is easily possible in most experimental setups.
We expect it to be straight-forward to generalize this work to take mixed input states into account.

\subsection{Our contribution}
We investigate several channel reconstruction protocols that exploit the measured channel to have an (approximately) low Kraus rank. 
Inspired by compressed sensing, these protocols require considerably fewer measurements of the form \eqref{eq:measurement_intro}
than traditional process tomography schemes and contain computationally tractable reconstruction algorithms.
We provide rigorous recovery guarantees featuring very desirable properties (uniform, stable, and robust). 
This analytical work is complemented by a comprehensive numerical analysis of the reconstruction protocols. 
In contrast to previous compressed sensing results that are applicable to process tomography, our reconstruction protocols specifically address the natural process tomography setup described above. 
Our focus is on a particularly simple reconstruction: 
\begin{compactenum}[i)]
\item \label{item:CPTfit}
A  least squares fit of the observed data subject to an additional positivity constraint. 
We call this approach \emph{CPT-fit}. 
A distinct advantage of the CPT-fit is that the algorithm does not depend on any parameters. 
In particular, no estimate of the true rank or the noise strength is required.
This stands in contrast to most compressed sensing-based based reconstruction schemes.
Also, the CPT-fit appears to be particularly robust against violations of the assumption of low Kraus rank. 
\end{compactenum}

\noindent In order to link and compare the CPT-fit to more established low-rank matrix reconstruction methods we also investigate  other reconstruction protocols. 
Their analytical analysis is also the basis for the recovery guarantees of the CPT-fit. 
\begin{compactenum}[i)]\setcounter{enumi}{1}
\item \label{item:TrNorm} 
The second reconstruction method resembles a typical low-rank matrix reconstruction protocol: minimize the trace norm subject to convex constraints that take into account acquired data of the form \eqref{eq:measurement_intro}.
\item \label{item:CTrNorm} 
The third recovery algorithm closely resembles the previous one, but contains trace preservation as an additional (linear) constraint.
\item \label{item:dnorm}
Otherwise similar to algorithm \ref{item:TrNorm}), this reconstruction protocol replaces the constrained trace norm minimization by a diamond norm minimization. The constraints remain unchanged.
The diamond norm is a well-motivated distance measure for quantum channels. 
A priori, this does not imply in any way that the diamond norm is also a good choice as a regularizing function for channel estimation.
However, Ref.~\cite{KliKueEis16} provides strong analytical and numerical arguments for why one should expect the diamond norm to outperform the trace norm as a regularizer in this setting. The numerical studies conducted in this paper lend further credence to this claim.

\item \label{item:Cdnorm}
Similar to algorithm \ref{item:dnorm}), but with an additional trace preservation constraint.

\end{compactenum}

We provide rigorous performance guarantees for the approaches \ref{item:CPTfit}), \ref{item:CTrNorm}) and \ref{item:Cdnorm}). 
The guarantees are valid for measurement setups where the  pure input states $\ket{\psi_i}$ and the eigenbases of the observables $A_i$ are drawn at random from certain ensembles. 
Technically, we require the state vectors and the eigenbases to approximately form \emph{4-designs} in a precise sense. 
Corresponding ensembles of unitaries can be generated efficiently using random quantum circuits \cite{BraHarHor12,BraHarHor16PRL,NakHirKoa17,HarMeh18} or fluctuating Hamiltonians \cite{OnoBueKli17}. 
Haar-random vectors have this property, but recent results on the fourth moments of the Clifford group \cite{ZhuKueGra16,HelWalWeh16} suggest that there are more structured and explicit ensembles with similar properties.

Even if a particular practical setup fails to show the (approximate) 4-design property, our results are still a relevant proof of concept.
Indeed, in contrast to previous compressed sensing approaches, our reconstruction guarantees show that the specific mathematical structure present in process tomography does not imply that a larger number of measurement settings is required.

What is more, we provide numerical simulations showing that the recovery algorithms work well for a number of measurement models not having the 4-design property. This includes the paradigmatic case of Pauli-type measurements.

We demonstrate our channel reconstruction procedures on generic quantum channels of varying Kraus rank, as well as the Toffoli gate. The latter is a three-qubit unitary quantum gate that is highly relevant in quantum processing.
Moreover, it has been experimentally implemented in various architectures \cite{CorPriMaa98,LanBarAlm09,MonKimHan09,FedSteBau12}.
Finally, we also comment that our reconstruction protocols can also be applied in the setting of bosonic and fermionic linear optical circuits. 

\subsection{Related work}
Recent years have seen considerable advance of compressed sensing
tools for the tomographic reconstruction of low-rank
quantum states \cite{FlaGroLiu12,ShaKosMoh11,PhysRevA.90.012110,KimLiu15, KueRauTer15}. 
To some extent, these findings have also been adapted to cover process tomography. 
Such compressed sensing approaches allow to reconstruct processes of low Kraus rank from much fewer
expectation values than a naive dimension count would suggest. 

Our trace norm minimization (algorithms \ref{item:TrNorm}) and \ref{item:CTrNorm})) builds on a by now well-established method for low-rank matrix reconstruction \cite{FlaGroLiu12}.
With the diamond norm minimization (algorithms \ref{item:dnorm}) and \ref{item:Cdnorm})) we continue a line of research that has been started in
Ref.\ \cite{KliKueEis16} and support it with rigorous performance guarantees.

The CPT-fit has been suggested in Ref.\ \cite{PhysRevA.90.012110}
and numerically compared to full tomography and compressed sensing in Ref.\ 
\cite{RodVeiBar14}. 
These works provide numerical studies that demonstrate the capability of such a reconstruction procedure. 
We expand on these ideas by conducting additional numerical experiments. But, more importantly, we also prove rigorous recovery guarantees for the CPT-fit. 
Rigorous recovery guarantees for a similar conic fitting reconstructions in 
(i) a low-rank matrix recovery setup and a certain class of measurements \cite{KabKueRau15} and 
(ii) a sparse vector recovery setup with certain binary measurements \cite{KueJun18}
have been proven recently. 

Ref.\ \cite{FlaGroLiu12} presents a process tomography protocol that is based on low-rank matrix reconstruction with random Pauli measurements \cite{GroLiuFla10,Gro11,Liu11}. 
These low rank recovery guarantees can be applied to the Choi-Jamio{\l}kowski representation of a quantum channel, since the rank of this matrix representation equals the Kraus rank of the original channel.
On first sight, such an approach requires the use of an ancilla in order to implement the Choi-Jamio{\l}kowski representation physically in a concrete application \cite{AltBranJef03}. 
However, Ref.\ \cite{FlaGroLiu12} also provides a more direct implementation of their protocol that does not require any ancillas.
Valid for multi-qubit processes this trick exploits the tensor-product structure of (multi-qubit) Pauli operators. It allows for effectively measuring the Pauli expectation value of a Choi matrix
by performing several natural channel measurements of the form \eqref{eq:measurement_intro} and evaluating particular linear combinations thereof in a classical post-processing step.
The demerit of this approach is that the
number of individual channel measurements required to evaluate a single Pauli expectation value scales with the dimension of the underlying Hilbert space.
The measurement model studied here does not require such a coarse-graining: every natural measurement \eqref{eq:measurement_intro} itself already corresponds to a valid measurement instance.
This considerably reduces the number of different measurement settings that are required in order to acquire sufficient data.

Randomized benchmarking has also been adapted to allow for process tomography \cite{KimSilRya14,KimLiu15,PhysRevLett.121.170502}. 
Ref.\ \cite{KimSilRya14} adopted randomized benchmarking techniques to obtain a process tomography protocol that is robust towards state preparation and measurement errors (SPAM). 
This is a distinct advantage over other protocols that do not share this additional feature. 
However, the trade of between an increase of the number of channel uses and the gained robustness towards SPAM errors remains unclear.   

Finally, Ref.\ \cite{ShaKosMoh11} also considered process tomography via compressed sensing techniques. However, this method is somewhat different from the other approaches presented here: Instead of assuming low Kraus rank, they consider processes that are element-wise sparse with respect to a known basis. Prior knowledge of this sparsifying basis is a necessary prerequisite for this approach. In turn, techniques from traditional compressed sensing \cite{CanRomTao06,Don06} are applied (rather than low-rank matrix reconstruction protocols) to reconstruct such processes from a number of measurements that is proportional to the sparsity and depends only logarithmically on the ambient system sizes. While requiring only very few measurement settings, the main disadvantage of this approach are stronger model assumptions (sparsity and knowledge of the sparsifying basis) and measurements that may be challenging to implement in practice.

\subsection{Experimental considerations}
In several physically important platforms,
process tomography has been experimentally realized \cite{YamNeeLuc10,AltBranJef03,QuantumDot}, 
both in the direct and hence ancilla-free \cite{QuantumDot,YamNeeLuc10} and 
ancilla-based \cite{AltBranJef03} reading of the task.
These works make use of different preparations and measurements.
It should be clear, however, that in several physical
architectures, random measurements of the type discussed here
can readily experimentally be implemented. 

Specifically, Haar
random unitary maps have been realized in a 16-dimensional Hilbert 
space associated with the $6\mathrm{S}_{1/2}$ 
ground state of $^{133}\mathrm{Cs}$ atoms \cite{Jessen}. This
has been achieved by
suitably exploiting a time-dependent Hamiltonian evolution giving rise to Haar
random unitaries, making use of
methods of quantum control. Such random unitaries have been put to use
in a tomographic protocol, in an approach of performing quantum tomography
based on Haar random unitaries that builds upon 
earlier theoretical ideas laid out in Ref.\ \cite{PhysRevA.81.032126}. The same idea 
of generating Haar random unitaries using suitable time-dependent Hamiltonians
should readily apply
to systems of trapped ions \cite{BlaRoo12}, where a suitable 
type of control can be achieved with present technology. Indeed, in a different context, 
methods of optimal control have readily been applied to optimize quantum gates for trapped ions~\cite{OptimalControlWunderlich}.
    
Random unitaries, albeit not
Haar distributed, that allow for quantum state tomography of quantum many-body systems
have been theoretically considered \cite{OhligerTomography}, making use of operations only that 
can be considered basically feasible in experiments with cold atoms in optical lattices \cite{BloDalNas12},
such as optical super-lattices and time-of-flight measurements.

In integrated linear optical architectures \cite{RevModPhys.79.135,Carolan711}, Haar-random circuits 
are readily conceivable \cite{OBrienRandom}. 
It is this type of setting in which the methods presented here are most applicable, 
even though they apply to any physical system in which unitary $4$-designs are feasible.
This includes settings in which random circuits \cite{BraHarHor16,NakHirKoa17}
or randomly time-fluctuating dynamics \cite{OnoBueKli17} can be used to generate approximate unitary designs.

\subsection{Notation and terminology}
In this section we introduce some basic preliminaries needed to understand our main results. 
For any integer $n \in \ZZ^+$ we use the notation $[n]\coloneqq \{1,2,\dots, n\}$. 
The space of linear operators on a vector space $\V$ is denoted by $\L(\V)$. 

If $\V$ is a Hilbert space, we often denote vectors by $\ket{\psi} \in \V$ and its adjoints (conjugate transposes) by $\bra \psi$. 
So the projector onto a normalized $\ket \psi$ is $\ketbra \psi \psi$. 
Moreover, the real vector space of self-adjoint operators is denoted by $\Herm(\V)$. 

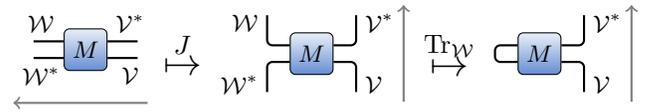
\begin{figure}
	\centering%
	\leavevmode%
	\beginpgfgraphicnamed{fig1_J}%

	\definecolor{niceblue}{rgb}{0.33,0.5,0.8}%
	\tikzset{%
		sbox/.style = {draw, rounded corners = .5ex,%
			minimum height = 1.8\baselineskip,%
			minimum width = 2.2em},
		blau/.style = {top color=niceblue!12,%
			bottom color=niceblue!90},
		Bbox/.style = {sbox, blau},
		leg/.style = {rounded corners = .5ex,thick},
		dir/.style = {gray,thick}
	}
	\newcommand{\Mbox}[1]{%
		\node (M) [Bbox] at (0,0) {#1};%
		\def\d{.4}
		\def\len{.5}
		\def\lenx{.4}%
		\path (M.west) ++ (0,\d\baselineskip) coordinate (Moli);%
		\path (M.west) ++ (0,-\d\baselineskip) coordinate (Muli);%
		\path (M.east) ++ (0,\d\baselineskip) coordinate (More);%
		\path (M.east) ++ (0,-\d\baselineskip) coordinate (Mure);%
	}
	\begin{tikzpicture}[node distance = 1ex]
	\node (M1) {
		\begin{tikzpicture}
		\Mbox{$M$}
		\draw [leg] (More) -- ++(\len,0) node (Xo) [near end, above] {$\V^\ast$};
		\draw [leg] (Mure) -- ++(\len,0) node (Xu) [near end, below] {$\V$};
		\draw [leg] (Moli) -- ++(-\len,0) node (Yo) [near end, above] {$\W$};
		\draw [leg] (Muli) -- ++(-\len,0) node (Yu) [near end, below] {$\W^\ast$};
		\draw [dir, ->] (Xu.south east) ++(0,-1ex) coordinate (ure) -- (ure-|Yu.west);
		\end{tikzpicture}    
	};
	\node (mapstoJ) [right = of M1, inner sep = 0, yshift = .5ex] {\Large{$\mapsto$}};
	\path (mapstoJ.north) node[anchor = south, inner sep = 2pt] {$J$};
	\node (M2) [right = of mapstoJ]{
		\begin{tikzpicture}
		\Mbox{$M$}
		\draw [leg] (More) -- ++(\lenx,0) --++(0,\len) node (Xo) [near end, right] {$\V^\ast$};
		\draw [leg] (Mure) -- ++(\lenx,0) --++(0,-\len) node (Xu) [near end, right] {$\V$};
		\draw [leg] (Moli) -- ++(-\lenx,0) --++(0,\len) node (Yo) [near end, left] {$\W$};
		\draw [leg] (Muli) -- ++(-\lenx,0) --++(0,-\len) node (Yu) [near end, left] {$\W^\ast$};
		\draw [dir, <-] (Xo.north east) ++(.1ex,0) coordinate (ore) -- (ore|-Xu.south);
		\end{tikzpicture}
	};
	   \node (mapsto) [right = of M2, inner xsep = .8em] {\Large{$\mapsto$}};
	   \path (mapsto.north) node (TrY) [anchor=south, inner sep = 0] {$\Tr_\W$};
	   \node (M3) [right = of mapsto]{
	     \begin{tikzpicture}
	     \Mbox{$M$}
	     \draw [leg] (Moli) -- ++(-\lenx,0) |- (Muli);
	     \draw [leg] (More) -- ++(\lenx,0) --++(0,\len) node (Xo) [near end, right] {$\V^\ast$};
	     \draw [leg] (Mure) -- ++(\lenx,0) --++(0,-\len) node (Xu) [near end, right] {$\V$};
	     \draw [dir, <-] (Xo.north east) ++(.1ex,0) coordinate (ore) -- (ore|-Xu.south);
	     \end{tikzpicture}
	   };
\end{tikzpicture}
\endpgfgraphicnamed
\caption{The Choi-{\Jamiolkowski} isomorphism and partial trace in terms of tensor network diagrams (explained in Figure~\ref{fig:TNs}).
	\newline
	\capstr{Left:} 
	Order-$4$ tensor $M\in \L(\L(\V) \to \L(\W))$ as a map from $\L(\V)\cong \V\otimes \V^\ast$ to $\L(\W) \cong \W\otimes \W^\ast$. \newline
	\capstr{Middle:} 
	Its Choi-matrix $J(M)$ as an operator on $\W^\ast \otimes \V\cong \W \otimes \V$. 
	\newline
	\capstr{Right:}
	Partial trace $\Tr_1[J(M)]$ of the Choi matrix $J(M)$. 
	This operator corresponds to the functional $\rho \mapsto \Tr[M(\rho)]$. 
}
\label{fig:JasTN}
\end{figure}

\subsubsection{Maps on operators}
We denote the space of linear \emph{maps} taking operators to operators by $\LL(\CC^n) \coloneqq \L(\L(\CC^n))$. 
The Choi-Jamio{\l}kowski isomorphism \cite{Cho75,Jam72} 
$J: \LL(\CC^n) \to \L(\CC^n \otimes \CC^n)$ takes such maps to operators on a tensor product space and is given by (see also Figure~\ref{fig:JasTN} for a tensor network description)
\begin{equation}\label{eq:choi}
	\begin{aligned}
		J(M)\coloneqq 
		\sum_{i,j=1}^{\dim(\V)} M(\ketbra i j) \otimes \ketbra i j  \, ,
	\end{aligned}
\end{equation}
where $\left\{ |i \rangle \right\}_{i=1}^{\mathrm{dim}(\V)}$ is an (arbitrary) orthonormal basis of $\V$.
Maps $M\in \LL(\CC^n)$ can have different properties:
$M$ is called \emph{hermiticity preserving} if it satisfies 
$M(A^\dagger) = M(A)^\dagger$ for all operators $A\in \L(\CC^n)$.
This is equivalent to demanding that $J(M)$ is Hermitian, i.e., $J(M) \in \Herm(\L(\CC^n))$.
$M$ is called \emph{trace preserving} if $\Tr[M(A)] = \Tr[A]$ for all $A\in \L(\CC^n)$ or, 
equivalently, if $\Tr_1[J(M)] = \1$, where 
$\Tr_1: \L(\CC^n \otimes \CC^n) \to \L(\CC^n)$ denotes the partial trace over the fist tensor factor.
The affine subspace of hermiticity and trace preserving maps is denoted by $\HT(\CC^n) \subset \LL(\CC^n)$. 
The identity map $\id_k \in \LL(\CC^k)$ is a particularly simple element of $\HT(\CC^n)$.
Moreover, $M \in \LL(\CC^n)$ is \emph{completely positive} if for all $k\geq 1$ one has that 
$
	(M\otimes \id_k)(A)
$
is positive semidefinite for every positive semidefinite operator $A \in \L(\CC^n \otimes \CC^k)$. We will use the shorthand notation $A \succeq 0$ to indicate positive-semidefiniteness of $A$.
In fact, a map $M$ is completely positive if and only if  	$(M\otimes \id_n)(A) \succeq 0$ ($k=n$).
This in turn is equivalent to $J(M) \succeq 0$.

The convex subset of completely positive and trace preserving (CPT) maps is denoted by $\CPT(\CC^n) \subset \LL(\CC^n)$. 
Importantly, these maps take density matrices to density matrices, even when applied to subsystems. 
Therefore, they are also called \emph{quantum channels} and are a very general description of quantum processes, e.g.\ dynamics.

\subsubsection{Norms}
For $q\in \left[1,\infty \right[ $, the \emph{$\ell_q$-norm} of a vector $v \in \CC^n$ is 
\begin{equation}
	\lqNorm{v} \coloneqq \left(\sum_{j=1}^n |v_j|^q\right)^{1/q}
\end{equation}
and
the $\ell_\infty$-norm is $\norm{v}_{\ell_\infty} \coloneqq \max_{j\in [n]} |v_j|$.
The \emph{Schatten $q$-norm} $\norm{A}_q$ of an operator $A$
corresponds to the $\ell_q$-norm of $A$'s singular values arranged in an $n$-dimensional vector.
Here, we will specifically encounter the following Schatten norms, 
\begin{align}
\TrNorm{A} &= \Tr\bigl[\sqrt{AA^\dagger}\bigr] && \text{(trace norm),}
\\
\TwoNorm{A} &= \sqrt{\Tr[AA\ad]} && \text{(Frobenius norm),}
\\
\snorm{A} &= \max_{\ket{\psi}} \frac{\lTwoNorm{A\ket{\psi}}}{\lTwoNorm{\ket{\psi}}}
&& \text{(spectral norm).} 
\end{align}
Note that rank-one projectors $\ketbra \psi \psi$ are unit normalized with respect to any Schatten-$p$ norm.

The spectral norm is also an example of an induced norm. This concept can be generalized to norms on maps
$M\in \LL(\CC^n)$. In particular,
\begin{equation}
\norm{M}_{1\to 1} \coloneqq \max_{A} \frac{\TrNorm{M(A)}}{\TrNorm{A}} \, .
\end{equation}
is the induced trace norm.
The \emph{diamond norm} is a stabilized version of this norm:
\begin{equation}
	\dnorm{M} \coloneqq \norm{M \otimes \id_n}_{1\to 1} \, .
\end{equation}
This is a particularly meaningful and widely appreciated distance measure for 
quantum processes \cite{GilLanLie05}. 
Perhaps surprisingly, it can be calculated efficiently \cite{Wat09,BenTa09,Wat12}
as a semidefinite program.
Regarding the reconstruction of structured maps, the diamond norm can serve as a convex surrogate for low Kraus rank \cite{KliKueEis16}.
This is the case for channel reconstruction and two of the five reconstruction protocols presented in this work build on this idea.

\subsubsection{Spherical and unitary designs}
\label{sec:designs}
The probability measures on the unitary group, respectively, the unit complex sphere that is invariant under multiplication with any unitary is called \emph{Haar measure} and is the natural uniform measure on these sets. 
A \emph{complex projective $t$-design} \cite{DelGoeSei77,RenBluRob04,AmbEme07} is a probability distribution $\mu$ on the unit sphere that reproduces the moments of the Haar measure up to order $t$ in both $| \psi \rangle$ and $\langle \psi|$:
\begin{equation}\label{eq:Def:SphericalDesign}
	\EE_{\ket{\psi} \sim \mu} 	\bigl[ (\ketbra\psi\psi )^{\otimes t} \bigr]
	=
	\EE_{\ket{\psi} \sim \Haar} \bigl[ (\ketbra\psi\psi )^{\otimes t} \bigr] \, .
\end{equation}
Similarly, 
a unitary design \cite{GroAudEis07,DanCleEme09} is a probability distribution $\nu$ on the unitary group satisfying 
\begin{equation}\label{eq:Def:UnitaryDesign}
	\EE_{U \sim \nu} 		\bigl[ U^{\otimes t}  X U^{\dagger \otimes t} \bigr]
	=
	\EE_{U \sim \Haar} 	\bigl[ U^{\otimes t}  X U^{\dagger \otimes t} \bigr]
\end{equation}
for all operators $X \in \L(\CC^n)^{\otimes t}$. 
Typically, designs are considered to be uniform distributions over finite sets. 

\subsubsection{Measurement terminology}
In the compressed sensing literature, a single \emph{measurement} is usually given by an inner product, which is a Hilbert-Schmidt inner product
$\Tr [ A T (\ketbra{\psi}{\psi}) ]$ in our case. 
In quantum physics terminology, $(\ket{\psi},A)$ corresponds to a \emph{measurement setting}, whereas a corresponding \emph{measurement} would be the von-Neumann measurement of the observable $A$ in the state $T(\ketbra{\psi}{\psi})$ giving rise to an expectation value $\Tr [ A T (\ketbra{\psi}{\psi}) ]$ in the limit of infinitely many repetitions. 
In order to also account for finite sample errors we always allow for additive noise on the obtained expectation values $\Tr [ A T (\ketbra{\psi}{\psi}) ]$.

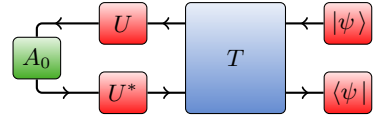
\begin{figure}
\centering
\leavevmode
%
	\beginpgfgraphicnamed{fig2_measurement}%
	\definecolor{niceblue}{rgb}{0.33,0.5,0.8}%
	\colorlet{mygreen}{OliveGreen!90!blue}%
	\tikzset{%
		decoration = {markings,%
			mark=at position 0.6 with {\arrow{>}} },%
		sbox/.style = {draw, rounded corners = .5ex,%
			minimum height = 1.8\baselineskip,%
			minimum width = 2.2em},
		blau/.style = {top color=niceblue!12,%
					   bottom color=niceblue!90},
		Bbox/.style = {sbox, rot, inner sep = 1pt},
		leg/.style = {rounded corners = .5ex,thick,postaction={decorate}},
		dir/.style = {gray,thick},
		rot/.style = {top color=red!10,%
					  bottom color=red!90},
		gruen/.style = {top color=mygreen!10,%
						bottom color=mygreen!70},
		Tbox/.style = {sbox,blau,%
					   minimum height = 4.5\baselineskip,%
					   minimum width = 4.2em},%
	}%
	\begin{tikzpicture}
		\def\len{.5}
		\node (T) [Tbox] {\large{$T$}};
		\path (T.north east) ++ (\len,0) node (psi) [Bbox,anchor = north west]{$\ket\psi$};
		\path (T.south east) ++ (\len,0) node (psiad) [Bbox,anchor = south west]{$\bra\psi$};
		\path (T.north west) ++ (-\len,0) node (U) [Bbox,anchor = north east]{$U$};
		\path (T.south west) ++ (-\len,0) node (Uad) [Bbox,anchor = south east]{$U^\ast$};
		\path (U.west) ++(-\len,0) coordinate (links);
		\node (A) [Bbox,gruen, anchor = east] at (links|-T) {$A_0$};
		\draw [leg] (psi) -- (psi-|T.east);
		\draw [leg] (psiad-|T.east) -- (psiad);
		\draw [leg] (T.west|-U) -- (U);
		\draw [leg] (Uad) -- (Uad-|T.west);
		\draw [leg] (U.west) -| (A.north);
		\draw [leg] (A.south) |- (Uad.west);
	\end{tikzpicture}    
\endpgfgraphicnamed
\caption{The tensorial structure of one measurement setting: 
	$T$ is mapped to $\Tr\bigl[ UA_0U\ad T(\ketbra \psi \psi)\bigr]$. 
}
\label{fig:TNmeasurement}
\end{figure}

\section{Results}\label{sec:results} 
In this section, we explicitly specify our measurement model of natural measurements, 
explain how we computationally reconstruct the quantum channels from the measurements, 
state our recovery guarantees, discuss their stability and robustness properties, discuss our numerics on Pauli-type measurements, and derive a sample complexity upper bound from our recovery guarantees. 

\subsection{Measurement model}
\label{sec:MeasurementModel}
We consider the task of reconstructing a quantum channel $T \in \CPT (\CC^n)$ from measurement data of the form \eqref{eq:measurement_intro}:
The unknown channel receives pure states $\ketbra {\psi_i} {\psi_i}$ as input and subsequently expectation values of observables $A_i$ are measured for the output state; see Figure~\ref{fig:TNmeasurement} for a pictorial description.
This procedure is repeated $m$ times with different measurement settings (input states and observables).
Hence, the entire measurement process leads to a \emph{measurement vector}
\begin{equation}\label{eq:MeasurementMap}
y = \A(T) + e \ \in \RR^m,
\end{equation}
with single expectation values
\begin{equation}\label{eq:MeasurementModel}
	\A(T)_i = \Tr[A_i T(\ketbra{\psi_i}{\psi_i}) ] + e_i.
\end{equation}
The vector $e\in \RR^m$ denotes additive noise present in the measurement process. 
In contrast to previous approaches \cite{Liu11,FlaGroLiu12}, no prior assumptions on the nature of this noise corruption are required. 

Throughout this work, we consider instances of \emph{random} measurements $\left\{A_i,\ket{\psi_i} \right\}_{i=1}^m$ that are independent instances of a measurement \emph{ensemble} 
$\left(A, \ket{\psi} \right)$
that meets the following requirements:

\begin{definition}[4-generic measurement ensemble] \label{def:measurement_ensemble}
We call a measurement ensemble $\left(A, \ket{\psi} \right)$ with observable~$A$ and state~$\ket{\psi}$ \emph{4-generic} if it fulfils the following criteria:
\begin{compactenum}[i)]
\item The distribution of $\ket{\psi}$ in $\CC^n$ is a spherical $4$-design \eqref{eq:Def:SphericalDesign}, i.e., it reproduces the first four moments of the unitarily invariant (Haar) measure on the complex unit sphere.
\item $A=U A_0 U^\dagger$, where $A_0 \in \Herm(\CC^n)$ is fixed and $U$ in $\U(n)$ is chosen from a unitary $4$-design 
\eqref{eq:Def:UnitaryDesign}, i.e., reproduces the first four moments of the Haar measure.
\end{compactenum}

The measurement ensemble is called \emph{normalized 4-generic} if the observables are traceless and normalized in spectral norm, i.e. $\Tr[A_0] = 0$ and $\snorm{A_0} = 1$. 

Corresponding expectation values $\Tr[A T(\ketbra \psi \psi)]$ of a quantum channel $T \in \CPT(\CC^n)$ are referred to as \emph{4-generic measurement} and \emph{normalized 4-generic measurement}, respectively.
\end{definition}

This definition encapsulates a variety of process measurement ensembles. 
In particular, the \emph{generic} measurement ensemble ($\ket{\psi}$ and $U$ are Haar random states and unitaries, respectively) meets the requirements by definition.
However, our recovery guarantees do not require such a strong notion of randomness in the measurement design: $4$-designs are sufficient. 
In fact, it is sufficient when the $4$-designs conditions are only fulfilled approximately, which we discuss and show in Section~\ref{sec:approximate}. 

We also expect that the actual channel recovery works similarly well if $\ket{\psi}$ and $U$ are chosen from other distributions, provided that they cover the complex unit sphere and $U(n)$ ``evenly'' enough 
(see Section~\ref{sec:outlook} for a further discussion). 
We confirm this expectation with numerical simulations with Pauli-type measurements, see Section~\ref{sec:OtherMeasurements}. 

\subsection{Reconstructions}
In this section, we lay out how a quantum channel $T$ can be reconstructed from data of the form
\eqref{eq:MeasurementMap} for measurement settings given by $\A$. We
put an emphasis on describing the fitting method $T^{\ell_2}$ referred to as \emph{CPT-fit}. 
To complement this
approach, we also investigate the reconstruction methods
$T^\ast_\eta $ and $T^{\ast c}_\eta $ as versions of low-rank matrix reconstruction as well as 
$T^\diamond_\eta $ and $T^{\diamond c}_\eta $ as reconstruction methods based on the diamond norm.

To start with the former, 
we minimize the square loss under the model constrained that $T$ is a quantum process to obtain
\begin{align}\label{eq:CPTRec}
T^{\ell_2}
\coloneqq 
\argmin\{ \lTwoNorm{\A(T) - y} : \, T\in \CPT(\CC^n)\} \, . 
\end{align}
This minimization is essentially a fit under a CPT-constraint and we will call it \emph{CPT-fit}.  
It has been suggested and numerically investigated in Ref.\ \cite{PhysRevA.90.012110}. 
This approach makes use of the measurement data $y$ alone. 

More common low-rank matrix reconstruction methods in compressed sensing use the trace norm as a so-called regularizer to favour low-rank solutions \cite{CanRec09,RecFazPar10,Gro11,CanPla11}. 
These reconstructions do not only make use of the data $y$, but require an a-priori bound $\eta>0$ on the noise $e$, 
\begin{equation}
\lTwoNorm{e} \leq \eta \label{eq:noise_bound} \, .
\end{equation}
Applied to the Choi matrix $J(T)$ of a quantum channel $T$, the usual trace norm regularization leads to the reconstructions 
\begin{align}
T^\ast_\eta 
&\coloneqq
\argmin\{ \tnorm{J(T)}: T\in \LL(\CC^n),\ \lTwoNorm{\A(T) - y}\leq \eta\} 
\label{eq:TrNormRec}
\, ,
\\
T^{\ast c}_\eta 
&\coloneqq
\argmin\{ \tnorm{J(T)}: T\in \HT(\CC^n),\ \lTwoNorm{\A(T) - y}\leq \eta\} 
\label{eq:CTrNormRec}
\, .
\end{align}
These two approaches are very similar. 
However, the second one contains an additional constraint that enforces hermicity and trace preservation.

The diamond norm is well-known in quantum information theory as a measure of distinguishability of quantum channels by expectation values \cite[Chapter~11]{KitSheVya02}, \cite{Wat11} and can practically be calculated and minimized using a semidefinite program \cite{Wat12}. 
For non-obvious reasons, it can also be used as a reguralizer for the reconstruction of certain maps on operators \cite{KliKueEis16}, leading to the reconstructions
\begin{align}
T^\diamond_\eta 
&\coloneqq
\argmin\{ \dnorm{T}: T\in \LL(\CC^n),\ \lTwoNorm{\A(T) - y}\leq \eta\} 
\label{eq:dNormRec}
\, ,
\\
T^{\diamond c}_\eta 
&\coloneqq
\argmin\{ \dnorm{T}: T\in \HT(\CC^n),\ \lTwoNorm{\A(T) - y}\leq \eta\} 
\label{eq:CdNormRec}
\, , 
\end{align}
where $\eta\geq 0$ is again an anticipated bound on the measurement noise. 

Superior performance over simple inversion methods is expected for channels with small Kraus rank with 
unitary channels being the most extreme case. 
Perhaps surprisingly at first sight, the CPT-fit performs equally well as the other reconstruction methods based on the trace \eqref{eq:TrNormRec}, \eqref{eq:CTrNormRec} or diamond norm \eqref{eq:dNormRec}, \eqref{eq:CdNormRec} minimization. 
However, a quantum channel $T$ has a Choi matrix $J(T)$ with constant trace norm $\tnorm{J(T)} = n$. 
Hence, one can omit the trace norm minimization in the the program \eqref{eq:CTrNormRec} when the additional constraint $J(T) \succeq 0$ is enforced, i.e.\  $T \in \CPT(\CC^n)$. 
Minimizing the square loss $\lTwoNorm{\A(T) - y}$ instead then yields the CPT-fit \eqref{eq:CPTRec}. 

\begin{figure}
\newgeometry{margin = 1.6cm}
\captionsetup[subfigure]{width=.45\linewidth}
\hspace{-1.5cm}
\subfloat[][
		Success rate of the reconstruction over the number of measurements $m$ for 
		Kraus rank $r=\rank(J(T_0))=2$. \newline
		The rate is the number of trails out of $100$ with small errors 
		$\TwoNormn{J(T^{\rec}-T_0)} \leq 10^{-5}$, for reconstructed channels $T^{\rec}$.
		\hfill\label{sfig:Noiseless:a}%
	]{%
  	\includegraphics[width=.5\linewidth]{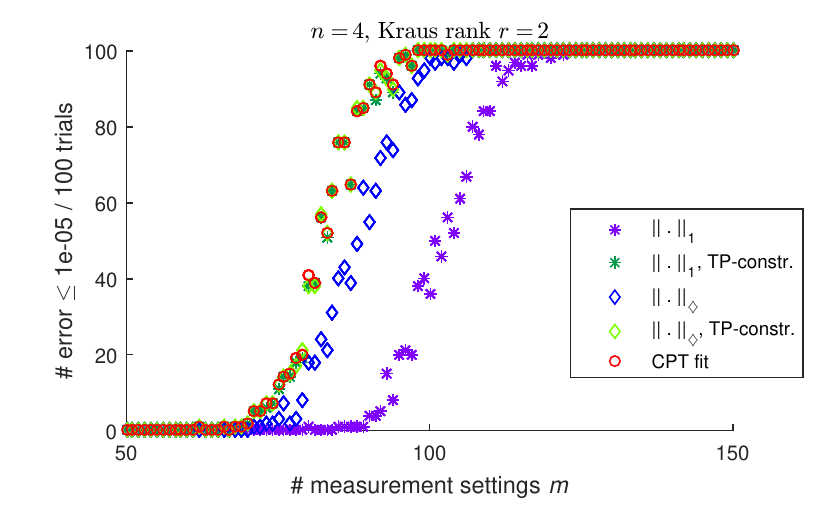}%
	}
\subfloat[][
		Average CPU time required for the reconstruction in~\subref{sfig:Noiseless:a} over the number of measurements $m$.
		\hfill\label{sfig:Noiseless:b}
		]{%
  		\includegraphics[width=.5\linewidth]{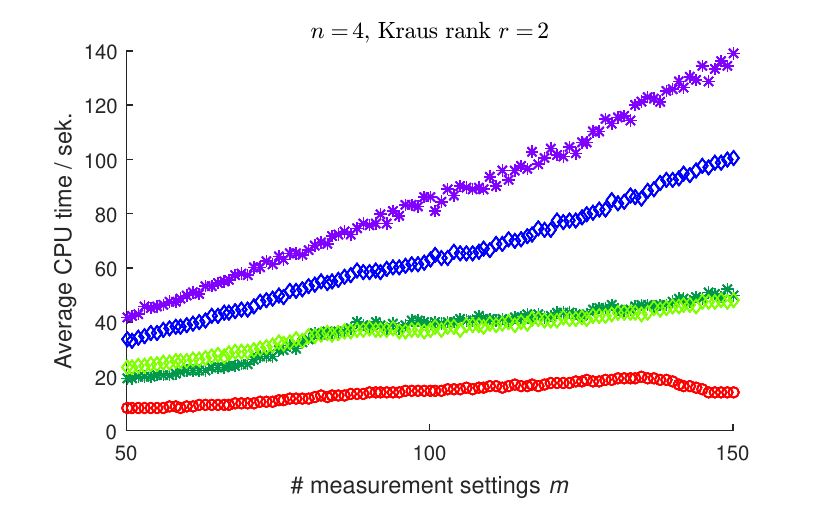}%
	}
\restoregeometry
	\caption{Retrieval of random quantum channels $T_0 \in \CPT(\CC^4)$ for vanishing noise $e=0$. 
		The measurement settings $\A$ are drawn i.i.d.\ in each trial. 
		The parameter $\eta$ in the reconstructions \eqref{eq:TrNormRec}, \eqref{eq:CTrNormRec}, \eqref{eq:CdNormRec} and \eqref{eq:CdNormRec} is chosen to be ten times machine precision. 
	}
	\label{fig:Noiseless}
\end{figure}

We mention once more that existing recovery guarantees for low-rank matrix reconstruction \cite{RecFazPar10,Gro11,Liu11,Kue15} do not apply to the measurement model considered here. 
Our measurement setting \eqref{eq:MeasurementModel} carries a certain product structure 
(see Figure~\ref{fig:TNmeasurement}) that is more restrictive than the structure covered in prior works.

\subsection{Noiseless case}
For sake of clarity, we first discuss the noiseless case ($e=0$) and will discuss stability and robustness in Section~\ref{sec:StabilityAndRobustness}. 
\subsubsection{Recovery guarantees I} 
We have recovery guarantees for the minimizations \eqref{eq:CTrNormRec}, \eqref{eq:CdNormRec}, and \eqref{eq:CPTRec}, where trace preservation features as a constraint. 

\begin{theorem}[Uniform recovery guarantees (noiseless case)]
	\label{thm:Noiseless}
	Fix $r \leq n^2$ and suppose that $\A : \CPT(\CC^n) \to \RR^m$ contains
\begin{equation}
	m \geq C \, r\, n^2\ln(n) \, ,\label{eq:mScaling}
\end{equation}
	4-generic measurements which are chosen independently from the measurement ensemble introduced in Definition~\ref{def:measurement_ensemble}. 
	Then, with probability at least 	$1-\e^{-\lambda m}$, 
	all $T_0 \in \CPT(\CC^n)$ with Kraus rank at most $r$ may be reconstructed exactly from noiseless measurements $y=\mathcal{A}(T_0)$. 

	This exact reconstruction is achieved by each of the minimizations \eqref{eq:CTrNormRec}, \eqref{eq:CdNormRec}, or \eqref{eq:CPTRec}. 
	Here, $C$ and $\lambda$ are universal constants (which can, in principle, be extracted explicitly from the proof).
\end{theorem}

This statement follows from more general results, namely Theorems~\ref{thm:TrNorm}, \ref{thm:CPTfit}, and~\ref{thm:dNorm} below. 
Stability and robustness are discussed with Theorems~\ref{thm:dNormStability} and~\ref{thm:robustness}. 

\myparagraph{Remarks}
\begin{compactenum}[i)]
\item 
	The number of real parameters required to specify a general hermiticity preserving map $T_0 \in \LL(\CC^n)$ is 
	$n^4$. 
	If said map is also trace preserving, then $n^4-n^2$ many parameters are required. 
	The same holds for $T_0\in \CPT(\CC^n)$. 
	If the Kraus rank of $T_0$ is $\rank(J(T_0))=r$ then $T_0$ is compressible in the sense that an order of $r n^2$ parameters suffices to describe $T_0$. 
	Therefore, up to the factor $\ln(n)$, the scaling \eqref{eq:mScaling} of the number of measurements is optimal. 
\item 
	The scaling \eqref{eq:mScaling} also contains a constant $C$ that we have not bounded explicitly. 
	Numerically, we find for $n=4,8$ that $m\geq 5 \, r\, n^2$ generic (Haar-random) measurements are sufficient for recovery, see Sections~\ref{sub:noiseless_numerics} and~\ref{sec:ToffoliNumerics}. 

	Such an approach is very common in compressed sensing: 
	Recovery guarantees, often with unknown constants and logarithmic factors in the scaling of the number of measurements, ensure the functioning of the reconstruction procedure also in the limit of large dimensions. 
	The precise number of measurements and the expected errors in a given setting are often determined numerically. 
	Here, also more special measurement settings can practically be considered. 
	In Section~\ref{sec:OtherMeasurements}, we numerically investigate Pauli-type measurements, which are not covered by our recovery guarantees. 
\item 
	The measurements are required to be given by exact $4$-designs. 
	This condition can be relaxed so that it is sufficient if the $4$-design conditions are fulfilled only approximately. 
	We prove that the recovery guarantees still hold for certain $\epsilon/n^4$-approximate $4$-designs in Theorem~\ref{thm:approximate}.
	One can use quantum pseudorandomness generation with random quantum circuits \cite{BraHarHor16,BraHarHor16PRL,NakHirKoa17,HarMeh18} or fluctuating Hamiltonians \cite{OnoBueKli17} to generate $\epsilon$-approximate designs. 
	In both cases, $\epsilon$ becomes exponentially small in the generation time (circuit length/runtime), i.e., $\epsilon/n^4$ can be made small efficiently. 
	In this case, the measurement map needs to be obtained from the gate sequences or the randomly fluctuating classical parameters of the Hamiltonian, respectively. 

\end{compactenum}

\subsubsection{Numerical demonstration} \label{sub:noiseless_numerics}
Our variational reconstructions
can be recast as standard convex optimization problems (see also Appendix~\ref{sec:SDPs}). 
These can be solved computationally efficiently\footnote{with computation time scaling polynomially in the dimension} and also practically using standard software such as CVX \cite{cvx,GraBoy08}. 
For the reconstructions~\eqref{eq:CPTRec}, \eqref{eq:CTrNormRec}, and~\eqref{eq:CdNormRec} that have trace preservation as a constraint we have proven that a number of
\begin{equation} \label{eq:mScalingDiscussion}
	m \geq m_0 \coloneqq C\, r\, n^2 \ln(n)
\end{equation}
measurement settings
is sufficient to recover quantum channels of Kraus rank $r$. 
We expect similar reconstruction properties for the unconstrained trace and diamond norm 
reconstructions~\eqref{eq:TrNormRec} and~\eqref{eq:dNormRec}, although our proofs do not directly carry over to that case. 
Indeed, low-rank matrices can often be recovered via trace norm minimization from a number of measurements with such an essentially optimal scaling \cite{RecFazPar10,Gro11,Liu11,KueRauTer15,KabKueRau15}. 

Different reconstruction procedures with the optimal scaling \eqref{eq:mScalingDiscussion} often have a different constant $C$.
For instance, Ref.~\cite{Tro15} shows that a constant of size $C \simeq 6$ provably suffices for Gaussian measurement matrices (even without the $ln(n)$ factor). 
On the other hand, the results presented in Refs.\ \cite{Gro11,Liu11} are valid for more structured Pauli measurements and require a much larger constant.
Numerical studies typically highlight a similar behaviour.

Exploiting additional prior information -- such as the fact that quantum channels preserve both hermiticity and traces -- in the algorithmic reconstructions \eqref{eq:CTrNormRec}, \eqref{eq:CdNormRec} can only lead to an improvement.
In fact, these constraints also facilitate the mathematical analysis. 
Besides being able to prove recovery guarantees, we also observe the benefit of such additional constraints numerically:
Figure~\ref{sfig:Noiseless:a} shows the \emph{recovery rates} of different approaches.

This recovery rate is determined as follows. 
We say that a channel $T_0$ is successfully reconstructed if the reconstructed channel $T^\rec$ is close to the original one: 
$\TwoNorm{J(T^\rec-T_0)} \leq 10^{-5}$. 
We have chosen $10^{-5}$ as threshold because we have observed the reconstruction errors to be typically well separated from this value. 
Varying it changes the curves in the plot only slightly. 
The precision of the convex optimization software CVX~\cite{cvx,GraBoy08} in terms of the machine precision $\eps = 2^{-52} \approx 2.2\cdot 10^{-16}$ is $\sqrt{\eps} \approx 1.4\cdot 10^{-8}$. This somewhat limits the choice of thresholds. 
In order to obtain the rate shown in the plots, 
we run $100$ independent instances and obtain the rate as the percentage of trials with a successful reconstruction.

Our numerical studies are based on generic measurement ensembles (Haar random input states, Haar random unitaries) and we have chosen $A_0$ to be a diagonal $n \times n$ matrix whose spectrum evenly covers the interval $[-1,1]$. 

Interestingly, the unconstrained diamond norm minimization \eqref{eq:dNormRec} performs almost as good as its constrained counterpart \eqref{eq:CdNormRec}. 
The heuristic reason for this behaviour is that the diamond norm ``favours'' maps that satisfy a certain trace preservation condition \cite{KliKueEis16,MicKliKue18}, which is fulfilled for CPT maps. 
Moreover, the so-called descent cone at CPT-maps of the diamond norm is contained in an intersection of descent cones of nuclear norms \cite{KliKueEis16}. 
If this containment is strict, it also leads to an improved recovery recovery guarantee. 

Arguably the simplest reconstruction procedure is the fit under the CPT constraint \eqref{eq:CPTRec}. 
Here, we observe that this protocol achieves a rate that is very similar to the other constrained procedures~\eqref{eq:CTrNormRec}, \eqref{eq:CdNormRec}.  
Moreover, the CPT-fit \eqref{eq:CPTRec} clearly has the fastest computation time in our simple implementation in CVX~\cite{cvx,GraBoy08} with Mosek~7.1 as a solver, see Figure~\ref{sfig:Noiseless:b}. 

The required number of measurement settings~\eqref{eq:mScalingDiscussion} depends linearly on the Kraus rank $r$ of the channels to be reconstructed.
This dependence is confirmed for small $r$ in our numerics for Hilbert space dimension $n=4$, see Figure~\ref{fig:RankT}. 
For dimension $n=4$ the Kraus rank is $r \leq n^2 =16$ and for $r = 16$,
 $m= \dim_{\RR}(\HT(\CC^n)) = n^4-n^2 = 252$ 
of measurement values are required for the reconstruction. 
Here, $\dim_{\RR}(\HT(\CC^n))$ denotes the real dimension of the affine space $\HT(\CC^n)$. 
This observation explains the non-linear behaviour for larger~$r$. 

\begin{figure}
	\centering
	\includegraphics[width = .65\linewidth]{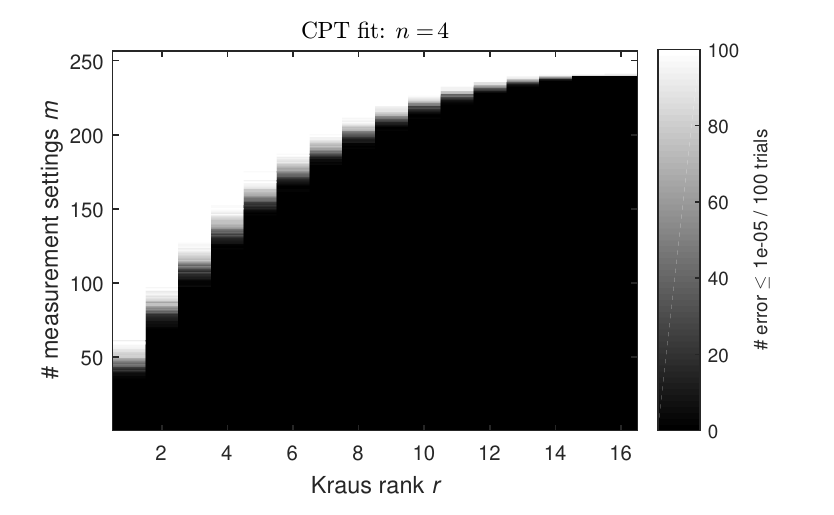}
	\caption{
		Reconstruction of random quantum channels $T_0 \in \CPT(\CC^4)$ with different Kraus ranks $r=\rank(J(T_0))$ from the CPT-fit \eqref{eq:CPTRec}. 
		The white region corresponds to $100\%$ observed recovery and the black one to $0\%$. 
		Parameters and measurements are as in Figure~\ref{fig:Noiseless}.
	}
	\label{fig:RankT}
\end{figure}

\subsection{Stability and robustness}\label{sec:StabilityAndRobustness}
We consider two types of errors: 
i) the measurement errors $e$ already indicated in the measurement model \eqref{eq:MeasurementModel} and 
ii) model mismatches capturing violations of the assumption 
$\rank(J(T_0)) \leq r$ on the Kraus rank. 

For this purpose, we need a technical definition of optimal low-rank approximations of quantum channels. 
We fix a maximum rank $r$.
For any operator $X \in \Herm(\CC^d)$ we define
\begin{align}
X\r &\coloneqq \argmin \bigl\{ \tnorm{X-Y}: \ Y \in \Herm(\CC^d) , \ \rank(Y)\leq r \bigr\} \, ,
\\
X\c &\coloneqq X-X\r \, .
\end{align}
The low-rank approximation $X\r$ and, hence, also $X\c$ has a simple formula:
we use the eigenvalue decomposition $X = \sum_{i=1}^d x_i \ketbra{x_i}{x_i}$ with $|x_1| \geq |x_2|\geq \dots\geq |x_d|$. 
Then 
\begin{equation}
	X\r = \sum_{i=1}^r x_i \ketbra{x_i}{x_i} \, .
\end{equation}

This construction is inherited by linear maps on operators. 
Specifically, for $T\in \CPT(\CC^n)$ we set  
\begin{equation}\label{eq:tail_def}
	J(T\r) \coloneqq J(T)\r \quad \text{and}\qquad J(T\c) \coloneqq J(T)\c \, . 
\end{equation}
The error term $T\c$ is called the \emph{model mismatch}. 
In turn, we relax our model assumption of low Kraus rank and allow for small model mismatches. 

A reconstruction is called \emph{stable} if it tolerates measurement errors and it is called \emph{robust} if it tolerates model mismatches. 

\subsubsection{Recovery guarantees II} 
We prove all reconstructions from Theorem~\ref{thm:Noiseless} to be stable against measurement noise. 
Moreover, we prove the trace norm minimization \eqref{eq:CTrNormRec} and the CPT-fit \eqref{eq:CPTRec} also to be robust against model mismatches. 

\begin{theorem}[Stability of the diamond norm reconstruction \eqref{eq:CdNormRec}]
	\label{thm:dNormStability}
	Consider normalized 4-generic measurements from Definition~\ref{def:measurement_ensemble}. 
	Then, under the hypotheses of Theorem~\ref{thm:Noiseless}, reconstruction via the constrained diamond norm minimization \eqref{eq:CdNormRec} is stable towards additive noise in the measurements $y = \A (T_0)+e$: 
	for any $\eta \geq \lTwoNorm{e}$, the associated reconstruction error obeys
		\begin{equation}\label{eq:ErrorBoundStability}
		\TwoNorm{ J(T_0-T^{\diamond c}_\eta) }
		\leq  
		\tilde c\, \frac{n^2}{\sqrt{m}\, \TwoNormn{A_0}} \, \eta \, .
	\end{equation}
	The constants $C$, $\lambda$, and $\tilde c$ only depend on each other. 
\end{theorem}

A more general version of this theorem -- which allows for quantifying the noise strength by any $\ell_q$-norm -- is also true, see 
Corollary~\ref{thm:dNorm} below. 

\begin{theorem}[Stability and robustness, trace norm minimization \eqref{eq:CTrNormRec} and CPT-fit \eqref{eq:CPTRec}]
	\label{thm:robustness}
	Under the hypotheses of Theorem~\ref{thm:Noiseless}, reconstruction via
	the CPT-fit \eqref{eq:CPTRec} is both stable towards additive noise corruption and robust with respect to relaxing the model assumption of Kraus rank $r$:
\begin{equation}\label{eq:FitError}
		\TwoNorm{ J(T_0-T^{\ell_2}) }
		\leq \\
		4 \tnorm{J({T_0}\c)}  
		+ \tilde c \, \frac{n^2}{ \sqrt{m} \, \TwoNormb{A_0} } \lTwoNorm{e}
		\, .
\end{equation}
This performance guarantee is also valid for the constrained trace norm minimization \eqref{eq:CTrNormRec} if one replaces $\| e \|_{\ell_2}$ by the optimization parameter $\eta$, provided that $\eta \geq \| e \|_{\ell_2}$

Once more, the constants $C$, $\lambda$, and $\tilde c$ again only depend on each other. 
\end{theorem}

The model mismatch is quantified by the trace norm of the Choi matrix. 
This is a relatively simple error measure that upper bounds the operationally relevant diamond norm. 

Theorem~\ref{thm:robustness} summarizes two results -- Theorems~\ref{thm:TrNorm} and~\ref{thm:CPTfit} below -- that are slightly more general: 
they allow for quantifying the reconstruction error~\eqref{eq:FitError} and~ by any Schatten-$p$ norm with $p \in [1,2]$. 
Moreover, the noise strength may be characterized by any $\ell_q$-norm: $\eta~\geq~\|e\|_{\ell_q}$.

Naturally, one would expect the statements of the recovery guarantees from Theorem~\ref{thm:robustness} to be invariant under i) a simultaneous rescaling of the noise strength $\eta$ and the measured observables and ii) under adding a multiple of the identity operator to the observables and shifting the measurement data accordingly. 
The following corollary makes that expectation precise. 

\begin{corollary}[Normalization of the measurements]
Similar statements as in Theorems~\ref{thm:dNormStability} and \ref{thm:robustness} hold also in the case of 4-generic measurements,
provided that $\eta$ is replaced by $\eta /\snorm{A_0}$ and $\TwoNorm{A_0}$ is replaced by the $2$-norm of the traceless part of $A_0$.
\end{corollary}

\myparagraph{Remarks}
\begin{compactenum}[i)]
\item 
	Theorem~\ref{thm:robustness} guarantees approximate reconstructions for any quantum channel $T_0 \in \CPT(\CC^n)$ without requiring an a priori rank constraint. 
	The deviation ${T_0}\c$ from a linear map of low Kraus-type rank, i.e., the model mismatch, enters the error bound linearly. 
	Note that we measure the model mismatch w.r.t.\ all rank-$r$ linear maps, rather than w.r.t.\ only the rank-$r$ quantum channels. This gives a more favorable error bound. 
	Such a robustness is desirable for practical applications where the model assumption of low Kraus rank is typically only approximately true. 

\item 
	For the constrained norm minimizations~\eqref{eq:CdNormRec} and~\eqref{eq:CTrNormRec}, a prior error threshold is required. 
	This additional model selection task is important in actual applications and can be a non-trivial task; see 
	c.f., Ref.\ \cite{SteRioCut16}.
	The CPT-fit \eqref{eq:CPTRec} has the distinct advantage that the reconstruction can be done without such a prior noise estimation.

\item 
  Theorems~\ref{thm:TrNorm} and~\ref{thm:CPTfit} quantify the reconstruction error in terms of Schatten $p$-norms. 
  One may obtain diamond norm estimates from these results via the well-known inequalities between the diamond norm and the Schatten norms (c.f.\ e.g.\ Ref.~\cite{MicKliKue18}).
  A direct derivation of diamond norm bounds does not seem to be achievable with the proof methods we employ, which rely heavily on properties of the various Schatten norms.

\item 
  We emphasize that the dimensional scaling of the error bound in Eq~\eqref{eq:FitError} is a consequence of the normalization we adopted. 
  To make this precise, it is useful to view these statements as results about stably reconstructing (particular) Choi matrices $J(T) \in \L(\CC^n \otimes \CC^n)\cong \L ( \CC^{n^2})$. Rewriting the single expectation values \eqref{eq:MeasurementModel} as
\begin{equation}\label{eq:MeasurementModelMatrix}
y_i = \Tr \Bigl[ A_i \otimes \left( \ketbra{\psi_i}{\psi_i} \right)^T J (T) \Bigr] + e_i 
\end{equation}
reveals that all individual measurement matrices $M_i = A_i \otimes \left( \ketbra {\psi_i} {\psi_i} \right)^T \in \L ( \CC^{n^2})$ have constant Frobenius norm
\begin{equation}
\TwoNorm{M_i} = 
\TwoNorm{ A_i \otimes \left( \ketbra{\psi_i}{\psi_i} \right)^T } = \TwoNorm{A_0} \, .
\end{equation}
Choosing a particular normalization influences the signal-to-noise ratio (SNR) in \eqref{eq:MeasurementModelMatrix} and, in turn, the stability assertions. 
In order to avoid these ambiguities, it is useful to rewrite Eq.~\eqref{eq:FitError} for vanishing model mismatch as
\begin{equation} \label{eq:stability_renormalized}
\TwoNorm{ J(T_0-T^{\ast c}_\eta) }
\leq 
\tilde{c} \frac{\sqrt{m} \, d }{\sum_{i=1}^m \TwoNorm{M_i} } \| e \|_{\ell_2}
\end{equation}
with $d \coloneqq \dim\bigl( \CC^{n^2} \bigr)$.
Note that the sum scales like $m$. 
Up to our knowledge, all stable low-rank matrix reconstruction guarantees are essentially\footnote{up to log-factors} of this form!
This in particular includes Gaussian measurement ensembles \cite[Theorem 2.3]{CanPla11}, random Pauli matrix measurements \cite[Proposition 2.3]{Liu11} and outer products of standard Gaussian vectors \cite[Theorem 2]{KueRauTer15}.

\item Regarding sample complexity, an order of $r n^5 \ln (n)/\epsilon^2$ independent channel evaluations (``samples'') are required to achieve a reconstruction error of at most $\epsilon$ in Frobenius norm; see Section~\ref{sec:SampleComplexity} below.
We expect this scaling to be close to optimal up to $\ln (n)$-factors (at least for channels with very low Kraus rank $r$).
Evidence for this is provided by the discussion above: the stability guarantees in Theorem~\ref{thm:dNormStability} and Theorem~\ref{thm:robustness} essentially match the best existing results on stable low-rank matrix reconstruction.
Among these are two that are applicable to the related problem of quantum state tomography: random Pauli measurements \cite{Liu11} and outer products of standard Gaussian vectors \cite{KueRauTer15}.

The sample complexity of the former approach has been determined in Ref.\ \cite{FlaGroLiu12}. Moreover, it has been shown to be close to optimal in the sense that it reproduces a fundamental lower bound -- valid for any tomographic procedure based on Pauli measurements -- up to a single $\ln (n)$-factor.

Fundamental lower bounds on the sample complexity achievable by any tomographic procedure have been derived in Ref.\ \cite{HaaHarJi15}. Said work also determines the sample complexity associated with measuring outer products of standard Gaussian vectors.
Interestingly, this sample complexity matches the fundamental lower bound up to a factor of $r \ln (n)$, where $r$ denotes the rank of the density operator in question. Thus, state tomography via low-rank matrix recovery from outer products of Gaussian vectors is close to optimal, at least for states that have very low rank.

We expect that similar results are true for quantum process tomography via low-rank matrix reconstruction of Choi matrices. 
However, while conceptually similar, the results regarding sample complexities of state tomography procedures \cite{FlaGroLiu12,HaaHarJi15} are not directly applicable to process tomography. 
We intend to address such an extension in future work.
\end{compactenum}

\subsubsection{Numerical example: Reconstructing the Toffoli gate}
\label{sec:ToffoliNumerics}
The Toffoli gate is a three qubit gate that has drawn a lot of attention from theorists as well as experimentalists. 
It is universal for classical computation 
and, together with the Hadamard gate, also for quantum computation 
\cite{Shi02}. 
Moreover, it has played an important role in the theory of gate sets \cite{BarBenCle95} and
can help to significantly reduce the number of gates in quantum algorithms, 
such as in quantum error correction. 
It also has been implemented experimentally in 
nuclear magnetic resonance \cite{CorPriMaa98},
linear optics \cite{LanBarAlm09},
trapped ions \cite{MonKimHan09}, and in 
superconducting circuits \cite{FedSteBau12}.
Moreover, the reconstruction using the CPT-fit \eqref{eq:CPTRec} has already been compared to full tomography \cite{RodVeiBar14}. 
Hence, the Toffoli gate is a good candidate to benchmark our process tomography schemes.

We demonstrate uniform recovery, i.e., 
we draw a fixed number of measurement settings at random and keep them fixed throughout the simulation in Figure~\ref{fig:noisy}. 
Then, we always use the first $m$ of them for reconstructions with $m$ settings. 

Since reconstruction via the unconstrained trace norm minimization \eqref{eq:TrNormRec} is numerically inferior to the other reconstructions (Figure~\ref{fig:Noiseless}) we will not investigate it in all simulations. 

There are two types of error sources: 
(i) imperfect measurements give rise to measurement noise $e \in \RR^m$ in our model \eqref{eq:MeasurementModel} and
(ii) the implemented channel $T_0 \in \CPT(\CC^n)$ could have violated the model assumption of a low Kraus rank.
Here, we confirm our analytic stability result towards both error sources numerically.

Our theorems put minimal assumptions on the potential noise corruption $e\in \RR^m$. 
In particular, it does not need to follow a specific statistical model. 
For our numerical analysis, however, we draw the measurement noise $e$ i.i.d.\ from a zero-mean Gaussian distribution and scale it so that the noise strength $\lTwoNorm{e}$ has the desired value (Figure~\ref{fig:noisy}, left). 
A Gaussian error model frequently occurs in practice:
When estimating expectation values from observed frequencies such an error model arises naturally in the limit of many measurements per setting. 
In practice, 
the parameter $\eta$ needs to be estimated for the reconstructions \eqref{eq:TrNormRec}, \eqref{eq:CTrNormRec}, \eqref{eq:CdNormRec} and \eqref{eq:CdNormRec}. 
Here, one could take the smallest $\eta$ so that the reconstructions succeeds. 
The CPT-fit \eqref{eq:CPTRec} has the advantage of not requiring such an estimation.

We reconstruct $T_\Toff$ from $m=320$ \emph{noisy} measurements with different values of $\lTwoNorm{e}$ without model mismatch ($\lambda = 0$), see Figure~\ref{fig:noisy}(left). 
For the trace and diamond norm reconstructions~\eqref{eq:CTrNormRec}, \eqref{eq:dNormRec}, and~\eqref{eq:CdNormRec} we set the error parameter $\eta = \lTwoNorm{e} + 10\, \mathrm{eps}$, where 
$\mathrm{eps}=2^{-52}\approx 2 \cdot 10^{-16}$ is the machine precision. 
As predicted by Theorems~\ref{thm:dNormStability} and~\ref{thm:robustness}, 
the reconstruction error $\TwoNorm{J(T^\rec - T_0)}$ 
scales linearly in the noise strength $\lTwoNorm{e}$, see Figure~\ref{sfig:noisy:b}. 
Here, the CPT-fit \eqref{eq:CPTRec} has the smallest reconstruction error. 
If one further increases the number of measurements, then the reconstruction error $\TwoNorm{J(T^\rec - T_0)}$ decreases further, as guaranteed by Theorems~\ref{thm:dNormStability} and~\ref{thm:robustness}.

In order to demonstrate robustness of our reconstructions, we set $T_0$ to be a convex combination of the Toffoli gate $T_\Toff$ and a completely depolarizing channel $T_\dep (\rho) = n^{-1}{\Tr (\rho)}\1$, 
\begin{equation}\label{eq:toffoli_depolarizing}
T_0 = (1-\lambda) T_\Toff + \lambda\, T_\dep \, , 
\end{equation}
where $\lambda \in [0,1]$.
The depolarizing channel corresponds to a physically relevant error model that maximally violates our model assumption of low Kraus rank; see, e.g., Ref.\  \cite[Chapter 8.3.4]{NieChu10}. 

We test the reconstruction of $T_0$ for different values of $\lambda \in [0,1]$.
For the sake of clarity, we completely suppress additive noise ($e=0$) and set the error threshold $\eta=10\, \eps$.
The results are presented in Figure~\ref{fig:noisy}(right) and demonstrate the robustness guaranteed by Theorem~\ref{thm:robustness}. 
The reconstruction error $\TwoNorm{J(T^\rec - T_0)}$ 
depends roughly linearly $\lambda$. 
Here, the diamond norm minimizations \eqref{eq:dNormRec} and \eqref{eq:CdNormRec}, perform worse than the constrained trace norm minimization \eqref{eq:CTrNormRec} and the CPT-fit \eqref{eq:CPTRec}.

In the case of measurement noise, the reconstruction error $\TwoNorm{J(T^\rec - T_0)}$ approaches the optimal value $\lTwoNorm{e}$ in the limit of large $m$, see Figure~\ref{sfig:noisy:a}(left).
In the case of a model mismatch, the reconstruction error decreases below the mismatch parameter $\lambda$ and vanishes if $m$ approaches the full dimension of $\CPT(\CC^n)$, see Figure~\ref{sfig:noisy:a}(right). 

We find it worthwhile to share one interesting numerical observation on the unconstrained diamond norm reconstruction~\eqref{eq:dNormRec}. 
For this purpose we assume that $T_0$ is a quantum channel, i.e. $T_0 \in \CPT(\CC^n)$, so that $\dnorm{T_0}=1$. 
Then the optimal value $\dnorm{T^\rec}$ of the reconstruction seems to decrease with the reconstruction error 
$\TwoNorm{J(T^\rec - T_0)}$, see Figure~\ref{sfig:noisy:c}. 
Hence, this reconstruction procedure does not only yield good approximations to the measured channel $T_0$, but its optimal value also provides some indication of the reconstruction error's size.

One can exploit this observation by using both the CPT-fit \eqref{eq:CPTRec} and the unconstrained diamond norm reconstruction \eqref{eq:dNormRec}.
The CPT-fit is the fastest reconstruction procedure and yields the smallest error. Complementing this, the diamond norm minimization provides an indication of what the error might be. 
In the CPT-fit, the reconstructed map $T^\rec$ is always a quantum channel, i.e., $T^\rec \in \CPT(\CC^n)$. 
In contrast, the solution of the diamond norm minimization can, in principle, be any map in $\LL(\CC^n)$. 
The latter also holds true for the unconstrained trace norm reconstruction \eqref{eq:TrNormRec}, but we could not observe a similar feature of its minimum value.

The error bounds from Theorem~\ref{thm:robustness} suggest that observables with larger Frobenius norm have a better noise suppression, as $\TwoNorm{A_0}$ appears in the denominator in the error bound~\eqref{eq:FitError}. 
We also tested this behaviour numerically in order to demonstrate that it is not just a proof artifact, but an actual feature.
We choose $A_0$ to have $r_A$ many non-zero eigenvalues, which we evenly distribute in the interval $[-1,1]$. 
These $A_0$ have a Frobenius norm in the interval $[1,2]$. 
Figure~\ref{fig:rankA} shows the average reconstruction error of the Toffoli gate for non-uniform measurements in dependence of $\TwoNorm{A_0}$ and noise strength $\lTwoNorm{e}=0.1$. 
This numerical analysis demonstrates that the reconstruction error can indeed be reduced with increasing $\TwoNorm{A_0}$. 

\clearpage
\newgeometry{margin = 2.2cm}
	\hspace{-.8cm}
	\begin{minipage}{\linewidth}
	\captionsetup{type=figure}
	\def\mywidth{.5}
	\small
	\stepcounter{figure}
	\begin{minipage}{\linewidth}
	\begin{tabular}{p{\mywidth\linewidth} p{\mywidth\linewidth}}
		\hspace{-.5cm}
		\includegraphics[width = 1.1\linewidth]{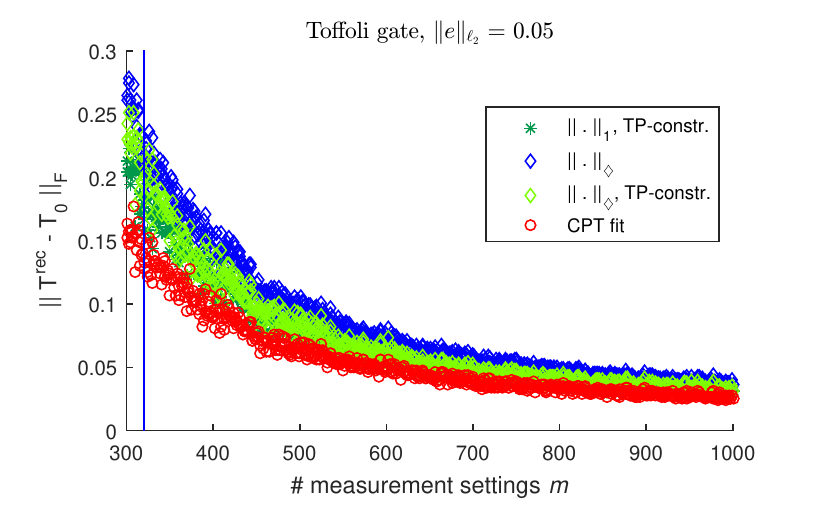}
		&
		\hspace{-.5cm}
		\includegraphics[width = 1.1\linewidth]{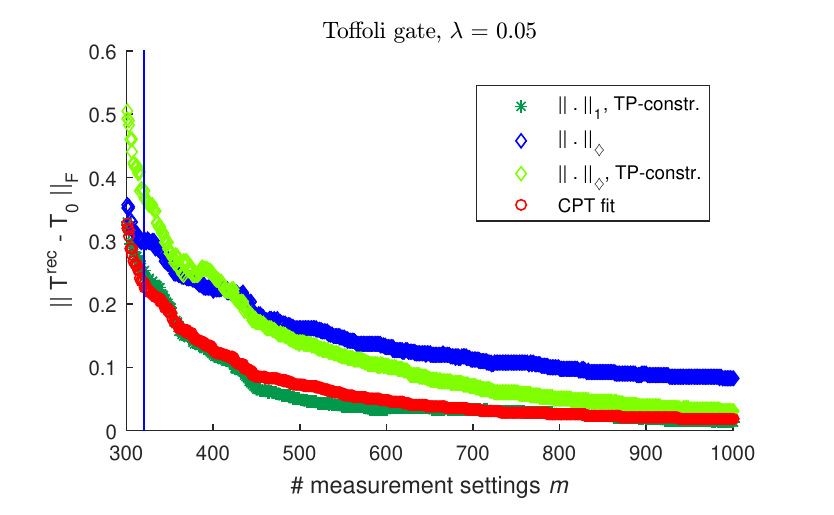}		
	\end{tabular}
	\\[-.3em]
	\mySubfloatCaption{%
		\label{sfig:noisy:a}
		The reconstruction error over the number of measurement settings $m$ for fixed measurement noise $\lTwoNorm{e} = 0.05$ (left) and fixed model mismatch $\lambda = 0.05$ (right), respectively, in a uniform recovery setting. 
		}
	\\[.8em]
	\begin{tabular}{p{\mywidth\linewidth} p{\mywidth\linewidth}}
		\hspace{-.5cm}
		\includegraphics[width = 1.1\linewidth]{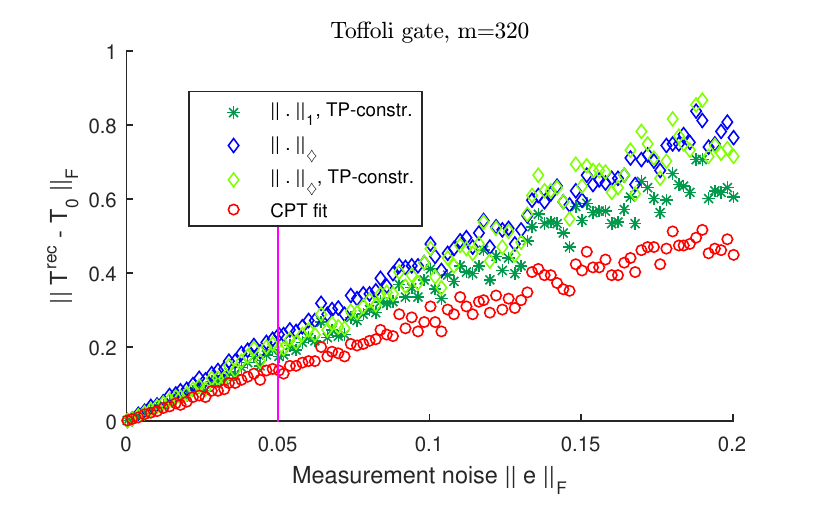}
		&
		\hspace{-.5cm}
		\includegraphics[width = 1.1\linewidth]{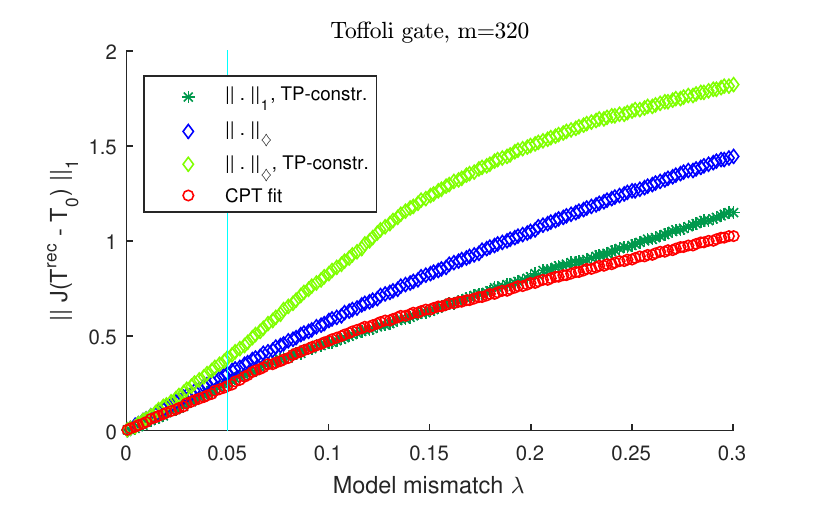}
	\end{tabular}
	\\[-.3em]
	\mySubfloatCaption{%
		\label{sfig:noisy:b}
		The reconstruction error over the measurement noise $\lTwoNorm{e}$ and model mismatch $\lambda$, respectively, for $m=320$ fixed measurement settings. 
		}
	\\[.8em]
	\noindent\hspace{-.05\linewidth}
	\begin{tabular}{p{\mywidth\linewidth} p{\mywidth\linewidth}}
		\hspace{-.5cm}
		\includegraphics[width = 1.1\linewidth]{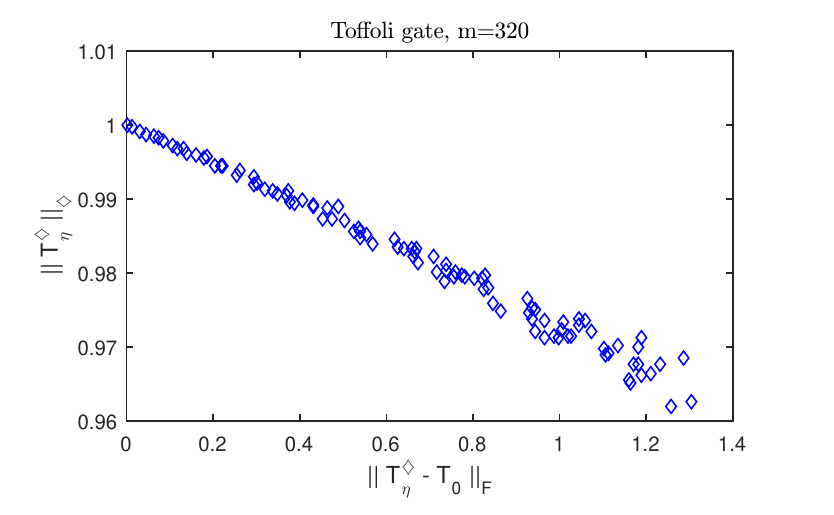}		
		&
		\hspace{-.5cm}
		\includegraphics[width = 1.1\linewidth]{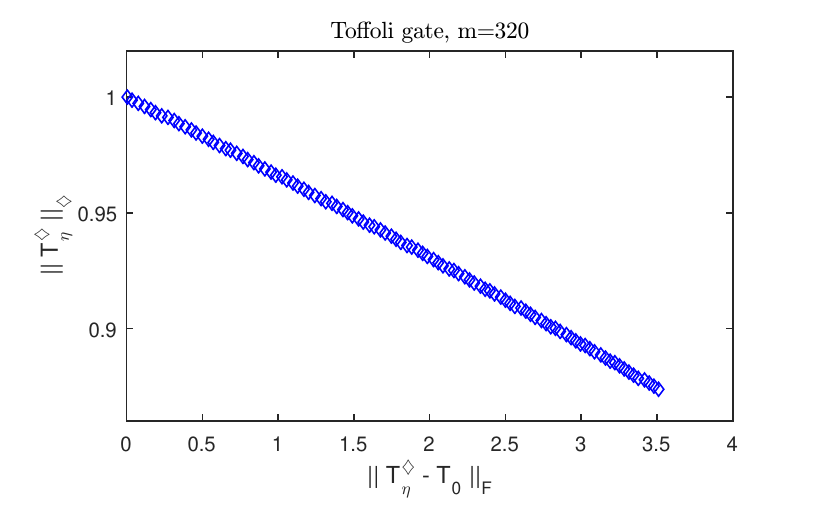}		
	\end{tabular}
	\\[-.3em]
	\mySubfloatCaption{%
		\label{sfig:noisy:c}%
 		The optimal value $\dnorm{T^\diamond_\eta}$ from the minimization \eqref{eq:dNormRec} over the reconstruction error $\TwoNorm{J(T^\diamond_\eta-T_0)}$ achieved by the unconstrained diamond norm minimization \eqref{eq:dNormRec} used in (b).
		}
	\\[.8em]
	\end{minipage}
	\\
	\addtocounter{figure}{-1}
	\normalsize
	\captionof{figure}{%
		\label{fig:noisy}
		Uniform recovery of the three qubit Toffoli gate $T_0 \in \CPT(\CC^8)$ in imperfect settings. 
		In the perfect setting $m=320$ measurement settings are sufficient for reconstruction w.h.p.\ while the total dimension is $\dim(\CPT(\CC^8)) = 4032$. 
		\newline
		\capstr{Left}: 
		Reconstruction of $T_0=T_\Toff$ from $y = \A(T_0)+e$ and $\A$ with measurement noise $e\in \RR^m$ being drawn uniformly from a scaled sphere and without model mismatch ($\lambda =0$). 
		The parameter $\eta$ is chosen to be $10$ times machine precision plus the chosen noise strength~$\lTwoNorm{e}$. 
		\newline
		\capstr{Right}:
		Reconstruction of $T_0$ from $y = \A(T)$ and $\A$ with
		$T_0 = (1-\lambda)\,  T_{\Toff} + \lambda \, T_{\mathrm{dep}}$, where the model mismatch 
		is caused by the completely depolarizing channel $T_{\mathrm{dep}}$.
	}
	\end{minipage}
\restoregeometry
\clearpage

\begin{figure}
	\centering
	\includegraphics[width = .65\linewidth]{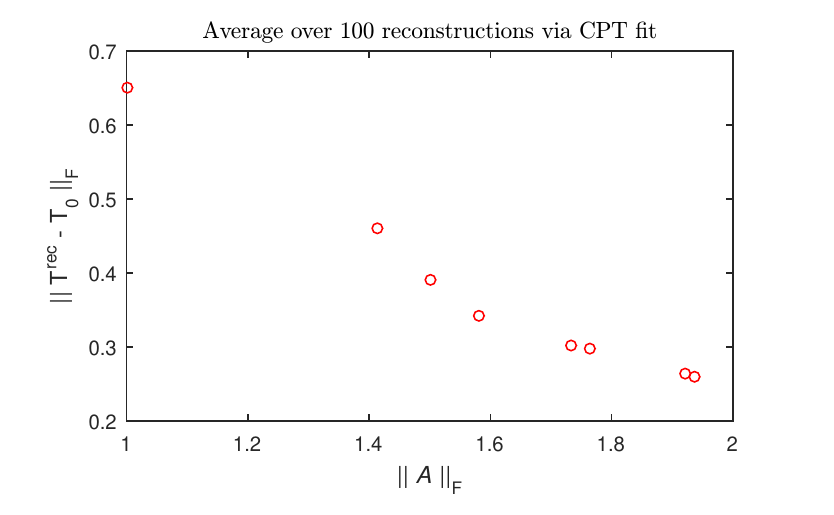}%
	\caption{The plot shows the average reconstruction error $\TrNorm{J(T^{\ell_2} - T_0)}$ for observables with different rank. 
	The observables' non-zero eigenvalues cover evenly the interval $[-1,1]$, giving rise to a Frobenius norm increasing monotonically with the rank. 
	Plotted are CPT-fit \eqref{eq:CPTRec} reconstructions of the Toffoli gate from $m=320$ i.i.d.\ measurements. 
	The noise is scaled to $\lTwoNorm{e} = 0.1$. 
	\newline 
	The plot highlights advantageous stability properties of non-degenerate -- and thus high rank -- observables. 
	}\label{fig:rankA}
\end{figure}

\subsection{Pauli measurements}
\label{sec:OtherMeasurements}
Our recovery guarantees hold for 4-generic measurements. 
However, these measurements can be challenging to implement in many experimental situations. 
So, how do our recovery schemes perform for more more restricted measurement ensembles? 
In this section we numerically investigate two practically relevant measurement scenario of Pauli-type measurements. 

For quantum state tomography \cite{FlaGroLiu12}, 
Pauli measurements are proven to satisfy the so-called \emph{restricted isometry property} for rank-$r$ matrices in $\L(\CC^d)$ for a number of measurement settings scaling as 
$r\, d\, \log(d)^6$ \cite{Liu11}. 

Motivated by this strong statement, we numerically investigate process measurements that are inspired by a Pauli setting.
We denote the set of Pauli strings by $\mc P \subset \U(2^L)$, i.e., the set of operators  
$P = \sigma^{(1)}\otimes \sigma^{(2)} \otimes \dots \otimes \sigma^{(L)}$ with Pauli matrices
$\sigma^{(i)} \in \{\1, \sigma_x, \sigma_y, \sigma_z\}$ for $i \in [L]$. 
We write $P \sim \mc P$ for a Pauli string that is drawn uniformly at random from $\mc P$. 
Then we choose the measurements as 
\begin{equation} \label{eq:PauliMeasurement}
	y_j \coloneqq \Tr[P_j T_0(\ketbra{\psi_j}{\psi_j})] \, ,
\end{equation}
where observables $P_j$ and input states $\ket{\psi_j}$ are i.i.d.\ selected as follows. 
Each $P_j \sim \mc P$ is a uniformly drawn Pauli string 
and each state vector $\ket \psi_j$ is a tensor product of uniformly i.i.d.\ drawn eigenvectors of random Pauli operators $\{\sigma_x, \sigma_y, \sigma_z\}$ (hence, an eigenvector of the corresponding Pauli string). 

Numerically, we observe that for random unitary quantum channels our reconstructions perform very similar for these Pauli measurements and the generic (Haar-random) measurements (not shown in the plots). 
However, for non-generic channels we observe that the two types of measurements lead to different reconstruction behaviours, see Figure~\ref{fig:Pauli}:
For the reconstruction of the Toffoli gate, more Pauli-type-measurements than generic measurements are required in the case of unconstrained trace norm regularization~\eqref{eq:TrNormRec}. 
In contrast, fewer Pauli-measurements are required for the other regularizations. 

We have also observed that reconstructions from Pauli-measurements have similar stability properties as the ones in the generic case (not shown in the plots). 

The identity quantum channel $T_0=\id$ is an extreme case in the sense that it is sparse in any basis and commutes with all transformations. 
For generic measurements we have observed that it has the same reconstruction behaviour as Haar-random unitary $T$ channels. 
In comparison, in case of Pauli measurements, the unconstrained trace norm reconstruction works worse and the other reconstructions better (not shown in the plots). 

\begin{figure}
\newgeometry{margin = 2cm}
	\def\mywidth{.5}
	\hspace{-2cm}
	\begin{tabular}{p{\mywidth\linewidth} p{\mywidth\linewidth}}
		\includegraphics[width = 1.1\linewidth]{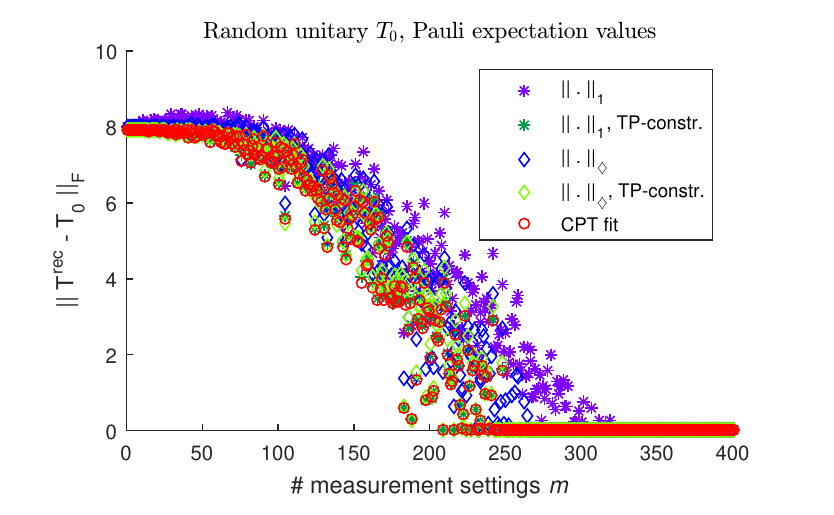}
		&  
		\includegraphics[width = 1.1\linewidth]{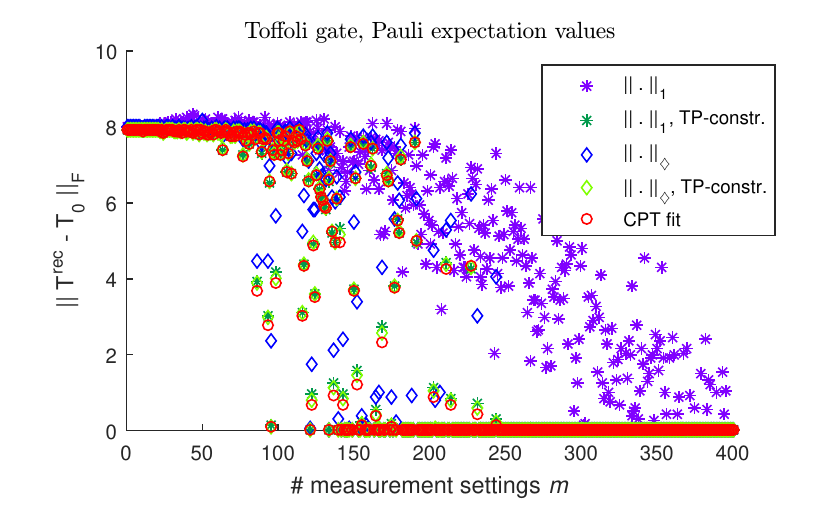}
	\end{tabular} 
\restoregeometry
	\caption{
		Pauli-type measurements: 
		Comparison of the reconstruction errors of generic (Haar-random) unitary quantum channels (left) and the Toffoli gate (right). 
		}\label{fig:Pauli}
\end{figure}

\subsection{Sample complexity}
\label{sec:SampleComplexity}
The expectation values in our measurements need to be estimated from finite samples leading to a reconstruction error $e$. 
Assume we want to estimate the measured channel $T_0 \in \CPT(\CC^n)$ in Frobenius norm up to an error 
$\lTwoNorm{e} \leq \epsilon$. 
Then the \emph{sample complexity} is the scaling of the optimal number of measurements in the ideal setting, i.e., without model mismatch or measurement errors. 

We use the Landau symbols $\LandauO$ and $\LandauOmega$ to denote the usual asymptotic upper and lower bounds, respectively. 
The Landau symbol $\tLandauO$ denotes the same scaling as $\LandauO$ up to log-factors. 

The error of an empirical mean of  
$\Tr[A_i T_0(\ketbra{\psi_i}{\psi_i})]$ form $\ell$ samples scales as $\LandauO(1/\sqrt{\ell})$ w.h.p.\ 
(due to, e.g., Chebyshev's inequality). 
This gives rise to an error vector $e$ bounded as 
\begin{equation}
\lTwoNorm{e}\in \LandauO(\sqrt{m/\ell}) \, .
\end{equation}
Hence, the total estimation error bounded as 
\begin{equation}
\epsilon \leq \frac{n^2\lTwoNorm{e}}{\sqrt{m} \TwoNormb{A_0}}
\end{equation} 
is small w.h.p.\ for $\ell\geq \ell_0$ with some $\ell_0$ being bounded as 
$\ell_0 \in \LandauO\Bigl(\frac{n^4}{\TwoNormn{ A_0}^2} \epsilon^2 \Bigr)$. 
For 
$\TwoNormb{A_0} \in \LandauOmega(\sqrt n)$ 
this yields a total sample complexity scaling as 
$\tLandauO(r n^2) \LandauO(n^3/\epsilon^2) = \tLandauO(rn^5/\epsilon^2)$.

\subsection{Applications to fault tolerant quantum computation}
The techniques presented here also have applications for fault tolerant quantum computation. 
Threshold theorems \cite{Kit97,KniLafZur98,AhaBen97} give a theoretical guarantee that quantum computers can be built in principle if the noise strength is below some threshold value. 
The strength of the errors needs to be quantified in diamond norm distance which is not directly accessible. Instead one typically evaluates average error rates using direct fidelity estimation \cite{FlaLiu11}, or randomized benchmarking, see, e.g., Ref.~\cite{EmeAliZyc05}.
However, these two error measures can differ by orders of magnitude \cite{WalFla14}. This, in particular, is the case for coherent error sources, such as unitary over/under-rotations \cite{KueLonDoh15}, where the diamond distance is proportional to the square root of the average error rate.
Thus, achieving fault-tolerance thresholds in the presence of unitary noise requires an exceedingly high error control (around $10^{-8}$ for typical threshold levels of a few times $10^{-4}$).
In contrast, other noise sources typically imply a much more favourable (linear) relation between both error measures.

From a practical perspective, these results are encouraging: Unlike incoherent noise, coherent noise effects can typically be corrected. 
A necessary subroutine for achieving this goal is to accurately estimate these error channels.
Our results considerably simplify this task, in particular for unitary errors which have Kraus rank one. They provide estimation techniques that require considerably fewer measurements than traditional process tomography protocols.

\section{Analytical details and proofs}
Our analytical results build on by now well-established mathematical proof techniques for low-rank matrix reconstruction \cite{Tro15,KueRauTer15,KabKueRau15}. 
The main technical contribution of this work is to extend these techniques to natural  measurements whose structure deviates from less structured measurement matrices common in low-rank matrix reconstruction.

Starting point of our analysis is a geometric approach to low-rank matrix recovery presented in Ref.\ \cite{Tro15}. It relates the reconstruction error from a constrained trace norm minimization to 
a certain quantity associated with the measurement map $\A$: the \emph{minimum conic singular value}; see Definition~\ref{def:conic_SV} below.
We bound this quantity by invoking Mendelson's small ball method \cite{Men14,KolMen14} -- a strong probabilistic tool that depends on certain concentration properties of the measurement ensemble. 
For 4-generic measurements these are derived using representation theory of the unitary group and general bounds on tensor network contractions besides probabilistic bounds commonly used on low-rank matrix reconstruction. 

This geometric analysis results in a reconstruction guarantee for the constrained trace norm minimization \eqref{eq:CTrNormRec} that is stable towards additive noise corruption. 
In turn, the geometric arguments provided in Ref.~\cite{KliKueEis16} suggest a strengthening of the obtained error bounds, if one replaces the trace norm by the diamond norm. 
Theorem~\ref{thm:dNormStability} -- or, more generally: Theorem~\ref{thm:TrNorm2} and Corollary~\ref{thm:dNorm} -- are consequences of such an approach. 

Our second main result -- Theorem~\ref{thm:robustness} -- assures stability towards noise corruption as well as robustness with respect to violating the model assumption of low Kraus rank.
This further improvement is achieved via establishing a robust version of the \emph{null space property} for 4-generic measurements. We refer to Section~\ref{sub:nsp} for a brief introduction.
The technical prerequisites for such an approach are largely identical to the ones associated with the more geometric framework outlined above. Mendelson's small ball method, in particular, is again applicable. 
Hence, relatively little additional effort is required for this approach which has the added benefit of ensuring robustness towards model mismatches. 

The remainder of this section is organized as follows:
After some preliminaries, we prove a bound on the minimum conic singular value for the case of our measurement setting in Section~\ref{sec:OurLambdaMin}. 
The proof used a general bound on tensor networks, which we state and prove in Section~\ref{sec:TNbound}. 
In Section~\ref{sec:ProofsThms}, we state and prove general versions of our main theorems. 
Finally, in Section~\ref{sec:LinarOptics}, we show that our results also hold for a quantum linear optical setting.

\subsection{Preliminaries}
Before we come the to proofs we 
introduce some helpful notation, 
discuss the minimum conic singular value and a useful bound to it, 
introduce the null space property and a subsequent recovery guarantee, and
explain the use of representation theory for bounding certain moments.

\subsubsection{Notation}
A \emph{vectorization} of a tensor is a vector containing all its elements. 
The \emph{Frobenius norm} $\fnorm{T}$ of a tensor $T$ is defined to be the $\ell_2$-norm of some vectorization of $T$. 
Importantly, for an operator $A \in \L(\CC^n)$ it holds that $\fnorm{A}=\TwoNorm{A}$. 

The permutation group on $k$ elements is denoted by $S_k$ and the unitary group by $\U(n) \subset \L(\CC^n)$. 
Its representation on $(\CC^n)^{\otimes k}$ is denoted by $R$, so that for $\sigma \in S_k$ the representation $R(\sigma) \in U(n^k)$ acts on $(\CC^n)^{\otimes k}$ by permuting the $k$ tensor copies according to $\sigma$. 

\myparagraph{Affine spaces}
In order to deal with the constraint that the reconstructed maps are trace preserving we need the affine version of the usual reconstruction problem \cite{Tro15} in compressed sensing. 

It is helpful to extend the basic algebraic operations to sets:
E.g., we denote the Minkowski sum of subsets $S$ and $R$ of some vector space by 
$S+R \coloneqq \{s+r: s\in S, r\in R\}$. 
Similarly, we define $-R$ and $s+R$ for some $s \in S$. 

An \emph{affine space} $\V$ over the field $\RR$ is a subset of a vector space over $\RR$ such that for all $\lambda \in \RR$ and $x,y\in \V$ one has $(1-\lambda) x + \lambda y\in \V$. 
Then $\V_0 \coloneqq \V-\V$ is a vector space and $\V = x+\V_0$ for any $x \in \V$. 
The linear span $\lin (\V)$ of $\V$ is another vector space containing both $\V$ and $\V_0$.

A map $A: \V \to \W$ between affine spaces $\V$ and $\W$ is called \emph{affine} if for all $\lambda\in \RR$ and $x,y\in \V$ one has $A((1-\lambda) x + \lambda y) = (1-\lambda) Ax + \lambda Ay$.

Given an affine map $\A : \V \to \RR^m$ (such as the measurement map defined in \eqref{eq:MeasurementMap}) one can extend it linearly to $\lin (\V)$. 
This extension will also be denoted by $\A$.

\myparagraph{Maps and operators}
Any map $M \in \LL(\CC^n)$ and operators $A, \rho \in \L(\CC^n)$ satisfy the identity
\begin{equation}\label{eq:ChoiMeasurement}
	\Tr[A\, M(\rho)] 
	= 
	\Tr[A \otimes \rho^T J(M)] \, .
\end{equation}
We will use this identity to write expectation values of the type $\Tr[A T(\rho)]$  in terms of the Choi matrix $J(T)$.  

By $M\ad$ we denote the Hilbert-Schmidt adjoint of $M \in \LL(\CC^n)$ and by $
M^{\star}$ the map which obeys $M(X)^\dagger = M^\star(X^\dagger)$ for all $X \in \L(\CC^n)$. 
These two involutions satisfy $M^{\dagger\star}=M^{\star\dagger}$. 
Moreover, we use the notation
\begin{equation}
	M^{k,l} \coloneqq M^{\otimes k} \otimes M^{\star \otimes l} .
\end{equation}

We note in passing that CPT maps $T \in \CPT(\CC^n)$ have diamond norm $\dnorm{T}=1$ and Choi matrices are normalized as $\TrNorm{J(T)} = n$. 

It will also be helpful to define the set containing all Choi matrices
\begin{equation}
	J\HT(\CC^n) = \{ X \in \Herm(\CC^n \otimes \CC^n): \ \Tr_1(X) = n\,\1_n\} \, ,
\end{equation}
where $\Tr_1$ denotes the partial trace over the first tensor factor. 

\subsubsection{Normalization and centralization of the observables}
We assume the fixed observable $A_0$ to be normalized ($\snorm{A_0}=1$) and centered (traceless): $\Tr [A_0]=0$.
Such restrictions somewhat simplify the technical analysis; the following observations highlight that little generalization is lost by imposing them:

\begin{observation}[Uncentered observables]\label{obs:centering}
Fix $A_0 \in \mathrm{Herm}(\CC^n)$ and let
	\begin{equation}
		\tilde A_0 \coloneqq A_0 - \Tr[A_0] \1 / n
	\end{equation}
	be the traceless part of $A_0$. Then, the associated 4-generic measurement maps $\A$ and $\tilde{\A}$ introduced in Definition~\ref{def:measurement_ensemble} do not necessarily coincide.
However, the algorithmic reconstructions $T^\rec$ and $\tilde{T}^\rec$ of any trace-preserving map $T$ nonetheless coincide for all reconstruction procedures \eqref{eq:TrNormRec}, \eqref{eq:CTrNormRec}, \eqref{eq:dNormRec}, \eqref{eq:CdNormRec}, or~\eqref{eq:CPTRec}.
\end{observation}

\begin{proof} 
Since any $T\in \CPT(\CC^n)$ is trace preserving, we have that
\begin{align}\nonumber
\tilde \A(T)_j
&= \Tr[ \tilde A_j T(\ketbra{\psi_j}{\psi_j}) ]
\\
&= 
\Tr[ A_j T(\ketbra{\psi_j}{\psi_j}) ] - \Tr[A_j]/n 
\\ \nonumber 
&= \A(T)_j - \Tr[A_j]/n . 
\end{align}
Together with the definition of $y$ and $\tilde y$ this implies
\begin{equation}
\A(T) - y  
= 
\tilde \A(T) - \tilde y \, .
\end{equation}
But this means that the corresponding minimizations are the same. 
\end{proof}

The following observation says that the observable's spectral norm suppresses the noise. 

\begin{observation}[Unnormalized observables]\label{obs:normalizing}
Let $\A$ be a 4-generic measurement and define $\hat{\A}:=\snorm{A_0}^{-1} \A$. Then the algorithmic reconstructions associated with measurements
\begin{equation}
y = \A (T_0) + e \quad \textrm{and} \quad \hat{y} = \hat{\A}(T_0) + \frac{e}{\snorm{A_0}}
\end{equation}
coincide for all five reconstruction procedures: \eqref{eq:TrNormRec}, \eqref{eq:CTrNormRec}, \eqref{eq:dNormRec}, \eqref{eq:CdNormRec} and~\eqref{eq:CPTRec}.
\end{observation}

\subsubsection{Minimum conic singular value}
The usual minimum conic singular value of a map $\A$ can be written variationally as the right hand side (RHS) of Eq.~\eqref{eq:lambdamin} with $K$ being the full space and $q=2$. 
In order for $\A$ to be invertible, the minimum conic singular value needs to be positive. 
Moreover, for fixed spectral norm of $\A$, the larger the minimum conic singular of $\A$ the more stable $\A$ can be inverted. 
An extension of these basic concepts can be extended to realm of convex analysis, which motivates the following Definition. 
\begin{definition}[$\ell_q$-minimum conic singular value]\label{def:conic_SV}
	Consider an affine space $\V$ where $\lin (\V)$ is equipped with a norm $\norm{\argdot}$. 
	Let $\A : \V \to \RR^m$ be an affine linear map and $K\subset \V_0$ be a cone. 
	The minimum singular value of $\A$ w.r.t.\ $K$, measured in $\ell_q$-norm with $q \geq 1$, is 
	\begin{equation}
		\lambda_{\min{}}(\A;K; \ell_q) 
		\coloneqq
		\inf_{u \in K} \frac{\lqNorm{\A u }}{\norm{u}}  \, , 
\label{eq:lambdamin}
	\end{equation}
	where $\A$ has been extended to $\lin (\V)$. 
\end{definition}

Typically, one chooses $q=2$ in order to define the minimum conic singular value \cite{Tro15}. Here, we opt for a more general definition that allows for adjusting stability guarantees towards noise to different types of noise models.

The conic singular value of our measurement map $\A$ (see Eq.~\eqref{eq:MeasurementModel}) w.r.t.\ a cone $K \subset \LL(\CC^n)$ can be written as
\begin{equation}
	\lambda_{\min{}}(\A;K; \ell_q) 
	= 
	\inf_{M \in E} \left(\sum_{i=1}^m \Tr[A_i M(\ketbra{\psi_i}{\psi_i})]^q \right)^{\frac{1}{q}}
\end{equation}
with $E = \{M \in K: \fnorm{M}=1\}$. 

We use a method established by Mendelson \cite{Men14,KolMen14} in order to bound the minimum conic singular value. 
As suggested in Ref.~\cite[Remark 5.1]{KabKueRau15}, 
we use a generalization of Tropp's version {\cite[Proposition~5.1]{Tro15}}.

\begin{lemma}[Bound on $\lambda_{\min{}}$]
	\label{lem:Mendelson}
	Let $E\subset \HT(\CC^n)$ be a cone of maps and 
	$(\ket{\psi_i}, A_i)_{i\in [m]}$ be an i.i.d.\ measurement settings. 
	Define the \emph{marginal tail function}
	\begin{equation}\label{eq:Q_def}
		Q_{\xi} \coloneqq \inf_{M \in E} \Pr\bigl[ \bigl|\Tr[A_i M(\ketbra{\psi_i}{\psi_i})]\bigr| \geq \xi\bigr]
	\end{equation}
	(the same for all $i$) and \emph{mean empirical width}
	\begin{equation}\label{eq:W_def}
		W_m(E) \coloneqq 
		\Ev \sup_{M \in E} \kw{\sqrt m}\sum_{i=1}^m \epsilon_i \Tr[A_i M(\ketbra{\psi_i}{\psi_i})] \, ,
	\end{equation}
	where $\epsilon_i \in \{-1,1\}$ are independent uniformly random signs. 
	Then, for any $\xi>0$, $\lambda>0$, and $q\geq 1$  
	\begin{equation}\label{eq:mendelson} 
	\inf_{M \in E} \left(\sum_{i=1}^m \left| \Tr[A_i M(\ketbra{\psi_i}{\psi_i})] \right|^q \right)^{\frac{1}{q}} 
	\geq m^{\frac{1}{q}-\frac{1}{2}} \left( \xi \sqrt m Q_{2\xi} -2 W_m - \xi \lambda \right)
	\end{equation}
	with probability at least $1-e^{-\lambda^2/2}$. 
\end{lemma}

\begin{proof}
Following Ref.~\cite[Remark 5.1]{KabKueRau15}, we point out that the proof of Ref.~\cite[Proposition~5.1]{Tro15} actually implies the stronger statement 
\begin{equation}
	\frac{1}{\sqrt{m}}\inf_{M \in E}  \sum_{i=1}^m \left| \Tr[A_i M(\ketbra{\psi_i}{\psi_i})] \right| 
	\geq \xi \sqrt m Q_{2\xi} -2 W_m - \xi \lambda \, .
\end{equation}
Using that 
$\norm{v}_{\ell_1} \leq m^{1/q} \lqNorm{v}$ 
for any $v \in \CC^m$  results in
\eqref{eq:mendelson} 
for any $q \geq 1$. 
\end{proof}

\subsubsection{Null space property}
\label{sub:nsp}
If an operator $X$ is of rank $r$ one has 
$\TrNorm{X} \leq \sqrt{r} \TwoNorm{X}$. 
Here, we rely on the idea to take a similar inequality to define the notion of effectively rank-$r$ elements. 
This notion of effective low rank is often enough for low-rank matrix reconstruction. 
Additionally, one can take into account violations of $X$ being of low rank. 
This idea is formalized with the \emph{null space property} (NSP). 
We use a version of Kabanava et al.~\cite{KabKueRau15} adjusted to our setting of subspaces, 
which uses the rank-$r$ approximation $X\r$ to $X$ from Eq.~\eqref{eq:tail_def}. 

\begin{definition}[Robust NSP for subspaces]\label{def:NSP}
	Fix a subspace $\mc V_0 \subset \Herm(\CC^d)$,
	$r\in \ZZ_+$, and $q \geq 1$. 
	We say that a linear map $\A : \Herm(\CC^d) \to \RR^m$ 
	satisfies the \emph{($\mc V_0$, $r$, $\ell_q$)-NSP with parameters
	$\mu \in \left]0,1\right[$ and $\tau>0$}
	if for all $X \in \mc V_0$ 
	\begin{equation}\label{eq:NSP}
		\TwoNormb{X\r} \leq \frac{\mu}{\sqrt{r}} \tnorm{X\c} +\tau \norm{\A(X)}_{\ell_q}.
	\end{equation}
\end{definition}

We will use the following result from Ref.~\cite{KabKueRau15} in our setting. 
It yields recovery guarantees based on the NSP. 

\begin{theorem}
	\label{prop:NSP_recovery_guarantee}
	Let $\mc V \subset \Herm(\CC^d)$ be an affine space and set $\mc V_0 \coloneqq \mc V-\mc V$ and
	fix $p \in [1,2]$. 
	If a linear map $\mc A: \Herm(\CC^d) \to \RR^m$ satisfies a ($\mc V_0$, $r$, $\ell_q$)-NSP with parameters
	$\mu \in \left]0,1\right[$ and $\tau>0$ 
	then 
	\begin{equation}
	\pNorm{X-Y}
	\leq
	\frac{(1+\mu)^2}{ (1-\mu) r^{1-1/p}} 
		\Bigl(\TrNorm{Y} - \TrNorm{X} + 2\TrNorm{X\c} \Bigr)
	+\tau\, r^{1/p-1/2}\, \frac{3+\mu}{1-\mu} \lqNorm{\A(X-Y)} \, .
	\end{equation}
	for all $X,Y\in \mc V$. 
\end{theorem}

This is a straightforward adaptation of the proof of \cite[Theorem~3.2]{KabKueRau15}: 
there the NSP is applied twice, each time to $\TwoNormb{(X-Y)\r}$. 
By assumption $(X-Y) \in \V_0$ and the subspace NSP handles this case. 

\subsubsection{Tools from representation theory}
In this section we simplify the $k$-th moments of the random variables 
$U A U\ad$ and $\ketbra \psi \psi$, 
where $U$ and $\ket \psi$ are drawn independently from the Haar measures. 
These derivations will be used when we bound the moments of 
$\Tr[U A \,U\ad \, M(\ketbra \psi \psi)]$ for maps $M \in \LL(\CC^n)$. 

In order to compute expectation values over the unitary group of the type
\begin{equation}\label{eq:def_k-th_moments}
	E\coloneqq \Ev_{U\sim \Haar}[(UAU\ad)^{\otimes k}]
\end{equation}
we will use some basic facts from representation theory. 
Specifically, we use the decomposition from Schur Weyl duality
\begin{equation}
	(\CC^n)^{\otimes k} \cong \bigoplus_\lambda \pi^k_\lambda \otimes \rho^n_\lambda
\end{equation}
into irreducible representations (irreps) $\rho^n_\lambda$ and $\pi^k_\lambda$ of the unitary group $\U(n)$ and the symmetric group $S_k$, respectively. 
The irreps are labelled by Young diagrams $\lambda$ with $k$ boxes and at most $n$ rows. 
Since $[U^{\otimes k}, E]=0$ for all $U \in \U(d)$, a famous theorem due to Schur (see, e.g., \cite[Theorem~4.2.10]{GooWal09}) implies that $E$ can be written as 
$E = \bigoplus_\lambda \1\otimes Y_\lambda$, 
where $Y_\lambda\in \rho^n_\lambda$. 
But we also have that $[\sigma,E]=0$ for all $\sigma\in S_k$ and Schur's Lemma implies that
\begin{equation}\label{eq:decompose_U_average}
	\Ev_{U\sim \Haar}[(UAU\ad)^{\otimes k}] = \sum_\lambda a_\lambda P_\lambda \, ,
\end{equation}
where each $a_\lambda \in \CC$ and each $P_\lambda$ is the projection onto $\pi^k_\lambda \otimes \rho^n_\lambda$. 
The coefficients can be calculated as the expansion coefficients
\begin{align}
a_\lambda = \frac{\Tr[A^{\otimes k}P_\lambda]}{\Tr[P_\lambda]} \, .
\end{align}

This argument also yields that 
\begin{align}
F\coloneqq\Ev_{\ket{\psi}\sim \Haar}[\ketbra{\psi}{\psi}^{\otimes k}] 
=  \sum_\lambda b_\lambda P_\lambda \, .
\end{align}
The coefficients are
$b_\lambda \propto \Tr[P_\lambda F] $ 
and we have
$\ketbra \psi \psi^{\otimes k} P_{\Sym^k} = \ketbra \psi \psi^{\otimes k}$ 
and $P_\lambda P_{\lambda'} = 0$ for $\lambda \neq \lambda'$. 
Together with $\Tr[F]=1$ we obtain that 
\begin{align}\label{eq:Sym_k}
\Ev_{\ket{\psi}\sim \Haar}[\ketbra{\psi}{\psi}^{\otimes k}]=\frac{1}{\Tr[P_{\Sym^k}]}\, P_{\Sym^k} \, ,
\end{align}
where $P_{\Sym^k}$ is the projector onto the fully symmetric subspace in $(\CC^n)^{\otimes k}$. 

\myparagraph{First moments}
Setting $k=1$ and using that $\Tr[A^2] = \TwoNorm{A}^2$ and $\ketbra{\psi}{\psi}^2=\ketbra{\psi}{\psi}$ yields that 
\begin{equation}\label{eq:EvUAU}
	\Ev_{U\sim \Haar}[U A U^\dagger] 
	= \frac{\Tr[A]}{n}\, \1
	.
\end{equation}
and
\begin{align}\label{eq:Epsipsi}
\Ev_{\ket{\psi}\sim \Haar} [(\ketbra\psi\psi)^2] 
= \kw{n}\,\1 \, .
\end{align}

\myparagraph{Second moments}
The \emph{flip operator} $\FF \in \L(\CC^n \otimes \CC^n)$ is given by the linear extension of 
\begin{equation}
	\FF \ket{\psi}\otimes \ket{\phi} \coloneqq\ket{\phi}\otimes \ket{\psi}\, .
\end{equation}
The projector onto the symmetric subspace of $\CC^n \otimes \CC^n$ can be written as 
$P_{\Sym^2} = \kw 2 \left( \1 + \FF\right)$. 
The dimension of the fully symmetric subspace in $\CC^n \otimes \CC^n$ 
is 
\begin{equation}
\begin{aligned}
\Tr[P_{\Sym^2}]
&= 
\Tr[\1]/2 + \Tr[\FF]/2
= 
n^2/2 + n/2 
\\
&= 
n(n+1)/2 \, .
\end{aligned}
\end{equation}
Together we obtain
\begin{equation}\label{eq:psi_second_moment}
	\Ev \left[ \left( |\psi \rangle \langle \psi | \right)^{\otimes 2} \right] 
	= 
	\frac{2}{n(n+1)} P_{\Sym^2} 
	= 
	\frac{1}{n(n+1)} \left( \1 + \FF \right) \, .
\end{equation}

In order to derive the second moment of $U A U\ad$ we also need the projector onto the anti-symmetric subspace $P_{\wedge^2} = \kw 2\left( \1 - \FF\right)$. 
From Eq.~\eqref{eq:decompose_U_average} and $\Tr[\1]=n^2$ and $\Tr[\FF] = n$ we obtain
\begin{equation}
		E=  \frac{2 \Tr \bigl[ P_{\Sym^2} A \bigr]}{n(n+1)} P_{\Sym^2} 
		+ \frac{2 \Tr \bigl[ P_{\wedge^2} A \bigr])}{n(n-1)} P_{\wedge^2} \, .
\end{equation}
Evaluating the remaining traces in a basis or using tensor network diagrams yields 
	\begin{equation}\label{eq:U_second_moment}
		\Ev \left[ \left( U A U\ad \right)^{\otimes 2} \right] 
		= 
		\frac{\Tr(A)^2+\| A \|_2^2 }{n(n+1)} P_{\Sym^2} 
		+ \frac{\Tr(A)^2 - \| A \|_2^2 }{n(n-1)} P_{\wedge^2} \, .
	\end{equation}

\subsection{Our bound on the minimum conic singular value}\label{sec:OurLambdaMin}
The following theorem is the main technical result of this work. 

\begin{theorem}[Minimum conic singular value bound]\hfill
	\label{thm:lambda_min}
	Let $(A_i, \ket{\psi_i})_{i \in [m]}$ be a normalized 4-generic measurement ensemble (Definition~\ref{def:measurement_ensemble}). 
	Denote by $\A: \HT(\CC^n) \to \RR^m$ the linear map given by the components
	\begin{equation}
		\A(M)_i \coloneqq \Tr\bigl[A_i M(\ketbra{\psi_i}{\psi_i})\bigr] \, .
	\end{equation}
	Moreover, for $c_\mu>0$ denote by
	\begin{equation}\label{eq:def:K}
		K\coloneqq
		\Bigl\{M \in \HT(\CC^n)_0: \TrNorm{J(M)} \leq c_\mu \sqrt{r} \TwoNorm{J(M)} \Bigr\} ,
	\end{equation}
	the cone of trace-annihilating maps of ``effective Kraus-rank'' at most $r$. 
	Then, for the constant $c$ from Lemma~\ref{lem:our_lower_tail_bound}, for any 
	$c_\lambda \in{} ]0,c[$, $q \geq 1$
	and any 
	\begin{equation} \label{eq:m_bound}
		m > m_0 \coloneqq 125\,\e\cdot\left(\frac{c_\mu}{c-c_\lambda}\right)^2 \, r \ln(n)n(n+1)
	\end{equation}
	the $\ell_q$-minimum conic singular value of $\A$ is lower bounded as
	\begin{equation}
		\inf_{M \in K} \frac{\lqNorm{\A(M)}}{\fnorm{M}}
		\geq
		\frac{c-c_\lambda}{5}
		\TwoNorm{A}
		\frac{m^{\frac{1}{q}} \left( 1 - \sqrt{\frac{m_0}{m}} \right)}{n(n+1)} 
	\end{equation}
	with probability at least 
	$1-\e^{-c_\lambda^2 m/5}$ 
	over $\A$. 
\end{theorem}

\begin{proof}[Proof of Theorem~\ref{thm:lambda_min}]
	We choose $E$ in Lemma~\ref{lem:Mendelson} to be $E_{\nu}$ from Lemma~\ref{lem:W}. 
	Inserting the bounds from Lemmas~\ref{lem:our_lower_tail_bound} and~\ref{lem:W} into Eq.~\eqref{eq:mendelson} implies
	\begin{equation}
		\begin{split}
			\inf_{M \in E_\nu}&m^{\frac{1}{2}-\frac{1}{q}} \lqNorm{\A(M)}
			\geq \\& \qquad
			c \sqrt m\, \xi
			\left(1-\frac
			{(2\xi)^2 (n-1)\,n\,(n+1)^2} 
			{2\TwoNorm{A}^2\nu^2}
			\right)^2
			-
			2\frac{c_\mu\sqrt{6\,\e\ln(n)r}}{n} \, \nu\TwoNorm{A} 
			-\xi \lambda
		\end{split}
	\end{equation}
	with probability at least $1-\e^{-\lambda^2/2}$ over $\A$. 
	Choosing $\nu=\nu_0\coloneqq n\sqrt{n(n+1)}/\TwoNorm{A}$ 
	and using that $(n+1)(n-1)/n^2 \leq 1$ 
	yields 
	\begin{equation}
			\inf_{M \in E_{\nu_0}}
			m^{\frac{1}{2}-\frac{1}{q}} \lqNorm{\A(M)}
			\geq
			c \sqrt m \,
			\xi \left(1- 2\xi^2 \right)^2
			-
			 c_\mu \sqrt{6\,\e\ln(n)n(n+1)\,r}
			-\xi \lambda \, .
	\end{equation}
	Next, choosing $\xi = 1/\sqrt{10}$ 
	yields a maximum value greater than $1/5$ for $\xi \left(1- 2\xi^2 \right)^2$.
	This choice yields the bound
	\begin{equation}
			\inf_{M \in E_{\nu_0}}
			m^{\frac{1}{2}-\frac{1}{q}} \lqNorm{\A(M)}
			\geq 
			\frac{c}{5} \sqrt m 
			-
			c_\mu \sqrt{6\,\e\ln(n)n(n+1)\,r}
			-\frac{\lambda}{\sqrt{10}}\, .
	\end{equation}
	Furthermore, we set 
	$\lambda = \frac{c_\lambda \sqrt{10}}{5} \sqrt{m}$ 
	with some constant $0<c_\lambda<c$ to obtain
	\begin{equation}
			\inf_{M \in E_{\nu_0}}
			m^{\frac{1}{2}-\frac{1}{q}} \lqNorm{\A(M)}
			\geq 
			\frac{c-c_\lambda}{5} \sqrt m 
			-
			 c_\mu \sqrt{6\,\e\ln(n)n(n+1)\, r}
			\, .
	\end{equation}
	So, we choose 
	$m>m_0$ with 
	\begin{equation}
		\sqrt{m_0}
		\coloneqq 
		\frac{5\sqrt{6\,e}\,c_\mu}{c-c_\lambda} \sqrt{\ln(n)n(n+1)\, r}
	\end{equation}
	in order guarantee that the infimum yields a positive value. 
	This choice leads to
	\begin{align}
	\inf_{M \in E_{\nu_0}}\lqNorm{\A(M)}
	&\geq 
	\frac{c-c_\lambda}{5} m^{\frac{1}{q}} \left(1 -\sqrt{\frac{m_0}{m}} \right) 
	\, .
	\end{align}
	
	As the infimum over $E_{\nu}$ is homogeneous in $\nu$ (i.e., ``proportional'' to all $\nu\geq 0$), 
	we obtain for the cone 
	$E = \bigcup_{\nu\geq 0} E_{\nu}$ generated by $E_{\nu_0}$ that
	\begin{equation}
	\begin{aligned}
	\inf_{M \in E} \frac{\lqNorm{\A(M)}}{\fnorm{M}}
	&=
	\inf_{M \in E_{\nu_0}}\lqNorm{\A(M)} /\nu_0
	\\
	&=
	\inf_{M \in E_{\nu_0}}\lqNorm{\A(M)} \frac{\TwoNorm{A}}{n(n+1)}
	\\
	&\geq
	\frac{c-c_\lambda}{5}
	\TwoNorm{A}
	\frac{m^{\frac{1}{q}} \left( 1- \sqrt{\frac{m_0}{m}} \right)}{n(n+1)} \, ,
	\end{aligned}	
	\end{equation}
	where we have remembered the choice of $\nu_0$ form the beginning of the proof. 
	This bound holds with probability at least 
	$1-\e^{-\lambda^2/2}
	=
	1-\e^{-c_\lambda^2 m/5}$ 
	over $\A$ for any $c_\lambda \in{} ]0,c[$. 
\end{proof}

\subsubsection{Upper bound on the mean empirical width \texorpdfstring{$W_m$}{}}
We prove an upper bound on $W_m$ defined in Eq.~\eqref{eq:W_def} for $E$ being a slice of the cone of effectively rank-$r$ maps. 

\begin{lemma}[Bound on $W_m$]\label{lem:W}	
	For $r \in \ZZ^+$, $\nu>0$, and $c_\mu>0$ let
	\begin{equation}\label{eq:lem_K_def}
	K\coloneqq
	\Bigl\{M \in \HT(\CC^n)_0: \TrNorm{J(M)} \leq c_\mu \sqrt{r} \TwoNorm{J(M)} \Bigr\} 
	\end{equation}
	and
	\begin{equation}
	E_\nu \coloneqq
	\{M \in K: \fnorm{M}=\nu \} .
	\end{equation}
	For $m\geq 2n^2 \ln(n)$ let 
	$(A_i, \ket{\psi_i})_{i \in [m]}$ 
	be a normalized 4-generic measurement ensemble 
	(Definition~\ref{def:measurement_ensemble}). 
	Then the mean empirical width \eqref{eq:W_def} is bounded as
	\begin{equation}
		W_m(E_{\nu})\leq 
		\frac{c_\mu \sqrt{6\, \e \ln(n)\,r}}{n} \, \nu\TwoNorm{A} \, .
	\end{equation}
\end{lemma}

\begin{proof}
By definition \eqref{eq:W_def}, $W_m$ reads as
\begin{align}
W_m(E_{\nu}) 
&=
\Ev \sup_{M \in E_{\nu} } \kw{\sqrt m}\sum_{i=1}^m \epsilon_i \Tr[U_iA\,U_i\ad M(\ketbra{\psi_i}{\psi_i})] .
\end{align}
Since $W_m(E_{\nu}) = \nu W_m(E)$ with $E \coloneqq E_{1}$, 
it is enough to prove the lemma for $\nu=1$. 

With 
\begin{equation}\label{eq:Hi}
	H_i \coloneqq  (U_iA\,U_i\ad) \otimes (\ketbra{\psi_i}{\psi_i}^T) 
	\in \L(\CC^n\otimes \CC^n)
\end{equation}
and using the identity \eqref{eq:ChoiMeasurement}, 
the ``expectation value'' can also be written as
\begin{equation}
	\Tr[U_iA\,U_i\ad M(\ketbra{\psi_i}{\psi_i})]
	=
	\Tr[H_i\,J(M)] \, .
\end{equation}
We define 
\begin{equation}\label{eq:H}
	H\coloneqq \kw{\sqrt m}\sum_{i=1}^m \epsilon_i H_i
\end{equation}
to arrive at the compact form
\begin{equation}
	W_m(E) =  \Ev \sup_{M \in E} \Tr[HJ(M)] \, .
\end{equation}
The application of H\"older's inequality 
yields
\begin{equation}
	W_m(E) \leq  \Ev \sup_{M \in E} \snorm{H} \TrNorm{J(M)} \, .
\end{equation}
Using the definition \eqref{eq:lem_K_def} of $K$ and that 
$\TwoNorm{J(M)}=\fnorm{M}=\nu=1$ we obtain 
\begin{equation}\label{eq:descent_cone_Hoelder}
	W_m(E)\leq 
	c_\mu \sqrt{r} \Ev \snorm{H} \, .
\end{equation}

In order to bound $\Ev \snorm{H}$ we proceed similarly as in the proof of 
Ref.~\cite[Proposition~13]{KueRauTer15}. 
By applying a non-commutative Khintchine inequality (see Theorem~\ref{thm:khintchine} in Appendix~\ref{sec:Khintchine}) to \eqref{eq:H} we obtain 
\begin{equation}
	\Ev_{\epsilon_i}[\snorm{H}] 
	\leq 
	\sqrt{\frac{2\ln(2n^2)}{m}} \Ev\, \snormB{\left(\sum_{i=1}^mH_i^2\right)^{1/2}} . 
\end{equation}
Thanks to Jensen's inequality, $\Ev \snormb{\sqrt{X}} \leq \Ev \sqrt{\snorm{X}} \leq \sqrt{\Ev \snorm{X}}$ and, hence, 
\begin{equation}\label{eq:jensen}
	\Ev_{\epsilon_i}[\snorm{H}] 
	\leq 
	\sqrt{\frac{2\ln(2n^2)}{m}} \left(\Ev \, \snormB{\sum_{i=1}^mH_i^2}\right)^{1/2} 
	.
\end{equation}
The Matrix Chernoff Bound \cite[Theorem~5.1.1]{Tro12} implies that
\begin{equation}\label{eq:chernoff}
	\Ev \, \snormB{\sum_{i=1}^mH_i^2}
	\leq 
	\frac{\e^\theta-1}{\theta} \, m\,  
	\snormb{\Ev\bigl[H_1^2\bigr]}
	+ \sup_{H_1}\snorm{H_1}\frac{\ln(n^2)}{\theta} \, ,
\end{equation}
where we have already used that $\{H_i\}$ are i.i.d.\ operators.
With the averages~\eqref{eq:EvUAU} and~\eqref{eq:Epsipsi} we find that
$H_1$ (defined in Eq.~\eqref{eq:Hi}) satisfies
\begin{align}\label{eq:H1_squared}
\snormb{\Ev\bigl[H_1^2\bigr]}
=&
\snormB{\frac{\Tr[A^2]}{n} \1}\; \snormB{\kw n \, \1}
=\frac{\TwoNorm{A}^2}{n^2} .
\end{align}
Moreover, $H_1$ always satisfies 
\begin{equation}
\snorm{H_1} =  \snorm{A}\snormb{\ketbra \psi \psi^T}=1 \, ,
\end{equation}
as $A$ and $\ket \psi$ are both normalized. 
Hence,
\begin{equation}
	\Ev \, \snormB{\sum_{i=1}^mH_i^2}
	\leq 
	\frac{\e^\theta-1}{\theta} 
		\, m \,
		\, \frac{\TwoNorm{A}^2}{n^2} 
	+ \frac{\ln(n^2)}{\theta} 
	\, .
\end{equation}
Choosing $\theta = 1$ and proceeding from Eq.~\eqref{eq:jensen} we obtain
\begin{equation}
	\Ev[\snorm{H}] 
	\leq 
	\sqrt{\frac{2\ln(2n^2)}{m}} 
	\left(
	(\e-1)\, 
	m\, 
	\frac{\TwoNorm{A}^2}{n^2}  
	+ 
	2 \ln(n)
	\right)^{1/2} .
\end{equation}
With $\TwoNorm{A}\geq \snorm{A}=1$ and the assumptions $m\geq 2n^2 \ln(n)$
and $\ln\bigl(\sqrt 2 n\bigr) \leq \frac 3 2 \ln(n)$
we obtain
\begin{equation}
\begin{aligned}
\Ev[\snorm{H}] 
&\leq 
\sqrt{4\ln\bigl(\sqrt{2}\,n\bigr)}\, \frac{\TwoNorm{A}}{n}
	\left((\e-1) + \frac{2 n^2 \ln(n)}{m} \right)^{1/2}
\\
&\leq
\frac{\sqrt{6\, \e \ln(n)}}{n} \, \TwoNorm{A} \, .
\end{aligned}	
\end{equation}
Finally, with Eq.~\eqref{eq:descent_cone_Hoelder},
\begin{equation}
	W_m(E)\leq \frac{c_\mu \sqrt{6\,\e \ln(n)\,r\, }}{n} \, \TwoNorm{A} \, ,
\end{equation}
which proves Lemma~\ref{lem:W}.
\end{proof}

\subsubsection{Lower bound on the marginal tail function \texorpdfstring{$Q_\xi$}{}}
In this section, we prove a lower bound on the marginal tail function \eqref{eq:Q_def}. 

\begin{lemma}[Lower tail bound]\label{lem:our_lower_tail_bound}
 	Let $(A_i, \ket{\psi_i})_{i \in [m]}$ be a normalized 4-generic measurement ensemble (Definition~\ref{def:measurement_ensemble}).
	For some trace annihilating $M \in \TP(\CC^n)_0$ define the random variable
	\begin{equation}\label{eq:def_S}
		S \coloneqq \Tr[U A \,U\ad \, M(\ketbra \psi \psi)] \, .
	\end{equation}	
	Then $S$ satisfies 
	\begin{align}
	\Pr\bigl[|S| \geq \xi \bigr]
	&\geq
	c \left(1-\frac
	{\xi^2 (n-1)\,n\,(n+1)^2} 
	{2\TwoNorm{A}^2\Bigl(\TwoNorm{M(\1)}^2 + \fnorm{M}^2\Bigr)}
	\right)^2
	\end{align}
	for all $\xi>0$, 
	where $c$ is an absolute constant.
\end{lemma}

We will prove this lemma using the Paley-Zygmund inequality. 
This inequality states that for every non-negative random variable $Z\geq 0$ and parameter 
$\theta \in [0,1]$ 
\begin{equation}\label{eq:PaleyZygmund}
	\Pr\bigl[Z > \theta \Ev[Z]\bigr] \geq (1-\theta)^2 \frac{\Ev[Z]^2}{\Ev\bigl[Z^2\bigr]} \, .
\end{equation}
We will choose $Z = |S|^2$ so that we can use a lower bound on the second moment and an upper bound on the fourth moment of $S$. 

\begin{proof}[Proof of Lemma~\ref{lem:our_lower_tail_bound}]
	As typically done \cite{Tro15}, we use the Paley-Zygmund inequality \eqref{eq:PaleyZygmund} to establish the lower tail bound,
	\begin{equation}
	\begin{aligned}
	\Pr\bigl[|S| \geq \xi \bigr]
	&=
	\Pr\bigl[|S|^2 \geq \xi^2 \bigr] 
	\\
	&\geq
	\left(1-\frac{\xi^2}{\Ev[|S|^2]}\right)^2\,
	\frac{\Ev[|S|^2]^2}{\Ev[|S|^4]} \, .
	\label{eq:Paley-Zygmund}
	\end{aligned}		
	\end{equation}	
	Inserting the fourth moment bound~\eqref{eq:fourth_moment_bound} and the second moment \eqref{eq:full_second_moment} 
	into \eqref{eq:Paley-Zygmund} 
	we obtain for another absolute constant $c>0$ that
	\begin{equation}
		\begin{aligned}
			\Pr\bigl[|S| \geq \xi \bigr]  
			&\geq
			c
			\left(1-\frac{\xi^2}{\Ev[|S|^2]}\right)^2
			\frac{(n-1)^2 n^2 (n+1)^2(n+2)(n+3)}{(n-1)^2n^2(n+1)^4} 
			\\
			&\geq
			c \left(1-\frac
			{\xi^2 (n-1)\,n\,(n+1)^2} 
			{2\TwoNorm{A}^2\Bigl(\TwoNorm{M(\1)}^2 + \fnorm{M}^2\Bigr)}
			\right)^2 , 
		\end{aligned}
	\end{equation}
	where \eqref{eq:full_second_moment} has been used again in the second step. 
	This bound proves Lemma~\ref{lem:our_lower_tail_bound}. 
\end{proof}

Actually, we can fully calculate the second moment without any assumptions on $\Tr[A]$. 

\begin{lemma}[Second moment]
 	Let 
	$(A_i, \ket{\psi_i})_{i \in [m]}$ be a 4-generic measurement ensemble (Definition~\ref{def:measurement_ensemble}) and $M \in \HT(\CC^n)_0$ be a trace annihilating map. 
	
	Then $S \coloneqq \Tr[U A \,U\ad \, M(\ketbra \psi \psi)]$ has the second moment
	\begin{equation}\label{eq:full_second_moment}
	\Ev\bigl[|S|^2\bigr]
	= 
	\frac{2\TwoNorm{A}^2\, n -2\Tr[A]^2}{(n-1)\,n^2\,(n+1)^2} 
	  \Bigl(\TwoNorm{M(\1)}^2 + \fnorm{M}^2\Bigr)  \, .
	\end{equation}
\end{lemma}

\begin{figure}[t]
	\centering
	\def\d{.35}
	\def\ydist{.2}
	\leavevmode
	\beginpgfgraphicnamed{fig9_TNs}%
	\begin{tikzpicture}
	\def\len{.4}%
	\node (vector){
	\begin{tikzpicture}
	\node (TNpsi) {\tikz{
			\node (psi) [Bbox] {$\ket\psi$};
			\draw [leg] (psi.west) --++(-\len,0);
		}};
	\path (TNpsi.west) ++(0,0) node [anchor = east] {Vector $\CC^n \ni \ket\psi = $};
	\end{tikzpicture}
	};
	\path (vector.south) ++ (0,-\ydist\baselineskip) node (operator) [anchor = north]{
	\begin{tikzpicture}
	\node (TNA) {\tikz{
			\node (A) [Bbox] {$A$};
			\draw [leg] (A.east) ++(\len,0) --(A.east);
			\draw [leg] (A.west) --++(-\len,0);
		}};
	\path (TNA.west) ++(0,0) node [anchor = east] {Operator $\L(\CC^n) \ni A = $};
	\path (TNA.east) ++(0,0) node (eq) [anchor = west] {$\cong$};
	\path (eq.east) ++(0,0) node (eq) [anchor = west] {\tikz{
			\node (A) [Bbox] {$A$};
			\SetCoordinates{A}
			\draw [leg] (Aoli) --++(-\len,0);
			\draw [leg] (Auli) ++(-\len,0)--(Auli);
		}};
	\end{tikzpicture}    
	};
	\path (operator.south) ++ (0,-\ydist\baselineskip) node (M) [anchor = north]{
	\begin{tikzpicture}
	\node (TNA) {\tikz{
			\node (A) [Bbox] {$A$};
			\path (A.west) ++(-\len,0) node (M) [Bbox, anchor = east] {$M$};
			\SetCoordinates{M}
			\SetCoordinates{A}
			\draw [leg] (Moli) --++(-\len,0);
			\draw [leg] (Muli) ++(-\len,0)--(Muli);
			\draw [leg] (Aoli) --(More);
			\draw [leg] (Mure) --(Auli);
		}};
	\path (TNA.west) ++(0,0) node [anchor = east] {$M \in \LL(\CC^n)$ applied to $A$:  $M(A) = $};
	\end{tikzpicture}    
	};
	\path (M.south) ++ (0,-\ydist\baselineskip) node (flip) [anchor = north]{
	\begin{tikzpicture}
	\node (TNF) {\tikz{
			\node (F) [sbox, draw = none]{};
			\SetCoordinates{F}
			\path (Fore) ++ (\len,0) coordinate (ORe);
			\path (Fuli) ++ (-\len,0) coordinate (ULi);
			\path (Fure) ++ (\len,0) coordinate (URe);
			\path (Foli) ++ (-\len,0) coordinate (OLi);				
			\draw [Leg] (ORe) to[out = 180, in = 0]  (ULi);
			\draw [Leg] (URe) to[out = 180, in = 0]  (OLi);
		}};
	\path (TNF.west) ++(0,0) node [anchor = east] {Flip operator $\L(\CC^n\otimes \CC^n) \ni \FF = $};
	\end{tikzpicture}  	
	};
	\path (flip.south) ++ (0,-\ydist\baselineskip) node (product) [anchor = north]{
	\begin{tikzpicture}
	\node (FMF) {
		\begin{tikzpicture}
		\def\llen{.6}
		\def\slen{.3}
		\node (M) [Bbox] at (0,0) {$M$};
		\path (M.south) ++(0,-.5\baselineskip) node (Mad) [Bbox, anchor = north]{$M^\star$};
		\SetCoordinates{M}
		\SetCoordinates{Mad}
		\draw [leg] (Madure) --++(\len,0);
		\draw [leg] (Madore) ++(\len,0) --(Madore);
		\draw [leg] (Maduli) ++(-\len,0)--(Maduli);
		\draw [leg] (Madoli) --++(-\len,0);
		\draw [leg] (Mure) --++(\len,0);
		\draw [leg] (More) ++(\len,0) --(More);
		\draw [leg] (Muli) ++(-\len,0)--(Muli);
		\draw [leg] (Moli) --++(-\len,0);
	\end{tikzpicture} 
	};
	\path (FMF.west) ++(0,0) node [anchor = east] {Product of maps $M \otimes M^\ast=$};
	\end{tikzpicture}   
	};
	\end{tikzpicture}
	\endpgfgraphicnamed
	\caption{Tensor network diagrams: tensors are denoted by boxes with one line for each index. 
		Contraction of two indices corresponds to connection of the corresponding lines. 
		Examples: 
		A vector $\ket \psi$, vectorization of an operator $A$, $M \in \LL(\CC^n)$ applied to that vectorization, the non-vectorized version of the flip operator $\FF$, and $M^{1,1} = M \otimes M^{\star}$. 
	}
	\label{fig:TNs}
\end{figure}
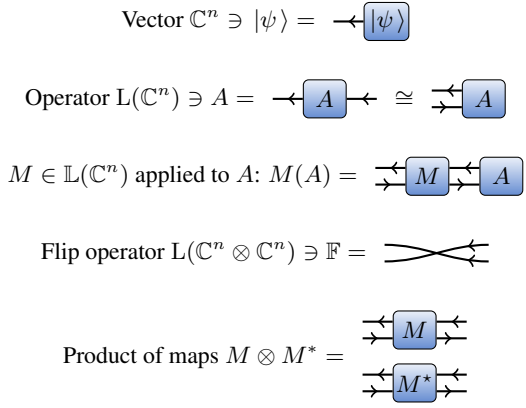
				
\begin{proof}
	First, we will derive some identities for certain traces containing $\1$, $\FF$, and $M$. 
	As $M$ is trace annihilating, i.e., $M\ad(\1)=0$, we obtain
	\begin{equation}
		\Tr[\1 M^{1,1}(\1)] = \Tr[M^{1,1\dagger}(\1)] = 0
	\end{equation}
	and
	\begin{equation}
		\Tr[\1 M^{1,1}(\FF)] = \Tr[M^{1,1\dagger}(\1) \FF]=0 \, .
	\end{equation}
	Moreover, using ``the swap-trick'' $\Tr[\FF (A \otimes B) ] = \Tr[AB]$ we obtain
	\begin{equation}
	\begin{aligned}
	\Tr[\FF M^{1,1}(\1)] &= \Tr\bigl[\FF \bigl(M(\1)\otimes M^\star(\1)\bigr)\bigr] 
	\\
	&= \Tr[M(\1)M^\star(\1)] 
	\\
	&= \Tr[M(\1)M(\1)\ad]  
	\\
	&= \TwoNorm{M(\1)}^2 \, .
	\end{aligned}
	\end{equation}
	The last of these identities is
	\begin{equation}
	\begin{aligned}
	\Tr[\FF M^{1,1}(\FF)] &= \Tr[\FF (M\otimes M^\ast)(\FF)]
	\\&=
	\begin{matrix}
	\leavevmode
	\beginpgfgraphicnamed{figProofLemma15}%
	\begin{tikzpicture}[anchor=base, baseline=-.7em]
		\def\d{.35}
		\def\llen{.6}%
		\def\slen{.3}%
		\node (M) [Bbox,anchor = south] {$M$};
		\path (M.south) ++(0,-1em) node (Mad) [Bbox, anchor = north]{$M^\star$};
		\SetCoordinates{M}
		\SetCoordinates{Mad}
		\draw [leg] (Madure) --++(\llen,0)  |- (More);
		\draw [leg] (Mure)   --++(\slen,0)  |- (Madore);
		\draw [leg] (Moli)   --++(-\llen,0) |- (Maduli);
		\draw [leg] (Madoli) --++(-\slen,0) |- (Muli);
		\path (Mure) ++(\slen,0)  ++ (11pt,0)  coordinate (anc);
	    \path (Madoli) ++(-\slen,0) ++ (-11pt,0) coordinate (anc);
	\end{tikzpicture} 
	\endpgfgraphicnamed
	\end{matrix}
	\\
	&=\Tr[J(M)J(M^\star)] \label{eq:M11fro}
	\\
	&= \Tr[J(M)J(M)\ad] 
	\\
	&= \TwoNorm{J(M)}^2 = \fnorm{M}^2 \, ,
	\end{aligned}
	\end{equation}
	see Figure~\ref{fig:TNs} for an explanation of the tensor network diagram. 
	
	Next, using the expressions for the second moments of $\ketbra \psi \psi$ and $UAU\ad$, 
	\eqref{eq:psi_second_moment} and \eqref{eq:U_second_moment},
	we obtain
	\begin{widetext}
		\begin{equation}
		\begin{aligned}
		\Ev[|S|^2] 
		&= 
		\Tr\bigl[\bigl(UA\, U\ad M(\ketbra \psi \psi)\bigr)
		\otimes \bigl(UA\, U\ad M^\star(\ketbra \psi \psi)\bigr) \bigr]
		\\
		&= \Tr\Bigl[
		\Ev\bigr[(UAU\ad)^{\otimes 2}\bigl] M^{1,1}
		\Bigl( \Ev\bigl[\ketbra{\psi}{\psi}^{\otimes 2}\bigr] \Bigr)
		\Bigr]
		\\
		&=\kw{2n^2(n+1)}\sum_\pm\frac{\Tr(A)^2 \pm \Tr(A^2)}{n\pm 1} \Tr\Bigl[(\1\pm\FF) M^{1,1} (\1+\FF ) \Bigr]
		\\
		&=
		\kw{2n^2(n+1)}
		\sum_\pm\frac{\pm\Tr(A)^2+\Tr(A^2)}{n\pm 1} \Tr\Bigl[\FF M^{1,1} (\1+\FF ) \Bigr] \, .
		\end{aligned}
		\end{equation}
		We finish the proof by using the traces containing $M$ from the beginning and that 
		$\Tr[A^2] = \TwoNorm{A}^2$,
		\begin{equation}
		\begin{aligned}
		\Ev[|S|^2]
		&=
		\frac{1}{\,n^2\,(n+1)} 
		\Bigl(\frac{\Tr[A]^2+\Tr[A^2]}{n+1} + \frac{-\Tr[A]^2+\Tr[A^2]}{n-1}\Bigr)
		\Bigl(\TwoNorm{M(\1)}^2 + \fnorm{M}^2\Bigr)
		\\
		&= 
		\frac{2\TwoNorm{A}^2\, n -2\Tr[A]^2}{(n-1)\,n^2\,(n+1)^2} 
		\Bigl(\TwoNorm{M(\1)}^2 + \fnorm{M}^2\Bigr) \, .
		\end{aligned}
		\end{equation}
		\qedhere 
	\end{widetext}
\end{proof}

\begin{lemma}[Upper bound on the fourth moment]\label{lem:FourthMoment}
The random variable $S$ from Lemma~\ref{lem:our_lower_tail_bound} has a fourth moment bounded as
	\begin{equation}\label{eq:fourth_moment_bound}
	\Ev[|S|^4] \leq c_3 
	\frac{ \TwoNorm{A}^4 \Bigl(\TwoNorm{M(\1)}^2+\fnorm{M}^2\Bigr)^2 }
	{(n-1)^2 n^2 (n+1)^2(n+2)(n+3)} \, ,
	\end{equation}
	where $c_3$ is an absolute constant. 
\end{lemma}

The proof of this lemma uses facts about the symmetric group $S_4$ all of which are stated in Appendix~\ref{sec:S4}. 

\begin{proof}
The fourth moment is
\begin{equation}
\begin{aligned}
\Ev[|S|^4] 
&= 
\Tr\bigl[\bigl(UA U\ad M(\ketbra \psi \psi)\bigr)^{\otimes 2}
\otimes \bigl(UA U\ad M^\star(\ketbra \psi \psi)\bigr)^{\otimes 2} \bigr]
\\
&= \Tr\Bigl[
\Ev\bigr[(UA\,U\ad)^{\otimes 4}\bigl] M^{\otimes 2,2}
\Bigl( \Ev\bigl[\ketbra{\psi}{\psi}^{\otimes 4}\bigr] \Bigr)
\Bigr] \, . \label{eq:momentS4}
\end{aligned}
\end{equation}
According to Eq.~\eqref{eq:Sym_k} the average over $\ket \psi$ yields
\begin{equation}
\Ev\bigl[\ketbra{\psi}{\psi}^{\otimes 4} \bigr]
= 
\kw{d_1(n)}P_{\Sym^4} ,
\end{equation}
where we have used that $\Tr(P_{\Sym^4})= d_1(n)$ with $d_1(n)=n\,(n+1)(n+2)(n+3)/24$
from Eq.~\eqref{eq:d_i}
being the dimension $\Tr(P_{\Sym^4})$ corresponding to the trivial representation $\Sym^4$.

According to Eq.~\eqref{eq:decompose_U_average} we have
\begin{align}
\Ev\bigr[(UA\,U\ad)^{\otimes 4}\bigl]
= \sum_{i=1}^5 a_i \,  P_i \, .
\end{align}
By taking the trace, we obtain the corresponding expansion coefficients
\begin{equation}
	a_i=\frac{\Tr[A^{\otimes 4}\,P_i]}{d_i(n)} 
\end{equation}
where $d_i(n) = \Tr[P_i]$ is also given in \eqref{eq:d_i} and each $P_i$ is the representation of the central minimal projection $p_i$ of $S_4$ (see \eqref{eq:p_i}) on $(\CC^n)^{\otimes 4}$.  
Inserting everything into \eqref{eq:momentS4}, we obtain
\begin{equation}\label{eq:S4_of_p_i}
	\Ev[|S|^4] = 
	\sum_{i=1}^5 a_i\,
	\frac{\Tr\Bigl[ P_i M^{\otimes 2,2} (P_1)\Bigr]}
	{d_1(n)}   \, . 
\end{equation}
For a permutation $\sigma \in S_4$ with representation $R(\sigma)\in \L((\CC^n)^{\otimes 4})$ the Hilbert-Schmidt overlap
$\Tr[A^{\otimes 4}R(\sigma)]$ only depends on the conjugacy class of $\sigma$. 
We denote the conjugacy class containing permutations composed of each $j_i$ disjoint cycles of sizes $\{k_i\}_{i\in [l]}$ by $k_1^{j_1}k_2^{j_2}\dots k_l^{j_l}$, e.g., $2^1 \subset S_4$ are the transpositions. 
One can conclude, e.g., from tensor network diagrams that
\begin{equation}
	\begin{aligned}
		\Tr[ \1\, A^{\otimes 4}] &= \Tr[A]^4=0 \, ,
		\\
		\Tr[ 2^2 A^{\otimes 4}] &= \Tr[A^2]^2,
		\\
		\Tr[ 2^1 A^{\otimes 4}] &= \Tr[A^2]\Tr[A]^2=0 \, ,
		\\
		\Tr[ 4^1 A^{\otimes 4}] &= \Tr[A^4] \, ,
		\\
		\Tr[ 3^1 A^{\otimes 4}] &= \Tr[A^3] \Tr[A] =0\, .
	\end{aligned}
\end{equation}

With 
the sizes of the conjugacy classes \eqref{eq:tab:S4}, 
the minimal projections \eqref{eq:p_i}, and
the dimensional factors \eqref{eq:d_i}
\begin{equation}
	\begin{aligned}
		a_1 &= \frac{3 \Tr[A^2]^2 + 6 \Tr[A^4]}
		{n\, (n+1)(n+2)(n+3)} \, ,
		\\
		a_2 &= \frac{3 \Tr[A^2]^2 - 6 \Tr[A^4]}
		{(n-3)(n-2)(n-1)\, n} \, ,
		\\
		a_3 &= \frac{6 \Tr[A^2]^2}
		{(n-1)\,n^2\,(n+1)} \, ,
		\\
		a_4 &= \frac{- 3 \Tr[A^2]^2 - 6 \Tr[A^4]}
		{(n-1)\,n\,(n+1)(n+2)} \, ,
		\\
		a_5 &= \frac{- 3 \Tr[A^2]^2 + 6 \Tr[A^4]}
		{(n-2)(n-1)\,n\,(n+1)} \, ,
		\, 
	\end{aligned}
\end{equation}
for $ n \geq 4$. 
For $n=3$ we have $P_2=0$ and, hence $a_2 = 0$. 
For $n=2$ we have $P_2=0$ and $P_5=0$ and, hence, $a_2=a_5 = 0$. 
In both cases, the remaining $a_i$ are as stated above.

From the submultiplicativity of the Schatten $2$-norm
follows that 
$\TwoNorm{A^j} \leq \TwoNorm{A}^j$ for $j\in \ZZ^+$.  
Also using the bound
\begin{equation}
	\frac{1}{n-j} \leq \frac{j+1}{n} \qquad \text{for integers}\ n>j\geq 0 
\end{equation}
we obtain 
\begin{equation}\label{eq:a_i-bound}
	|a_i| \leq
	\frac{c_1 \TwoNorm{A}^4}{(n-1)^2\, n\, (n+1)}
\end{equation}
for all $i\in [5]$, where $c_1$ is an absolute constant. 

In order to bound $\Tr[ P_i M^{\otimes 2,2} (P_1)]$ in Eq.~\eqref{eq:S4_of_p_i} we observe that each projection $P_i$ is a linear combination of permutation matrices $\{R(\sigma)\}_{\sigma \in S_4}$, see \eqref{eq:p_i}, where only the permutation matrices $R(\sigma)$ depend on $n$. 
Hence, it is enough to bound $\Tr[R(\sigma) M^{2,2}(R(\tau))]$ for all permutations 
$\sigma, \tau\in S_4$.

Combining \eqref{eq:fourth_moment_bound} from Lemma~\ref{lem:bound_sigmaMtau} below, \eqref{eq:a_i-bound}, and \eqref{eq:S4_of_p_i} yields
\begin{equation}
	\Ev[|S|^4] \leq c_3 
	\frac{ \TwoNorm{A}^4 \Bigl(\TwoNorm{M(\1)}^2+\fnorm{M}^2\Bigr)^2 }
	{(n-1)^2 n^2 (n+1)^2(n+2)(n+3)}
\end{equation}
for some absolute constant $c_3$.  
\end{proof}

\begin{lemma}\label{lem:bound_sigmaMtau}
	For any $M\in\HT(\CC^n)_0$ and $\sigma,\tau\in S_4$
	\begin{equation}
		\bigl|\Tr[R(\sigma) M^{2,2}(R(\tau))]\bigr| 
		\leq 
		c_2 \Bigl(\TwoNorm{M(\1)}^2+\fnorm{M}^2\Bigr)^2
	\end{equation}
	for some absolute constant $c_2$. 
\end{lemma}

\begin{proof}
	We will conclude from Proposition~\ref{prop:TNbound} that 
	\begin{equation}\label{eq:fourth_moment_bound_max}
		\bigl|\Tr[R(\sigma) M^{2,2}(R(\tau))]\bigr|
		\leq 
		c_2' \max \bigl\{ \TwoNorm{M(\1)}^4, \fnorm{M}^4\Bigr\} 
	\end{equation}
	for some absolute constant $c_2'$, which implies the lemma. 
	
	We consider 
	$\HT(\CC^n)_0\subset \L(\L(\X) \to \L(\Y))$ as a subspace, where $\X =\CC^n=\Y$ are labels for the input and output Hilbert space. 
	The permutation operator $R(\tau)$ permutes and contracts the indices of $M^{2,2}$ in 
	$\Tr[R(\sigma) M^{2,2}(R(\tau))]$ that correspond to $\X$. 
	Similarly, the permutation operator $R(\sigma)$ permutes the indices of $M^{2,2}$ in 
	$\Tr[R(\sigma) M^{2,2}(R(\tau))]$ that correspond to $\Y$, which are contracted subsequently by $\Tr$. 
	So, no $\X$-index is contracted with a $\Y$-index. 
	Hence, 
	$\Tr[R(\sigma) M^{2,2}(R(\tau))]$ is a contraction without self-contractions of the tensors 
	$\{M, M^\star, M(\1), M^\star(\1), \Tr\circ M, \Tr \circ M^\star\}$. 
	The tensors $\Tr \circ M^\star$ and $\Tr\circ M$ vanish due to maps in $\HT(\CC^n)_0$ being trace annihilating. 
	Moreover, $\fnorm{M} = \fnorm{M^\star}$ and 
	$\TwoNorm{M(\1)} = \TwoNorm{M^\star(\1)}$. 
	Proposition~\ref{prop:TNbound} tells us that arbitrary closed tensor networks without self-contractions are bounded by the product of the Frobenius norms of the single tensors and implies the bound~\eqref{eq:fourth_moment_bound_max}. 
\end{proof}

\subsection{General tensor network bound}
\label{sec:TNbound}
In this section, we prove a simple bound on a fully contracted tensor network, which we used to prove Lemma~\ref{lem:bound_sigmaMtau}. 
A \emph{tensor} is a vector in 
$\CC^{n_1}\otimes \CC^{n_2} \otimes \dots \otimes \CC^{n_K}$, 
where $\X \otimes \Y$ denotes the tensor product of vector spaces $\X$ and $\Y$. 
It is often helpful to identify a tensor with its representation $t$ in terms of a product basis of the canonical bases of $\CC^{n_i}$. 
Then $t$ is given as an array of numbers, i.e., 
$t \in \CC^{n_1\times n_2\times \dots \times n_K}$. 

A \emph{tensor network} is a set of tensors together with a contraction corresponding to pairs of indices where both indices have the same dimension $n_i$. 
Instances of such tensor networks are, e.g., the workhorse of powerful simulation techniques for strongly correlated quantum systems \cite{Sch11}. 
We present a general version of such tensor networks and then prove our bound. 
For this purpose, we use a notation that is completely disjoint from the other sections. 
E.g., $C$ will be a contraction instead of a constant. 

A \emph{pointer} to an index of a tensor in 
$\CC^{n_1}\otimes \CC^{n_2} \otimes \dots \otimes \CC^{n_{K}}$
is just a number $k\in [K]$ used to identify the $k$-th index. 
A \emph{contraction} is a linear map 
\begin{equation}\nonumber
	C : 
	\CC^{n_1}\otimes \CC^{n_2} \otimes \dots \otimes \CC^{n_{K}} 
	\to 
	\CC^{n_{i_1}}\otimes \CC^{n_{i_2}} \otimes \dots \otimes \CC^{n_{I}} 
\end{equation}
given by a set of pairs of pointers 
$P=(\{k_l,k'_l\})_{l \in [L]}$ with even $K-I=2L$ so that 
each number $k\in [K]$ occurs at most once in at most one of the pairs. 
In particular, $k_l \neq k'_l$ for all $l \in [L]$. 
Moreover, the consistency condition $n_{k_l} = n_{k'_l}$ is required to hold for all $l\in [L]$ 
(in many relevant cases the dimensions $n_k$ assume only a very few different values). 
Let $\mc P= \{k_l,k'_l\}_{l \in [L]}$ be the set of pointers occurring in the pairs $P$ and $\mc I \coloneqq [K]\setminus \mc P$ the other pointers. 
Then the contraction $C(t)$ of a tensor 
$t\in \CC^{n_1}\otimes \CC^{n_2} \otimes \dots \otimes \CC^{n_{K}}$ is given by the components
\begin{equation}
		C(t)_{\alpha'_{i_1}, \alpha'_{i_2}, \dots, \alpha'_{i_I}}
		\coloneqq 
		\sum_{\alpha_{k} \in [n_{k}],  \ k \in [K]}
		t_{\alpha_1, \alpha_2, \dots, \alpha_{K}}
		\prod_{l\in \mc P}\delta_{\alpha_{k_l}, \alpha_{k'_l}} 
		\prod_{m \in \mc I} \delta_{\alpha_{m}, \alpha'_m}
		\, ,
\end{equation}
i.e., all index pairs $\{k_l,k'_l\}$ are summed over and the remaining indices are the indices of $C(t)$. 

A \emph{tensor network} is a set of tensors $T=(t^j)_{j\in J}$ with 
$t^j \in \CC^{n_1^j\times n_2^j\times \dots \times n_{K^j}^j}$ 
together with a contraction $C$ of their tensor product $t^1 \otimes t^2 \otimes \dots \otimes t^{J}$, 
where we use the convention 
\begin{equation}\label{eq:tensors}
		(t^1 \otimes t^2 \otimes \dots \otimes t^{J})_{i_1, i_2, \dots, i^{K}}
		=
		t^1_{i_1, i_2, \dots, i_{K^1}}
		t^2_{i_{K^1+1},\dots, i_{K^1+K^2}} \dots 
		t^J_{i_{K-K^J+1},\dots , i_{K}} 
\end{equation}
with $K \coloneqq \sum_{j=1}^J K^j$. 
For short, we write $C(T) \coloneqq C(t^1 \otimes t^2 \otimes \dots \otimes t^{J})$. 
We say that a tensor network $(T,C)$ with tensors $T=(t^j)_{j\in J}$ has a \emph{self-contraction} if there is a tensor $t^j$ 
such that both pointers of one of the pairs $\{k_l,k'_l\}$ defining $C$ point to indices of $t^j$.  
We call a tensor network $(T,C)$ \emph{closed} if $C(T)$ is a number, i.e., if $\mc I$ is empty. 

The relevance of tensor networks comes from the fact that contractions of \emph{certain} tensor networks (such as \emph{matrix product states}) 
can be calculated efficiently, 
while only to store general tensors requires exponentially many parameters in the total number of indices. 
In general, estimating the outcome of a contraction is a \sharpP-hard problem \cite{GhaLanShi15}. 
However, there is a simple and natural upper bound, which might already be known. 
At least for sake of a self-contained presentation we provide a proof below. 

\begin{proposition}[Bound on tensor network contractions]\label{prop:TNbound}
	Let $(T,C)$ be a tensor network with $J\geq 2$ tensors $T=(t^j)_{j\in [J]}$ and contraction $C$ without self-contractions. 
	Then
	\begin{equation}
		\fnorm{C(T)} \leq \prod_{j=1}^J \fnorm{t^j} . 
	\end{equation}
\end{proposition}

If a tensor network has a self-contraction this statement can fail to hold. 
Such an example can be easily constructed by taking the tensors so that the trace of an identity matrix occurs. 

Here we briefly sktech the proof of the tensor network bound and provide a full proof in Appendix~\ref{sec:Pf_TN_bound}. 

\begin{proof}[Proof idea of Proposition~\ref{prop:TNbound}]
	Reshaping $t^j$ suitably into matrices $\tilde t^j$, the tensor network $C(T)$ can be rewritten as a matrix product of $\tilde t^j\otimes \1$ sandwiched between vectorizations $\bra{t^1}$ and $\ket{t^J}$ of $t^1$ and $t^J$, 
	\begin{equation}
		C(T) = \sandwich{t^1}{\tilde t^2 \otimes \1 \dots \tilde t^{J-1}\otimes \1}{t^J} \, .
	\end{equation}
	Hence, 
	\begin{equation}
	\begin{aligned}
	|C(T)| &\leq  
	\lTwoNorm{\bra{t^1}} \snorm{\tilde t^2\otimes \1} \dots \snorm{\tilde t^{J-1}\otimes \1} \lTwoNorm{\ket{t^J}}
	\\
	&= \fnorm{t^1} \snorm{\tilde t^2} \dots \snorm{\tilde t^{J-1}} \fnorm{t^J} \, ,
	\end{aligned}	
	\end{equation}
	where we have used that $\snorm{A\otimes B} = \snorm{A}\snorm{B}$, $\snorm{\1}=1$, and that the Frobenius norm is the $\ell_2$-norm of a vectorization. 
	Then the bound $\snorm{\tilde t^j}\leq \TwoNorm{\tilde t^j} = \fnorm{t^j}$ between the Schatten  norms finishes the proof. 
\end{proof}

\subsection{Proofs of the reconstruction theorems}\label{sec:ProofsThms}
In this section we prove generalized versions of the theorems presented in Section~\ref{sec:results} providing recovery guarantees for the constrained trace and diamond norm regularization \eqref{eq:CTrNormRec} and \eqref{eq:dNormRec} and the CPT-fit~\eqref{eq:CPTRec}. 

\subsubsection{Constrained trace norm minimization}
In this section we prove a stable and robust reconstruction guarantee for:
\begin{equation}
T^{\ast c}_{\eta,q} \label{eq:CTrNormRecq} 
= 
\argmin\{ \tnorm{J(T)}: T\in \HT(\CC^n),\ \lqNorm{\A(T) - y}\leq \eta\} 
\end{equation}
with $q \geq 1$. The reconstruction procedure $T^{\ast c}_\eta$ introduced in \eqref{eq:CTrNormRec} is a special case, namely $q=2$, of this class of minimization problems.

\begin{theorem}[Stable and robust reconstruction from trace norm minimization]
	\label{thm:TrNorm}
	Let $\A : \CPT(\CC^n) \to \RR^m$ given by
	$\A(T)_j = \Tr[A_j T(\ketbra{\psi_j}{\psi_j})]$, 
	where 
	$(A_i, \ket{\psi_i})_{i \in [m]}$ is a 4-generic measurement ensemble (Definition~\ref{def:measurement_ensemble})
	and	the observables' traceless part is $\tilde A_0 \coloneqq A_0 - \Tr[A_0] \1/n \neq 0$.

	Then there are constants $C$ and $\lambda$ such that the following holds. 
	Fix $q \geq 1$, some Kraus rank $r$, a number of measurement settings
	\begin{equation}
		m \geq C \, r\, \ln(n) n^2 
	\end{equation}
	and let $p\in [1,2]$. 
	Then, with probability at least 
	$1-\e^{-\lambda m}$, for all $T_0 \in \CPT(\CC^n)$ the solution $T_{\eta,q}^{\ast c}$ of the minimization \eqref{eq:CTrNormRecq} with 
	$y = \A(T_0)+e$ approximates $T_0$ with an error 
	\begin{equation}
			\bigl\| J(T_0-T^{\ast c}_\eta) \bigr\|_p
			\leq \\
			\frac{4}{r^{1-1/p}} \tnorm{J({T_0}\c)}  
			+ 
			\tilde c\, \frac{n^2 r^{1/p-1/2}}{m^{\frac{1}{q}} \TwoNormb{\tilde A}} \, \frac{\eta}{\snorm{A_0}}
	\end{equation}
	provided that $\lqNorm{e} \leq \eta$. 
	The constants $C$, $\lambda$, and $\tilde c$ only depend on each other. 
\end{theorem}

\begin{proof}
Thanks to Observations~\ref{obs:centering} and~\ref{obs:normalizing} it is enough to prove the theorem of the case where $A_0$ is traceless (i.e., $A_0 = \tilde A_0$) and normalized to $\snorm{A_0}=1$. 
		
We prove the theorem by establishing the subspace NSP \eqref{eq:NSP} for Choi matrices. 
The main technical ingredient to prove the NSP is the bound on the minimum conic singular value from Theorem~\ref{thm:lambda_min}. 
Then Theorem~\ref{prop:NSP_recovery_guarantee} will give the desired result. 

We prove the subspace NSP with the subspace being $J\HT(\CC^n)_0$, i.e., 
the Choi matrices of hermiticity preserving and trace annihilating maps. 
It is helpful to consider two cases. 
First, operators $X\in J\HT(\CC^n)_0$ satisfying 
$\TwoNormb{X\r}\leq \frac{\mu}{\sqrt r} \tnormb{X\c}$ 
are effectively of high rank and satisfy the NSP by default. 
Thus, it suffices to consider the case where 
$\TwoNormb{X\r} \geq \frac{\mu}{\sqrt r} \tnormb{X\c}$. 
Every such matrix satisfies
\begin{equation}
\begin{aligned}
\tnorm{X}
&\leq 
\tnormb{X\c} + \tnormb{X\r}
\\
&\leq 
\frac{1+\mu}{\mu} \sqrt{r}\, \TwoNormb{X\r} 
\\
&\leq 
 \frac{1+\mu}{\mu} \sqrt{r}\, \TwoNorm{X} \, . \label{eq:rank-r-part_smaller}
\end{aligned}
\end{equation}
Hence, $M = J^{-1}(X)$ is contained in the cone $K$ from Eq.~\eqref{eq:def:K} with 
$c_\mu = \frac{1+\mu}{\mu}$. 
For any $q \geq 1$,  Theorem~\ref{thm:lambda_min} yields the bound
\begin{equation}
	\inf_{M \in K} \frac{\norm{\A(M)}_{\ell_q} }{\fnorm{M} }
	\geq 
	\kw \tau 
\end{equation}
with probability over $\A$ at least $1-\e^{-\lambda m}$, where
\begin{equation}
	1/\tau \geq 
	\tilde C
	\TwoNorm{A}
	\frac{m^{\frac{1}{q}}}{n^2} \
\end{equation}
with $\tilde C$ being a constant only depending on $\lambda$ and $\mu$, and $m\geq m_0$ with $m_0$ given in Eq.~\eqref{eq:m_bound} (with the dependence $m_0 \propto c_\mu^2$). 
This establishes the $(\HT(\CC^n)_0, r, \ell_q)$-NSP for our measurement map $\mc A$ with parameters $\mu$ and $\tau$ given as above. 
Hence, Theorem~\ref{prop:NSP_recovery_guarantee} yields
	\begin{equation}\label{eq:rec_bound_in_pf}
	\begin{split}
	\|J(T_0-T^{\ast c}_{\eta,q})\|_p
	\leq
	\frac{(1+\mu)^2}{(1-\mu)r^{1-1/p}} \Bigl(
	\tnormb{J(T^{\ast c}_{\eta,q})} - &\tnorm{J(T_0)} + 2\tnorm{J({T_0}\c)} \Bigr)
	\\
	&+\tau r^{1/p-1/2}\, \frac{3+\mu}{1-\mu} 
	  \lqNorm{\A(T_0-T^{\ast c}_{\eta,q})}
	\end{split}
	\end{equation}
	for any $\mu \in \left]0,1\right[\,$. 
	The minimization in \eqref{eq:CTrNormRecq} assures $\tnormb{J(T^{\ast c}_{\eta,q})} \leq \tnorm{J(T_0)}$ by construction, because $J(T_0)$ is a feasible point of this optimization problem.
	Moreover,
	\begin{align}
	\lqNorm{\A(T_0-T^{\ast c}_{\eta,q})}
	&\leq 
	\lqNorm{\A(T_0)-y} + \lqNorm{y-\A(T^{\ast c}_{\eta,q})} \nonumber 
	\\
	&\leq 
	2 \eta 
	\label{eq:TwoEtaBound}
	\end{align}
	since $\A(T_0)-y=e$ with $\lqNorm{e} \leq \eta$ and 
	$\lqNorm{y-\A(T^{\ast c}_\eta)} \leq \eta$ due to a constraint in the minimization \eqref{eq:CTrNormRec}.
	Choosing $\mu = \sqrt{5} -2$ leads to $(1+\mu)^2/(1-\mu) = 2$. 
	Simplifying the bound \eqref{eq:rec_bound_in_pf} with these observations 
	completes the proof. 
\end{proof}

\subsubsection{CPT-fit}

In this section, we prove a stable and robust reconstruction guarantee for the following generalization of the CPT-fit \eqref{eq:CPTRec}:
\begin{equation}
T^{\ell_q}
\coloneqq 
\argmin\{ \lqNorm{\A(T) - y} : \, T\in \CPT(\CC^n)\} \, , \label{eq:CPTRecq}
\end{equation}
where $q \geq 1$ arbitrary. Similar to before, the CPT-fit discussed in the introductory section arises from fixing $q=2$.

\begin{theorem}[Stable and robust reconstruction from CPT-fit]\label{thm:CPTfit}
Consider a 4-generic measurement setup as in Theorem~\ref{thm:TrNorm} and the same constants $C$ and $\lambda$. 
Fix $q \geq 1$, some Kraus rank $r$, a number of measurement settings
	\begin{equation}
		 m \geq C \, r\, \ln(n) n^2 
	\end{equation}
and $p\in [1,2]$. 
    Then, with probability at least $1-\e^{-\lambda m}$, for all $T_0 \in \CPT(\CC^n)$ the solution $T^{\ell_q}$ of \eqref{eq:CPTRecq} with $y = \A(T_0)+e$ approximates $T_0$ with an error 
	\begin{equation}
			\normb{ J(T_0-T^{\ell_q}) }_{p}
			\leq \\
			\frac{4}{r^{1-1/p}} \tnorm{J({T_0}\c)} 
			+ \tilde c \frac{n^2  r^{1/p-1/2}}{ m^{\frac{1}{q}} \TwoNorm{\tilde A_0} } \frac{\lqNorm{e}}{\snorm{A_0}}
			\, .
	\end{equation}
	The constants $C$, $\lambda$, and $\tilde c$ again only depend on each other. 
\end{theorem}

Before coming to the actual proof we provide some intuition why this theorem is almost a corollary of Theorem~\ref{thm:TrNorm}. 
Adding the positivity constraint $J(T) \geq 0$ to the trace norm minimization \eqref{eq:CTrNormRecq} from Theorem~\ref{thm:TrNorm} can only improve the recovery of CPT maps. 
But as $\tnorm{J(T)} = n$ for all $T \in \CPT(\CC^n)$, the constrained minimization in \eqref{eq:CTrNormRecq} degenerates into a feasibility problem.
Hence, to achieve a further improvement in the reconstruction, it is natural to minimize the residual $\lqNorm{\A(T) - y}$ instead. 
This is exactly what the CPT-fit does. 

\begin{proof}[Proof of Theorem~\ref{thm:CPTfit}]
	This proof is analogous to the proof of Theorem~\ref{thm:TrNorm} with two exceptions: 
	First, $\tnorm{J(T^{\ell_q})}= \tnorm{J(T_0)}=n$, because both are constrained to be Choi matrices of CPT maps. 
	Consequently, their difference vanishes and \eqref{eq:rec_bound_in_pf} becomes 
	\begin{equation}
	\norm{J(T^{\ell_q})}_p
	\leq
	\frac{(1+\mu)^2}{(1-\mu)r^{1-1/p}} \cdot
	2 \tnorm{J({T_0}\c)} 
	\\
	+\tau r^{1/p-1/2}\, \frac{3+\mu}{1-\mu} 
	  \lqNorm{\A(T_0-T^{\ast c}_{\eta,q})} .
	\end{equation}
	Secondly, the minimization \eqref{eq:CPTRecq} assures
	\begin{equation}
	\begin{aligned}
	\lqNorm{\A(T_0-T^{\ell_q}_\eta)} 
	\leq &
	\lqNorm{\A(T^{\ell_q}_\eta)-y} +\lqNorm{y-\A(T_0)} \\
	\leq & 2 \lqNorm{e},
	\end{aligned}
	\end{equation}
	because $T_0$ is a feasible point satisfying $\lqNorm{y-\A(T_0)} = \lqNorm{e}$
	and $T^{\ell_q}$ is the minimizer.
\end{proof}

\subsubsection{Constrained trace and diamond norm minimization}
In Ref.~\cite{KliKueEis16} is shown that certain recovery guarantees for trace norm regularization (such as \eqref{eq:TrNormRec} and \eqref{eq:CTrNormRec}) automatically imply recovery guarantees for the analogous diamond norm regularization (such as \eqref{eq:dNormRec} and \eqref{eq:CdNormRec}). 
This holds when the recovery guarantees can be proven as, e.g., in Ref.~\cite{Tro15} via the descent cone of the reguralizer. 
In this section, we follow this strategy in order to obtain a recovery guarantee for diamond norm regularization in our quantum process tomography setting. 

\myparagraph{An error bound from the descent cone}
For the proof of Theorem~\ref{thm:TrNorm} we use an error bound relying on the so-called descent cone of the function that is minimized in the reconstruction. 
The descent cone of a function is the cone of directions in which the function decreases \cite{Tro15}: 

\begin{definition}[Descent cone]\label{def:DC}
	Let $\V$ be an affine space and $f: \V \to \overline\RR$ be a proper convex function. 
	The \emph{descent cone} of $f$ at a point $x\in \V$ is 
	\begin{equation}
		\DC(f,x) \coloneqq \cone\, \{u\in \V_0 : \ f(x+ u) \leq f(x) \} \, .
	\end{equation}
\end{definition}

The following error bound is the basis for many recovery guarantees from bounds to the conic singular value from Definition~\ref{def:conic_SV}.  
It has been proven in Ref.~\cite{ChaRecPar12} (where the descent cone is given as a tangent cone) and has later been restated by Tropp \cite{Tro15} for optimizations over a vector space. 
One can easily see from its proof that it also holds if one optimizes over an affine space and chooses arbitrary $\ell_q$-norms ($q \geq 1$) to measure the size of the minimum conic singular value.

\begin{proposition}[Error bound for convex recovery, Tropp's version {\cite[Proposition~2.6]{Tro15}}]
	\label{prop:general_reconstruction}
	Let $x_0 \in \V$ be a signal, $\A\in \V \to\RR^m$ be an affine linear measurement map, 
	$y = \A(x_0)+\ev$ a vector of $m$ measurements with additive error $\ev \in\RR^m$. Fix $q\geq 1$ and let $x^f_{\eta,q}$ be the solution of the optimization 
	\begin{equation}\label{eq:general_reconstruction}
		x^f_{\eta,q} = \argmin\{ f(x) \mid x \in \V: \ \lqNorm{\A(x) - y} \leq \eta\}
	\end{equation}
	for a convex function $f: \V \to \overline{\RR}$.
	If $\lqNorm{\ev} \leq \eta$ then
	\begin{equation}
		\fnormb{x^f_\eta - x_0}\leq \frac{2\eta}{\lambda_{\min{}} (\A;\DC(f,x_0),\ell_q)} \, .
	\end{equation}
where $\lambda_{\min{}} (\A;\DC(f,x_0),\ell_q)$ has been defined in \eqref{eq:lambdamin}.
\end{proposition}

\begin{proof}
	In the proof of Ref.\ \cite[Proposition~2.6]{Tro15} one uses Definition~\ref{def:conic_SV} of the conic singular value only for element in $\V-\V$.
Generalizing the proof from errors measured in Euclidean norm ($q=2$) to any $\ell_q$-norm is also straightforward, provided that the definition of the minimum conic singular value is properly adjusted.
\end{proof}

\myparagraph{Our recovery guarantee for the diamond norm}
First we prove a weaker version of Theorem~\ref{thm:TrNorm} with a different argument. 
In turn this reconstruction guarantee for minimization~\eqref{eq:CTrNormRecq} will imply a recovery guarantee for the following generalization of diamond norm reconstruction:
\begin{equation}
T^{\diamond c}_{\eta,q} 
=
\argmin\{ \dnorm{T}: T\in \HT(\CC^n),\ \lqNorm{\A(T) - y}\leq \eta\} .
\label{eq:CdNormRecq}
\end{equation}
Note that the minimization~\eqref{eq:CdNormRec} discussed in the main text is a special case of \eqref{eq:CdNormRecq}, where $q=2$.

\begin{theorem}[Stable reconstruction from trace norm minimization]
\label{thm:TrNorm2}
Consider a measurement setting as in Theorem~\ref{thm:TrNorm} with possibly different  constants $C$ and $\lambda$. 
Fix some Kraus rank $r$ and 
	\begin{equation}
		m \geq C \, r\, \ln(n) n^2 \, ,
	\end{equation}
	as well as $q\geq 1$.
	Then, with probability at least $1-\e^{-\lambda m}$, for all 
	$T_0 \in \HT(\CC^n)$ with Kraus rank $\rank(J(T_0))\leq r$ the solution $T^{\ast c}_{\eta,q}$ of the minimization \eqref{eq:CTrNormRecq} with 
	$y = \A(T_0)+e$ approximates $T_0$ with an error 
	\begin{equation}
		\bigl\| J(T_0-T^{\ast c}_\eta) \bigr\|_2
		\leq  
		\tilde c \frac{n^2}{ m^{\frac{1}{q}} \TwoNormb{\tilde A_0} } \, \frac{\eta}{\snorm{A_0}}
	\end{equation}
	provided that $\lqNorm{e} \leq \eta$. 
	The constants $C$, $\lambda$, and $\tilde c$ again only on each other. 
\end{theorem}

\begin{proof}
	Thanks to Observations~\ref{obs:centering} and~\ref{obs:normalizing} it is enough to prove the theorem of the case where $A_0$ is traceless (i.e., $A_0 = \tilde A_0$) and normalized to $\snorm{A_0}=1$. 
	This proof relies on Proposition~\ref{prop:general_reconstruction} where the reconstruction error is bounded in terms of the $\ell_q$-minimum conic singular value of $\A$ w.r.t.\ the descent cone of the trace norm at $T_0$. 
	
	From Ref.\ \cite[Lemma~10]{KueRauTer15} it follows that any $T_0\in \HT(\CC^n)$ with $\rank(J(T_0))\leq r$ and any 
	$M \in \DC(\tnorm{J(\argdot)}, T_0)$ the H\"older-type inequality 
	\begin{equation}
		\tnorm{J(M)} \leq 2\sqrt{r} \TwoNorm{J(M)} 
	\end{equation}	
	is satisfied. 
	Hence, we choose $c_\mu$ in Theorem~\ref{thm:lambda_min} to be $c_\mu = 2$. 
	Then the theorem implies that with probability at least $1-\e^{-\lambda m}$
	the minimum conic singular value $\lambda_{\min{}}$ of $\A$ w.r.t.\ the descent cone $\DC(\tnorm{J(\argdot)}, T_0)$ is bounded as 
	\begin{equation}
		\lambda_{\min{}}\left( \A; \DC(\tnorm{J(\argdot)},T_0); q \right)
		\geq 
		\tilde C\,
		\frac{\TwoNorm{A} m^{\frac{1}{q}}}{n^2} \, \quad \forall q \geq 1 ,
	\end{equation}
	where $\tilde C$ is a constant only depending on $\lambda$. 
	Proposition~\ref{prop:general_reconstruction} finishes the proof. 	
\end{proof}

Now, the results from Ref.~\cite{KliKueEis16}, specifically \cite[Implication~9]{KliKueEis16} tell us that the diamond norm minimization performs at least as well as the trace norm minimization. 

\begin{corollary}
[Stable reconstruction from diamond norm minimization \cite{KliKueEis16}]
\label{thm:dNorm}
Choosing the diamond norm minimization~\eqref{eq:CdNormRec} instead of the trace norm minimization~\eqref{eq:CTrNormRec} in Theorem~\ref{thm:TrNorm2} yields a smaller smaller reconstruction error, 
	\begin{equation}
		\TwoNormb{ J(T_0-T^{\diamond c}_\eta) }
		\leq 
		\TwoNormb{ J(T_0-T^{\ast c}_\eta) } \, .
	\end{equation}
\end{corollary}

\subsection{Reconstruction from approximate 4-generic measurements}
\label{sec:approximate}
As stated, our recovery guarantees hold if the input states are drawn from a complex $4$-design and the output states are measured with observables that have unitary $4$-designs as eigenbases. 
Here, we show that for the recovery guarantees to hold $\epsilon/n^4$-approximate $4$-designs are enough, in the sense that the $\epsilon$ only changes the constants in the recovery guarantee but not directly the reconstruction error. 
This significantly lesses the burden for experimental realizations. 
Similar generalization have been proven for low-rank matrix recovery \cite{KueRauTer15} and we will extend those arguments here. 
$\epsilon$-approximate designs can be implemented, e.g., by local random quantum circuits \cite{BraHarHor12,NakHirKoa17,HarMeh18} or fluctuating Hamiltonians \cite{OnoBueKli17}, where the $\epsilon$ can be reduced exponentially quickly. 

For some probability distribution $\mu$ on the sphere in $\CC^n$ we denote the corresponding average of $\ketbra \psi \psi^{\otimes k}$  by 
\begin{equation}
	\psi^{(k)}_\mu \coloneqq \EE_{\ket \psi \sim \mu}\myleft[ \ketbra \psi \psi^{\otimes k} \myright] 
\end{equation}
and set $\psi^{(k)} \coloneqq \psi^{(k)}_\Haar$ to be the uniform average. 
We call $\mu$ and \emph{$\epsilon$-approximate spherical $k$-design} if
\begin{equation}
	\snorm{\psi^{(k)}_\mu - \psi^{(k)}}
	\leq \epsilon \snorm{ \psi^{(k)} }
	= \frac{\epsilon}{\binom{n+k-1}{k}} \, .
\end{equation}
Such an $\epsilon$-approximate $k$-design is also an $\epsilon$-approximate $k'$-design for any 
$k'\leq k$ (see, e.g., \cite[Lemma~16]{KueRauTer15}). 

Analogously, given a probability measure $\nu$ on $U(n)$ we define the $k$-th moment (super)operator by
\begin{equation}
	\G_\nu^{(k)}(X) 
	\coloneqq 
	\EE_{U \sim \nu} \myleft[ U^{\otimes k} X U^{\dagger \otimes k} \myright] 
\end{equation}
and set $\G^{(k)} \coloneqq \G_\Haar^{(k)}$. 
We define $\nu$ to be an \emph{$\epsilon$-approximate unitary $k$-design}
if for all traceless product operators $X$
\begin{equation}
\begin{aligned}
	\snorm{ \G_\nu^{(k)}(X) - \G^{(k)}(X) }
	&\leq 
	\epsilon \snorm{\G^{(k)}(X)} \, . 
\end{aligned}
\end{equation}
This bound means that the traceless part of observables is not changed too much. 
By the embedding $X \mapsto X \otimes \1_n$ it is easy to see that an $\epsilon$-approximate unitary $k$-design is also an $\epsilon$-approximate unitary $k'$-design for all $k'\leq k$. 

There are also other definitions of approximate designs, which are partially discussed, e.g., in Low's PhD thesis \cite[Section~2.2]{Low10}. 
However, for our purposes these definitions are the most natural ones. 

We note that for a traceless product operator $X$
\begin{equation}
	\snorm{\G^{(k)}(X)} \leq \frac{c(k)}{n^k} \TwoNorm{X} \, ,
\end{equation}
which follows from Eq.~\eqref{eq:decompose_U_average}, expanding $P_\lambda$ in terms of permutations $\sigma$, viewing $\Tr[\sigma X]$ as tensor network, and using Proposition~\ref{prop:TNbound}. 
Hence, any $\epsilon$-approximate unitary $k$-design from one of the other definitions \cite[Section~2.2]{Low10} is also an $\epsilon'$-approximate unitary $k$-design according to our definition. 
As usual, there is some dimension factor overhead when going from one definition to another. 

\begin{definition}[$\epsilon$-approximate $4$-generic measurement ensemble] 
\label{def:measurement_ensemble_approximate}
We call a measurement ensemble $\left(A, \ket{\psi} \right)$ with observable~$A$ acting on $\CC^n$ and state~$\ket{\psi}$ in $\CC^n$ \emph{$\epsilon$-approximate 4-generic} if it fulfills the following criteria:
\begin{compactenum}[i)]
	\item 
	The distribution of $\ket{\psi}$ in $\CC^n$ is an $\epsilon$-approximate spherical $4$-design.
	\item 
	$A=U A_0 U^\dagger$, where $A_0 \in \Herm(\CC^n)$ is fixed and $U$ in $\U(n)$ is an 
	$\epsilon$-approximate unitary $4$-design. 
\end{compactenum}
The measurement ensemble is called \emph{normalized $\epsilon$-approximate 4-generic} if the observables are traceless and normalized in spectral norm, i.e. $\Tr[A_0] = 0$ and $\snorm{A_0} = 1$. 
\end{definition}

\begin{theorem}[$\lambda_{\min{}}$ for $\epsilon/n^4$-approximate $4$-generic measurements]
\label{thm:approximate}
Theorem~\ref{thm:lambda_min} also holds for $\epsilon$-approximate $4$-generic measurements with a smaller constant $c>0$ whenever 
$\epsilon \leq \frac{c_4}{n^4}$, where $c_4>0$ is an absolute constant. 
\end{theorem}

As a direct consequence, all reconstruction guarantees proven in this work also hold for $\epsilon/n^4$-approximate $4$-generic measurements with different absolute constants. 

The theorem is a consequence of the two following lemmas. 

\begin{lemma}[$Q_\xi$ for $\epsilon/n^4$-approximate $4$-generic measurements]
Lemma~\ref{lem:our_lower_tail_bound} also holds for $\epsilon$-approximate $4$-generic measurements with a smaller constant $c>0$ whenever 
$\epsilon \leq \frac{c_4}{n^4}$, where $c_4>0$ is an absolute constant. 
\end{lemma}

\begin{proof}
Let $\mu$ and $\nu$ be an $\epsilon$-approximate spherical and unitary $k$-design, respectively.  
We will use that for any operators 
$G,B\in \L(\CC^d)$ and superoperator 
$M\in \M(\CC^d)$
\begin{equation}
	|\Tr[G M(B)]| 
	\leq 
	\TwoNorm{G} \snorm{M} \TwoNorm{B} 
	\leq 
	d \snorm{G} \snorm{M} \snorm{B} \, .
\end{equation}
Moreover, we denote
$M_{2k}\coloneqq M^{k,k} = M^{\otimes k} \otimes M^{\ast \otimes k}$. 
Then, for any traceless and normalized observable $A \in \Herm(\CC^n)$ and $l = 2k$,
\begin{equation}
\begin{aligned}
&\phantom{\leq.}
\myleft| \Tr\myleft[ \G^{(l)}_\nu(A^{\otimes l}) M_l(\psi^{(l)}_\mu) \myright]
		-\Tr\myleft[ \G^{(l)}(A^{\otimes l}) M_l(\psi^{(l)}) \myright] \myright|
\\
&\leq 
	\myleft| \Tr\myleft[ \bigl(\G^{(l)}_\nu(A^{\otimes l}) -\G^{(l)}(A^{\otimes l}) \bigr) M_l\bigl(\psi^{(l)}_\mu-\psi^{(l)}\bigr) \myright] \myright|
	\\&\qquad
	+\myleft| \Tr\myleft[ \G^{(l)}(A^{\otimes l}) M_l\bigl(\psi^{(l)}_\mu-\psi^{(l)}\bigr) \myright] \myright|
	\\& \qquad 
	+ \myleft| \Tr\myleft[ \bigl(\G^{(l)}_\nu(A^{\otimes l}) - \G^{(l)}(A^{\otimes l})\bigr) M_l(\psi^{(l)}) \myright] \myright|
\\
&\leq 
	n^l \snorm{\bigl(\G^{(l)}_\nu(A^{\otimes l})-\G^{(l)}(A^{\otimes l})\bigr)}
		\snorm{M_l}  \snorm{\psi^{(l)}_\mu-\psi^{(l)}}
	\\&\qquad 
	+ n^l \snorm{\G^{(l)}(A^{\otimes l})}  \snorm{M_l}  
		\snorm{\psi^{(l)}_\mu-\psi^{(l)}} 
	\\&\qquad 
	+ n^l \snorm{\G^{(l)}_\nu(A^{\otimes l}) - \G^{(l)}(A^{\otimes l})}
		\snorm{M_l}  \snorm{\psi^{(l)}}
\\
&\leq 
	\frac{c(l)}{n^l} \TwoNorm{A}^l \snorm{M}^l \epsilon 
\, ,
\end{aligned}
\end{equation}
where $c(l)$ is a constant only depending on $l$. 
We note that $\snorm{M}^2 \leq \TwoNorm{M(\1)}^2 + \fnorm{M}^2$ and choose a small enough absolute constant $c_4>0$ and 
\begin{equation}
	\epsilon
	\leq
	\frac{c_4}{n^4} 
\end{equation}
so that the ($l=4$)-th moment \eqref{eq:fourth_moment_bound} and the second moment \eqref{eq:full_second_moment} (with traceless $A$) change only by a constant. 
This proves the lemma. 
\end{proof}

\begin{lemma}[$W_m$ for $\epsilon$-approximate 4-generic measurements]
Lemma~\ref{lem:W} still holds for $\epsilon$-approximate $4$-generic measurements with $\TwoNorm{A}$ replaced by $(1+\epsilon) \TwoNorm{A}$. 
\end{lemma}

\begin{proof}
The only thing that needs slightly to be changed in the proof of Lemma~\ref{lem:W} is the bound \eqref{eq:H1_squared} on $\snorm{\EE[H_1^2]}$. 
Using the definitions above with $k=1$, we have  
\begin{align}
\snorm{\psi^{(1)}_\mu} 
&\leq 
\snorm{\psi^{(1)}_\mu - \psi^{(1)}} + \snorm{\psi^{(1)}}
\leq 
(1+\epsilon) \snorm{\psi^{(1)}} 
\end{align}
and, similarly, 
\begin{align}
\snorm{\G^{(1)}_\nu(X)} 
\leq 
(1+\epsilon)\snorm{\G^{(1)}(X)} \, . 
\end{align}
Hence, 
\begin{equation}
\snorm{\EE_{\mu,\nu}[H_1^2]}
\leq 
(1+\epsilon)^2 \snorm{\EE[H_1^2]} \, .
\end{equation}
\end{proof}

\subsection{Bosonic and fermionic linear optical circuits}\label{sec:LinarOptics}
It is worth mentioning that the above results largely 
carry over to another important task which is the tomography of bosonic and fermionic circuits. 
For bosonic systems, this refers to  the tomography of linear optical circuits, which play an important role in 
quantum information processing, with such circuits becoming available for a large number of modes
with the advent of integrated optical circuits \cite{RevModPhys.79.135,Carolan711}. 
For fermionic systems, this applies to what is called fermionic linear optics
\cite{FermionicLinearOptics}.
Interestingly, the results laid out above readily apply to both situations, once the objects of interest are
appropriately identified.

For the bosonic setting, consider $n$ modes associated with bosonic annihilation operators $(b_1,\dots, b_n)$
and Hilbert space ${\cal H}_n^B$ in second quantization. The correlation matrix of a quantum state $\rho\in {\cal S}({\cal H}_n^B)$ 
is defined as $C\in \Herm(\CC^n)$ with $C_{j,k}=\Tr(b_j^\dagger b_k \rho)$.
Mode transformations that preserve the boson number are reflected by maps $(b_1,\dots, b_n)\mapsto (c_1,\dots, c_n)$, where
\begin{equation}
	c_j = \sum_{k=1}^n U_{j,k} b_k 
\end{equation}
with $U\in \U(n)$.
Such maps are usually referred to as passive transformations, or linear optical transformations in the 
quantum optical context. Under such mode transformations, correlation matrices transform as
\begin{equation}
	C\mapsto U^\dagger C U \, .
\end{equation}	
Initial correlation matrices reflecting the quantum state 
can hence be taken as 
\begin{equation}
	C= U^\dagger \ketbra 1 1 U \, ,
\end{equation}
 reflecting the situation that mode labelled $1$ is first 
prepared in a Gaussian state
satisfying $\Tr(b_1^\dagger b_1\rho)=1$, while the others are kept in the vacuum
$\Tr(b_j^\dagger b_j\rho)=0$ for $j=2,\dots, n$. 
This is then conjugated by a Haar random unitary
$U\in \U(n)$, giving formally rise to an identical transformation as considered above, with
$\ketbra {\psi_i} {\psi_i}\sim \ketbra \psi \psi$ are i.i.d. realizations of $\ketbra \psi \psi= U^\dagger \ketbra 1 1 U$,
with $U\in \U(n)$ drawn from the Haar measure. Such preparations are optically readily feasible with present technology,
as Gaussian states are particularly accessible with common sources.

Random measurements can again be seen as i.i.d. realizations $A_i\sim  U A U^\dagger$ with $U\in \U(n)$ being Haar random. 
Here $A$ does not take the role of the observable itself, but reflect natural homodyne measurements on the level of
correlation matrices. 
Their expectation values are obtained as $\Tr(A C)$ for correlation matrices $C$. So again, while
the type of measurement is different and the objects involved take an altered physical role, the map realized is formally identical
with
\begin{equation}
	y = \A(T) + e \ \in \RR^m
\end{equation}
with single expectation values
\begin{equation}
	\A(T)_i = \Tr[A_i T(\ketbra{\psi_i}{\psi_i}) ] + e_i\, , 
\end{equation}
for $T(\ketbra{\psi_i}{\psi_i}):= V^\dagger \ketbra{\psi_i}{\psi_i} V$, with $V\in \U(n)$ reflecting the unknown linear process.
In this way, process tomography of the kind discussed here is applicable to the bosonic setting. 
This seems particularly
important with the advent of monolithic bosonic integrated optical devices \cite{Carolan711}, 
as they are, e.g., employed in boson samplers.

Fermionic linear circuits associated with fermionic annihilation operators
$(f_1,\dots, f_n)$ have the same structure (on that level). Again, correlation matrices $C\in \Herm(\CC^n)$ 
transform as 
\begin{equation}
	C\mapsto U^\dagger C U 
\end{equation}
for $U\in \U(n)$, and the same preparations are feasible. 
Here $\ketbra 1 1$
reflects a fermionic Gaussian state in which the first mode contains exactly a single fermion, while the other
$n-1$ modes contain no fermion. 
The same type of measurement is therefore again possible.

\section{Conclusion and outlook}\label{sec:outlook}
We first conclude and then give an outlook towards potential future work. 

\subsection{Conclusion}
We have proven that quantum processes can be reconstructed from an essentially optimal number of expectation values without the requirement of ancillary quantum systems. 
Moreover, by an extensive numerical analysis we have (i) demonstrated the practical feasibility of our approach and (ii) that the reconstruction procedures also work for Pauli-type measurement settings. 
The number of necessary expectation values scales as $\sim r\, n^2 \ln(n)$, where $r$ is the anticipated Kraus rank of the channel. 
The reconstructions are stable against measurement noise and robust against violations of the measured quantum channel having the anticipated Kraus rank. 
In particular, no strict assumptions on the noise level or the Kraus rank are required for a simple fitting procedure (CPT-fit) to be guaranteed to give a good approximation of the measured quantum channel. 
In several physically feasible and realistic setting, the prescriptions laid out give direct
and concrete advice on how to optimally perform quantum process tomography.

\subsection{Outlook}
In this outlook, we present a short outline of several aspects that seem to be interesting starting points for future research. 

\myparagraph{Mixed input states} The first potential generalization concerns the use of pure quantum states in the reconstruction 
procedure.
Numerically, we have observed that our reconstructions work almost equally well when the input states to the channels are mixed.
Finding a recovery guarantee following this observation would be a step towards more practical measurements. 

\myparagraph{The restricted isometry property (RIP) and perspectives for thresholding methods}
The RIP can be adapted to our setting.
A measurement map $\A$ is said to fulfill RIP for Kraus rank $r$ if
\begin{equation}\label{eq:RIP}
	(1-\delta) \fnorm{T} \leq \lTwoNorm{\A(T)} \leq (1+\delta) \fnorm{T}
\end{equation}
for all quantum channels $T \in \CPT(\CC^n)$ with Kraus rank at most $r$. 
The lower RIP bound is --- as in our case --- typically enough to obtain recovery guarantees for optimization procedures. 
But their computational cost is often not optimal. 
For instance, iterative hard \cite{CaiCanShe10,Cha15} and soft \cite{BoyParChu11} thresholding algorithms are faster in many instances. 
But here, recovery guarantees are typically more difficult to prove and also required the upper RIP bound. 
More recently, a non-covex algorithm with a global convergence guarantee has been proposed for quantum state tomography \cite{KyrKalPar18}. 
Analyzing such algorithms in the process tomography setting is an interesting endeavour for future research. 

\myparagraph{Random Clifford unitaries} 
The work \cite{ZhuKueGra16} shows that random Clifford gates are very close to being unitary $4$-designs, and that they provide a precise characterization in terms of irreducible representations of the Clifford group. 
Such tools might be helpful to further relax the requirements on the measurement settings. 
Here, it would be interesting to see if these new insights can be used in order to prove our recovery guarantees with the input states of the channels being random stabilizer states and the bases of the observables random Clifford unitaries. 
This might be particular useful in order to achieve fault tolerance. 

\myparagraph{Diamond norm} 
We have observed numerically (Figure~\ref{fig:noisy}) that the minimum value (the diamond norm of the reconstructed channel) of the unconstrained diamond norm reconstruction \eqref{eq:dNormRec}, \cite{KliKueEis16} is one if the reconstruction error is small and decreases with increasing reconstruction error. 
It would be interesting to also understand this observation analytically. 
In this way one can obtain some confidence about reconstructed quantum channels. 
Of course, also other types of reconstruction certificates would be of great interest. 

\myparagraph{Frequencies and dependent measurements} 
In a typical experiment, one does not only learn the expectation value of an observable $A$, but rather acquires statistics about the POVM associated with its spectral decomposition.
Indeed, it is straightforward to apply our reconstruction procedures to observed frequencies of POVM measurements. 
However, in this work we disregard this finer-grained information.
We have made this decision to avoid technical complications: The various POVM elements associated with any given setting are clearly not independent. However our proof techniques work best for independently drawn measurements and our natural measurements lead to independence in a natural way. 
\\
There are now related theoretical recovery guarantees that can handle some form of dependency -- e.g.\ Refs.~\cite{Vor13,Kue15,acharya2016statistically,acharya2016statistical_analysis,CanLiSol14,GroKraKue15_masked}. 
Applying such techniques to process tomography remains an interesting open problem.
As a first step towards dependent measurements, it would be interesting to numerically compare sample complexities in the cases of our natural measurement setup and the corresponding frequency measurements, with a fair accounting for statistical noise. 

\myparagraph{Sample complexity}
For quantum state tomography, fundamental lower bounds on the sample complexity have been established in Refs.~\cite{HaaHarJi15} and \cite{FlaGroLiu12}. 
They are valid for arbitrary POVM measurements \cite{HaaHarJi15} and measuring Pauli observables \cite{FlaGroLiu12}, respectively. 
Moreover, these works also determine the sample complexity associated with different compressed-sensing based state tomography techniques.
A comparison with the associated fundamental lower bounds shows a close to optimal scaling (at least for low-rank states). 

In contrast to state tomography, very little is known about the sample complexity associated with process tomography. 
A straightforward adaptation of the results \cite{HaaHarJi15} and \cite{FlaGroLiu12} is hindered by:
(i) Typical process tomography measurements -- such as the 4-generic measurements considered here -- have neither a Pauli structure, nor can they be interpreted as state POVMs in a strict sense \cite{Zim08}. 
This makes the task of determining the exact sample complexity associated with a concrete tomographic procedure more difficult. 
(ii) Due to the trace preservation condition, quantum channels are more restricted than quantum states. Suitable packing nets -- a key ingredient in the derivation of the fundamental lower bounds in Refs.\ \cite{HaaHarJi15,FlaGroLiu12} -- 
must fulfill these additional requirements which makes their construction considerably more challenging.
Despite these obstacles, we believe that such a generalization is timely and well-motivated.

\section{Acknowledgements}
We thank I.\ Roth and M.\ Horodecki for helpful discussions. 
MK has been funded by the National Science Centre, Poland (Polonez 2015/19/P/ST2/03001) within the European Union's Horizon 2020 research and innovation programme under the Marie Sk{\l}odowska-Curie grant agreement No 665778.  
RK and DG have received funding from the Deutsche Forschungsgemeinschaft (DFG, German Research
Foundation) under Germany's Excellence Strategy - Cluster of Excellence Matter and Light for Quantum Computing (ML4Q) EXC 2004/1 - 390534769 and Grant ZUK 81, 
the ARO under contract W911NF-14-1-0098 (Quantum Characterization, Verification, and Validation), and
Universities Australia and DAAD's Joint Research Co-operation Scheme (using funds provided by the German Federal Ministry of Education and Research).
RK, DG, and JE have benefitted from the DFG through SPP1798 CoSIP.
Moreover, RK and JE have received funding from
the Templeton Foundation, the ERC (TAQ), the DFG 
(EI 519/9-1, EI 519/7-1, CRC 173, DAEDALUS), and the  MATH+ excellence cluster.
This work has also received funding from the European Union's Horizon 2020
research and innovation programme under grant agreement No.~817482 (PASQuanS).

\appendix
\section*{Appendices}
\addcontentsline{toc}{section}{Appendices}
\renewcommand{\thesubsection}{\Alph{subsection}}
In this appendix, we provide some auxiliary statements in order to keep this work largely self-contained. 
In Appendix~\ref{sec:SDPs}, we provide semidefinite programming formulations of the reconstruction procedures \eqref{eq:TrNormRec}, \eqref{eq:CTrNormRec}, \eqref{eq:dNormRec}, and \eqref{eq:CdNormRec}. 
In Appendix~\ref{sec:Khintchine}, we state a non-commutative Khintchine inequality used in the proof of our bound on the conic minimum singular value, more precisely in the proof of Lemma~\ref{lem:W} with the bound on the mean empirical width. 
Then, finally, we provide some facts about the symmetric group $S_4$ in Appendix~\ref{sec:S4}. 
These facts are also used in the derivation of the bound on the conic minimum singular value in order to bound the fourth moment of our measurements (Lemma~\ref{lem:FourthMoment}). 

\appendix
\subsection{Semidefinite programs for trace and diamond norm reconstruction}%
\label{sec:SDPs}
Our reconstructions can be implemented as semidefinite programs (SDPs) \cite{KliKueEis16}, 
which can practically be solved using standard software such as CVX \cite{cvx,GraBoy08}. 
Let us consider the reconstruction of a quantum channel mapping operators in $\L(\X)$ to operators in $\L(\Y)$. 
The minimization \eqref{eq:TrNormRec} can be rewritten as the following SDP,
\begin{equation}
	\arraycolsep=1.5pt
	\begin{array}{r l l}
		T^\ast_\eta
		=&\underset{T,X,Y}{\argmin} & \frac{1}{2}\bigl( \Tr[X] + \Tr[Y] \bigr),
		\\[.3em]
		&\st &
		\begin{pmatrix}
			X & -J(T) \\ -J(T)\ad & Y
		\end{pmatrix}
		\succeq 0 \, ,
		\\[.9em]
		&&  X,Y \in \Pos(\Y\otimes \X) \, ,
		\\
		&& \fnorm{\A(T)-y} \leq \eta \, .
	\end{array}
\end{equation}
The reconstruction \eqref{eq:CTrNormRec} is obtained by adding the constraint 
$T\ad(\1_\Y) = \1_\X$ into this SDP. 

By only changing the objective function with the spectral norm of \emph{partial} traces $\Tr_\Y$ over the output space $\Y$ of $T:\L(\X) \to \L(\Y)$, we obtain the minimization \eqref{eq:dNormRec} as the following SDP \cite{Wat12,KliKueEis16},
\begin{equation}
	\arraycolsep=1.5pt
	\begin{array}{r l l}
		T^\diamond_\eta
		=&\underset{T,X,Y}{\argmin} & \frac{1}{2}\bigl( \snorm{\Tr_\Y[X]} + \snorm{\Tr_\Y[Y]} \bigr),
		\\[.3em]
		&\st &
		\begin{pmatrix}
			X & -J(T) \\ -J(T)\ad & Y
		\end{pmatrix}
		\succeq 0 \, ,
		\\[.9em]
		&&  X,Y \in \Pos(\Y\otimes \X) \, ,
		\\
		&& \fnorm{\A(T)-y} \leq \eta \, .
	\end{array}
\end{equation}
Again, the constrained minimization \eqref{eq:CdNormRec} is obtained by adding the constraint 
$T\ad(\1_\Y) = \1_\X$ into the SDP.

\subsection{A non-commutative Khintchine inequality}\label{sec:Khintchine}
\begin{theorem}[{\cite[Remark~5.27.2]{Ver12}} with constants from {\cite[Exercise~8.6(d)]{FouRau13}}
	] \label{thm:khintchine}
	Let $B_1,\ldots,B_m$ be self adjoint $N \times N$ matrices and $\epsilon_1,\ldots,\epsilon_m\in \{-1,1\}$ uniformly and independently drawn signs (called a Rademacher sequence). 
	Then
	\begin{equation}
		\Ev \biggl[ \snormB{ \sum_{j=1}^m \epsilon_j B_j} \biggr]
		\leq \sqrt{2 \ln (2N)} \, \snormB{ \sum_{j=1}^n B_j^2}^{1/2}.
	\end{equation}
\end{theorem} 

\subsection{Linear representation of the symmetric group \texorpdfstring{$S_4$}{S4}}
\label{sec:S4}
In order to bound the fourth moment of the measurement map some facts about the representation of the permutation group $S_4$ are helpful. 
In this section, they are summarized and partially derived.
Facts that we just state can, e.g., be found in the Wikis~\cite{WikiS4} and~\cite{WikiRepS4}. 

By $k_1^{j_1}k_2^{j_2}\dots k_l^{j_l}$ we denote the conjugacy class containing permutations composed of each $j_i$ disjoint cycles of sizes $\{k_i\}_{i\in [l]}$, e.g., 
$2^2\subset S_4$ are products of disjoint transpositions. 
Corresponding to each conjugacy class there is one irrep. 
They are given by the Young Frames (see e.g.~\cite{WikiRepS4})
\begin{equation}
	\begin{aligned}
		\mc F_1 &\coloneqq \yng(4) && \text{(trivial rep.)}
		\\
		\mc F_2 &\coloneqq \yng(1,1,1,1) && \text{(sign rep.).} 
		\\
		\mc F_3 &\coloneqq \yng(2,2) && \text{(degree two irreducible rep.)}
		\\
		\mc F_4 &\coloneqq \yng(3,1) && \text{(standard rep.)}
		\\
		\mc F_5 &\coloneqq \yng(2,1,1) && \!\!\!
		\begin{array}{l}
			\text{(product of standard rep.}\\ \text{\phantom{({}}and sign rep.)}
		\end{array}
	\end{aligned}
\end{equation}
We will denote the character corresponding to $\F_i$ by $\chi_i : S_4 \to \ZZ$. 
The characters are constant on the conjugacy classes. 
The sizes of the conjugacy classes and the characters of $S_4$ are (see e.g.~\cite{WikiS4})
\begin{equation}\label{eq:tab:S4}
	\begin{array}{| c " c | c | c | c | c |}
		\hline
		\text{Cycle type} \hspace{-1em}{\phantom{|^{|^|}}} & 1^4 & 2^2 & 2^1 & 4^1 & 3^1
		\\\hline
		\text{\# elements} & 1 & 3 & 6 & 6 & 8
		\\ \thickhline 
		\chi_1 & 1&1&1&1&1
		\\\hline
		\chi_2 & 1&1&-1&-1&1
		\\\hline
		\chi_3 & 2&2&0&0&-1
		\\\hline
		\chi_4 & 3&-1&1&-1&0
		\\\hline
		\chi_5 & 3&-1&-1&1&0
		\\\hline
	\end{array} 
\end{equation}
The dimension of the representation $\F_i$ (in the group algebra) will be denoted by $d_i$ (also called \emph{degree} \cite{WikiRepS4}). 
These dimensions are given by $d_i = \chi_i([1])$. 

Expanding a character $\chi$ in the group algebra yields
\begin{equation}
\begin{aligned}
\chi &= \kw{|S_4|} \sum_{\sigma\in S_4} \chi(\sigma) \langle \sigma, \argdot \rangle
\\
&= \kw{|S_4|} \sum_C \chi(C) \langle \Sigma C, \argdot \rangle \, ,
\end{aligned}
\end{equation}
where $\langle \argdot , \argdot \rangle$ is the inner product of the group algebra, $\sum_C$ denotes the sum over all conjugacy classes $C$ in $S_4$, and $\Sigma C \coloneqq \sum_{\sigma \in C} \sigma$ the sum over the conjugacy class $C$. 
The central minimal projections 
\begin{equation}
p_i = \frac{d_i}{|S_4|} \chi_i
\end{equation}
(see, e.g., Ref.\ \cite[Theorem~III.7.2]{Sim96})
are hence given by
\begin{equation} \label{eq:p_i}
	\begin{aligned}
		p_1 &= \kw{24} \left(\id +\Sigma[2,2]+\Sigma[2]+\Sigma[4]+\Sigma[3]\right) \, ,
		\\
		p_2 &= \kw{24} \left(\id +\Sigma[2,2]-\Sigma[2]-\Sigma[4]+\Sigma[3]\right) \, ,
		\\
		p_3 &= \kw{12} \left( 2\,\id + 2\, \Sigma[2,2] - \Sigma[3] \right) \, ,
		\\
		p_4 &= \kw{8} \left( 3\, \id - \Sigma[2,2] + \Sigma[2] - \Sigma[4]\right) \, ,
		\\
		p_5 &= \kw{8} \left( 3\, \id - \Sigma[2,2] - \Sigma[2] + \Sigma[4]\right) \, ,
	\end{aligned}
\end{equation}
where we have made the identification 
$\langle \sigma, \argdot \rangle \cong \sigma$. 

Now we consider the linear representation $R:S_4\to(\CC^n)^{\otimes 4}$, which is given by permuting the for tensor factors, i.e., the representation of $\sigma\in S_4$ is a unitary 
$R(\sigma)$ on $(\CC^n)^{\otimes 4}$ given by
\begin{equation}
	R(\sigma) = \sum_{i_1,i_2,i_3,i_4=1}^n 
	\ketbra{i_{\sigma(1)}, i_{\sigma(2)}, i_{\sigma(3)}, i_{\sigma(4)}}{i_1,i_2,i_3,i_4} \, .
\end{equation}
Often we write $\sigma$ instead of $R(\sigma)$. 
This representation naturally extends to a representation of the group algebra. 

As $\sum_{i=1}^5 p_i$ is a decomposition of the identity on the group algebra, 
we obtain the decomposition
\begin{equation}
	\1_{(\CC^n)^{\otimes 4}} = \sum_{i=1}^5 P_i 
\end{equation}
with $P_i = R(p_i)$. 
The dimension of the representation $\F_i$ is 
$\Tr[R(p_i)]$ and is given by the dimension $d_i(n)$ of the Schur functor corresponding to $\F_i$ when applied to a vector space of dimension $n$ \cite{WikiRepS4} 
times the degree $\chi_i(\id)$ of the irreducible representation, where
\begin{equation}\label{eq:d_i}
	\begin{aligned}
		d_1(n) &= n\, (n+1)(n+2)(n+3)/24 \, ,
		\\
		d_2(n) &= (n-3)(n-2)(n-1)\, n/24  \, ,
		\\
		d_3(n) &= (n-1)\,n^2\,(n+1)/12 \, ,
		\\ 
		d_4(n) &= (n-1)\,n\,(n+1)(n+2)/8 \, ,
		\\
		d_5(n) &= (n-2)(n-1)\,n\,(n+1)/8 \, .
	\end{aligned}
\end{equation}

\subsection{Proof of the tensor network bound Proposition~\ref{prop:TNbound}}
\label{sec:Pf_TN_bound}
We start with some preliminaries that are helpful for the proof. 
A contraction $C$ of a tensors with $K$ indices is defined by pairs of pointers $\{(k_l,k'_l)\}_{l \in [L]}$. 
For any $\mcL  \subset [L]\subset [K]$, we define the partial contraction $C_{\mcL }$ of a tensor 
$t$ to be the one given by the subset of pointer pairs $\{(k_l,k'_l)\}_{l \in \mcL }$. 
For $\mcL_1\subset [K]$ and $\mcL_2 \subset[K]$ with $\mcL_1\cap \mcL_2 = \emptyset$ we can naturally apply $C_{\mcL_2}$ to $C_{\mcL_1}(t)$ and it holds that 
$C_{\mcL_2}(C_{\mcL_1}(t)) = C_{\mcL_1}(C_{\mcL_2}(t))$ and $C(t) = C_{[K]\setminus \mcL}(C_{\mcL}(t))$. 

We will use the following facts about matrices $A \in \CC^{n_1 \times n_2}$ and $B \in \CC^{n_2\times n_3}$ with dimensions $n_1,n_2,n_3\geq 1$,
\begin{align}
\snorm{AB} &\leq \snorm{A} \snorm{B}\, , \label{eq:submultiplicativity}
\\
\snorm{A} & \leq \TwoNorm{A} = \TwoNorm{A} \, . \label{eq:snorm2fnorm}
\end{align}
For $n_1 = n_3 = 1$ one has $\snorm{A} = \norm{A}_2$ and similarly for $B$. 
We will also use the identity
\begin{equation}\label{eq:normMultiplicativity}
	\snorm{A\otimes B} = \snorm{A} \snorm{B} \, ,
\end{equation}
which holds for arbitrary matrices 
$A \in \CC^{n_1 \times n_2}$ and $B \in \CC^{n_3\times n_4}$. 

Any bipartition of the indices of a tensor $t$ yields a class of unitarily equivalent matrices. 
More specifically, a \emph{matricization} of a tensor 
$t\in \CC^{n_1\times n_2\times\dots \times n_K}$ is a matrix $A$ of which the matrix elements are given by 
\begin{equation}
	A_{(i_{\sigma(1)}, i_{\sigma(2)}, \dots , i_{\sigma(K')}), 
		(i_{\sigma(K'+1)}, \dots , i_{\sigma(K)})}
	=
	t_{i_1, i_2, \dots, i_{K}}
\end{equation} 
for some permutation $\sigma \in S_K$. 
Two such matricizations given by $\tau, \sigma \in S_K$ and the same $K' \in [K]$ are unitarily equivalent if $\{\sigma(i)\}_{i\in [K']} = \{\tau(i)\}_{i\in [K']}$, i.e., if $\sigma$ and $\tau$ yield the same bipartition of the pointers $[K]$.
For a tensor 
$t\in \CC^{n_1\times n_2\times\dots \times n_{K}}$
and a disjoint bipartition $\mcL \dunion \mcR = [K]$ of the pointer set $[K]$ 
we denote by $t_{\mcL,\mcR}$ some matricization of $t$ that comes from this bipartition. 
For any such matricization $t_{\mcL,\mcR}$ holds that 
\begin{equation}\label{eq:invarianceFnorm}
	\TwoNorm{t_{\mcL,\mcR}} = \fnorm{t} \, .
\end{equation}
A \emph{vectorization} of a tensor is a matricization that yields a vector, i.e., a matrix with one column or row. 

\begin{proof}[Proof of Proposition~\ref{prop:TNbound}]
	It is enough to prove the proposition for the case where the tensor network is closed, i.e., where $C(T) \in \CC$. 
	This is so because we can write $\fnorm{C(T)}^2$ always as a closed tensor network, where all previous tensors plus their complex conjugates occur. 
	A tensor and its conjugate have the same norm, which shows the reduction argument. 
	
	We denote the tensors of the tensor network by
	$t^j \in \CC^{n_1^j\times n_2^j\times \dots \times n_{K^j}^j}$ where $j\in [J]$ 
	and the pointer pairs defining the contraction $C$ by 
	$(\{k_l,k'_l\})_{l \in [K]}$ 
	with 
	$K= \sum_{j\in [J]} K^j$.
	
	Let us consider first the case $J=2$. 
	As there are no self-contractions, we can relabel the pointer pairs 
	$(\{k_l,k'_l\})_{l \in [K]}$ so that the $k_l$ belong to $t^1$ and the $k'_l$ to $t^2$.	
	Hence, $C(T)$ is an inner product $C(T) = \braket{t^1}{t^2}$ of matricizations of $t^1$ and $t^2$ into vectors $\ket{t^1}$ and $\ket{t^2}$. 
	Therefore, the Cauchy-Schwarz inequality and the invariance of the Frobenius norm \eqref{eq:invarianceFnorm} prove the proposition for $J=2$. 
	
	Now we consider $J\geq 3$.
	If there are contractions $\{k_l,k'_l\}$ where one of the indices belongs to $t^1$ and one to $t^J$, 
	for each such $l$ we introduce an auxiliary tensor $s^l$ as identity matrix with components 
	$s^l_{m_{k_l},m_{k'_l}} \coloneqq \delta_{m_{k_l},m_{k'_l}}$ that is contracted with $t^1$ and $t^2$ and replaces their contraction $\{k_l,k'_l\}$. 
	This corresponds to tensoring 
	$t^2\otimes t^3 \otimes \dots \otimes t^{J-1}$ with an $n_{k_l} \times n_{k_l'}$ identity matrix. 
	Denote by $s^{l_1}, s^{l_2}, \dots, s^{l_M}$ the tensors that are introduced in this way. 
	We obtain $\tilde C$ from $C$ by, the this modification, so that $\tilde C$ contracts the modified tensor network
	$\tilde T = (t^1, \tilde T', t^J)$ with	
	\begin{equation}\label{eq:Tsplit}
		\begin{aligned}
			\tilde T'&=	S \cup T'\\
			S & \coloneqq \bigl(s^{l_1}, s^{l_2}, \dots , s^{l_M}\bigr)\\
			T' & \coloneqq \bigl(t^2, t^3, \dots , t^{J-1}\bigr)
		\end{aligned}
	\end{equation}
	where $\tilde C$ is obtained from $C$ by adding the contractions 
	$\{k_l, k'_l\}_{l =L+1, \dots, K+M}$ with $s^{l_1}, s^{l_2}, \dots, s^{l_M}$ 
	to the previous ones. 
	With this construction we have $C(T) = \tilde C(\tilde T)$. 
	
	Next, we rename the pointers such that the $k_l$ of the first $K^1$ pairs 
	$(\{k_l,k'_l\})_{l \in [K^1]}$ belong to $t^1$. 	 
	As there are no contractions between $t^1$ and $t^J$, 
	we rename the pointers so that the $k'_l$ of the last $K^J$ pairs
	$(\{k_l,k'_l\})_{l \in [K]\setminus [K-K^J]}$ belong to $t^J$. 
	The remaining pointers pairs are 
	\begin{equation}
		\begin{aligned}
			\tilde{\K}' &\coloneqq \K'\cup \mc M \,,
			\\
			\K' &\coloneqq \{K^1+1, K^1+2, \dots , K-K^J\}\,,
			\\
			\mc M &\coloneqq \{K+1,K+2, L+M\} \, ,
		\end{aligned}
	\end{equation} 
	which contain no further contractions between $t^1$ and $t^J$.
	Now we have achieved that the pointer pairs connected to $t^1$ are 
	$\K^1 \coloneqq [K^1]$ and the ones of $t^J$ are 
	$\K^J \coloneqq [K]\setminus [K-K^J]$, so that they are clearly disjoint. 
	Hence, we can write
	\begin{align}
	C(T) = \tilde C_{\K^1}\Bigl( t^1 \otimes \tilde C_{\K^J}\bigl(\tilde C_{\tilde{\K}'}(\tilde T') \otimes t^J\bigr)\Bigr) \, .
	\end{align}
	In fact $\tilde C_{\K^J}\bigl(\tilde C_{\tilde{\K}'}(\tilde T') \otimes t^J\bigr) \mapsto \tilde C(T)$ is the action of the functional given by 
	$\tilde C_{\K^1}\Bigl( t^1 \otimes \argdot \Bigr)$, which, in turn, is given by a vectorization $\ket{t^1}$ of $t^1$. 
	Hence, we can obtain
	\begin{align}
	|C(T)| \leq \fnorm{t^1} \fnorm{\tilde C_{\K^J}\bigl(\tilde C_{\tilde{\K}'}(\tilde T') \otimes t^J\bigr)} \, .
	\end{align}
	Similarly, $t^J \mapsto \tilde C_{\K^J}\bigl(\tilde C_{\tilde{\K}'}(\tilde T') \otimes t^J\bigr)$ is a linear mapping, where we view 
	$\tilde C_{\tilde{\K}'}(\tilde T')$ as a linear map contracting the indices $\K^J$ and having non-contracted indices $\K^1$. 
	This yields a vectorization $\ket{t^J}$ of $t^J$ an a map is represented by a matricization $\tilde C_{\tilde{\K}'}(\tilde T')_{\K^1,\K^J}$ of $\tilde C_{\tilde{\K}'}(\tilde T')$ 
	so that
	\begin{equation}
	\begin{aligned}
	\fnorm{\tilde C_{\K^J}\bigl(\tilde C_{\tilde{\K}'}(\tilde T') \otimes t^J\bigr)} 
	&=
	\norm{\tilde C_{\tilde{\K}'}(\tilde T')_{\K^1,\K^J} \ket{t^J}}
	\\
	&\leq 
	\snorm{\tilde C_{\tilde \K'}(\tilde T')_{\K^1,\K^J}} \fnorm{t^J} 
	\end{aligned}
	\end{equation}
	Using Eq.~\eqref{eq:Tsplit} we arrive at
	\begin{equation}
		\tilde C_{\tilde \K'}(\tilde T')_{\K^1,\K^J}
		= 
		\bigl(s^{l_1}\otimes s^{l_2}\otimes \dots \otimes s^{l_M} \otimes C_{\K'}(T') \bigr)_{\K^1,\K^J}
		\, . 
	\end{equation} 
	By construction, each auxiliary tensor $s^l$ is an identity matrix with one index in $\K^1$ and one in $\K^J$ and the corresponding matricization has unit spectral norm. 
	Using Eq.~\eqref{eq:normMultiplicativity}, we obtain
	\begin{equation}
		\snorm{\tilde C_{\tilde \K'}(\tilde T')_{\K^1,\K^J}}
		=
		\snorm{C_{\K'}(T')_{{\K'}^1,{\K'}^J}}, 
	\end{equation}
	where ${\K'}^1\subset \K^1$ and ${\K'}^J\subset\K$ are those pointer pairs of $\K^1$ and $\K^2$ have no pointer to any $s^l$ (there are no contractions between $S$ and $T'$). 
	
	Iterating Lemma~\ref{lem:spectral_norm_TN_bound} and using the bound \eqref{eq:snorm2fnorm} and Eq.~\eqref{eq:invarianceFnorm} we obtain that
	\begin{equation}
		\snorm{C_{\K'}(T')_{{\K'}^1,{\K'}^J}}
		\leq 
		\fnorm{t^2} \fnorm{t^3} \dots \fnorm{t^{J-1}}  \, .
	\end{equation}
	This completes the proof. 		
\end{proof}

We use the notation for matricizations introduced right before Proposition~\ref{prop:TNbound} to state the following. 

\begin{lemma}\label{lem:spectral_norm_TN_bound}
	Let 
	$t^1\in \CC^{n_1\times n_2\times\dots \times n_{K^1}}$
	and 
	$t^2\in \CC^{n_{K^1+1}\times n_{K^2+2}\times\dots \times n_{K^1+K^2}}$
	be tensors with index pointers $\K^1 \coloneqq[K^1]$ and $\K^2 \coloneqq [K^1+K^2]\setminus [K^1]$, respectively. 
	Further, let $C_{\mc M}$ be a partial contraction over indices 
	$\mcM = \{(k_l,k'_l)\}_{l \in [L]}$ with 
	$M^1\coloneqq \{k_l\}_{l \in [L]}\subset \K^1$ 
	pointing to indices of $t^1$ and  
	$M^2\coloneqq \{k'_l\}_{l \in [L]}\subset \K^2$ 
	pointing to indices of $t^2$.
	Let $C_\mcM(T)_{L,R}$ be a matricization of $C_\mcM(s,t)$ with row indices $L$ and column indices $R$ where 
	$L\dunion M^1 \dunion M^2 \dunion R = \K^1 \dunion \K^2$. 
	Then 
	\begin{equation}
		\snorm{C_\mcM(s,t)_{L,R}} 
		\leq 
		\snorm{t^1_{L^1,R^1 \cup M^1} } 	
		\snorm{t^2_{M^2\cup L^2,R^2} } \, ,
	\end{equation}
	where 
	$L^1=\K^1 \cap L$, 
	$R^1=\K^1 \cap R$,
	$L^2=\K^2 \cap L$, and
	$R^2=\K^2 \cap R$ are the row/column indices of $t^1$ and $t^2$, respectively. 
\end{lemma}

\begin{proof}
	We write the matricized partially contracted tensor network as a matrix product,
	\begin{equation}
		C_\mcM(s,t)_{L,R}
		=
		\left( t^1_{L^1,R^1 \cup M^1} \otimes \id_{L^2} \right)
		\left( \id_{R_1} \otimes t^2_{M^2\cup L^2,R^2} \right),
	\end{equation}
	where $\id_{L^2}$ denotes the identity matrix with row indices given by $L^2$ and matching column indices and similarly for $\id_{R_1}$;
	e.g., $\id_{\{3,5\}}$ has the matrix components 
	$(\id_{\{3,5\}})_{(i_3,i_5), (i'_3, i'_5)} = \delta_{i_3,i'_3}\delta_{i_5,i'_5}$ for $i_j,i'_j \in [n_j]$. 
	Using Eqs.~\eqref{eq:submultiplicativity} and \eqref{eq:normMultiplicativity} finishes the proof. 
\end{proof}

\bibliographystyle{./myapsrev4-1}

\bibliography{martin,other,jensadditionalreferences}
\end{document}

%% file: mymath.tex
\newtheorem{theorem}{Theorem}
\newtheorem{corollary}[theorem]{Corollary}
\newtheorem{lemma}[theorem]{Lemma}
\newtheorem{proposition}[theorem]{Proposition}
\newtheorem{definition}[theorem]{Definition}
\newtheorem{observation}[theorem]{Observation}

\newcommand{\myleft}{\mathopen{}\mathclose\bgroup\left}
\newcommand{\myright}{\aftergroup\egroup\right}
\newcommand{\e}{\ensuremath\mathrm{e}}
\renewcommand{\i}{\ensuremath\mathrm{i}}
\newcommand{\id}{\ensuremath\mathrm{id}}

\DeclareMathOperator{\Tr}{Tr}

\DeclareMathOperator{\rank}{rank}

\DeclareMathOperator*{\argmin}{arg\,min}

\newcommand{\fro}{\mathrm{F}}

\DeclareMathOperator{\cone}{\operatorname{cone}}
\DeclareMathOperator{\lin}{\operatorname{lin}}

\DeclareMathOperator{\Pos}{Pos}
\DeclareMathOperator{\CPT}{CPT}
\DeclareMathOperator{\Herm}{Herm}
\DeclareMathOperator{\HT}{HT}
\DeclareMathOperator{\TP}{TP}
\newcommand{\Sym}{\mathrm{Sym}}

\DeclareMathOperator{\U}{U}

\DeclareMathOperator{\LandauO}{\mathrm{O}}
\DeclareMathOperator{\tLandauO}{\tilde{\mathrm{O}}}
\DeclareMathOperator{\LandauOmega}{\Omega}

\makeatletter 
\newcommand\complexity@possiblymakesmaller[1]{#1} 
\newcommand\complexity@fontcommand{\mathsf} 
\newcommand{\ComplexityFont}[1]{%
{\ensuremath{\complexity@possiblymakesmaller{\complexity@fontcommand{#1}}}}
}
\makeatother

\newcommand{\sharpP}{\#\ComplexityFont{P}}

\newcommand{\CC}{\mathbb{C}}
\newcommand{\RR}{\mathbb{R}}

\newcommand{\ZZ}{\mathbb{Z}}

\newcommand{\FF}{\mathbb{F}}
\newcommand{\1}{\mathds{1}}
\newcommand{\EE}{\mathbb{E}}
\newcommand{\PP}{\mathbb{P}}

\newcommand{\mc}[1]{\mathcal{#1}}
\newcommand{\X}{\mc{X}}
\newcommand{\Y}{\mc{Y}}
\newcommand{\V}{\mc{V}}
\newcommand{\W}{\mc{W}}
\newcommand{\F}{\mc{F}}
\newcommand{\K}{\mc{K}}
\newcommand{\mcM}{\mc{M}}
\newcommand{\mcL}{\mc{L}}
\newcommand{\mcR}{\mc{R}}

\newcommand{\G}{\mc{G}}

\newcommand{\kw}[1]{\frac{1}{#1}}

\newcommand{\argdot}{{\,\cdot\,}}

\providecommand{\DC}{\operatorname{\mathscr{D}}} 
\newcommand{\ad}{^\dagger}
\renewcommand{\L}{\operatorname{\mathrm{L}}}
\newcommand{\LL}{\operatorname{\mathbb{L}}}
\newcommand{\M}{\operatorname{\mathbb{L}}}

\newcommand{\dunion}{\,\dot{\cup}\,}

\newcommand{\norm}[1]{\left\Vert #1 \right\Vert} 
\newcommand{\normn}[1]{\lVert #1 \rVert} 
\newcommand{\normb}[1]{\bigl\Vert #1 \bigr\Vert} 
\newcommand{\normB}[1]{\Bigl\Vert #1 \Bigr\Vert} 
\newcommand{\snorm}[1]{\norm{#1}_\infty} 

\newcommand{\snormb}[1]{\normb{#1}_\infty}
\newcommand{\snormB}[1]{\normB{#1}_\infty}
\newcommand{\tnorm}[1]{\norm{#1}_{1}} 

\newcommand{\tnormb}[1]{\normb{#1}_{1}}
\newcommand{\TrNorm}[1]{\norm{#1}_{1}} 

\newcommand{\pNorm}[1]{\norm{#1}_p} 

\newcommand{\fnorm}[1]{\norm{#1}_\fro} 

\newcommand{\fnormb}[1]{\normb{#1}_\fro}
\newcommand{\dnorm}[1]{\norm{#1}_\diamond} 

\newcommand{\lTwoNorm}[1]{\norm{#1}_{\ell_2}} 

\newcommand{\lqNorm}[1]{\norm{#1}_{\ell_q}} 

\newcommand{\TwoNorm}[1]{\norm{#1}_{2}} 
\newcommand{\TwoNormn}[1]{\normn{#1}_{2}}
\newcommand{\TwoNormb}[1]{\normb{#1}_{2}}

\newcommand{\ket}[1]{\left.\left|{#1}\right.\right\rangle}

\newcommand{\bra}[1]{\left.\left\langle{#1}\right.\right|}

\newcommand{\braket}[2]{\left\langle #1 \middle| #2 \right\rangle}

\newcommand{\ketbra}[2]{\ket{#1} \!\! \bra{#2}}

\newcommand{\sandwich}[3]
  {\left\langle  #1 \right| #2 \left| #3 \right\rangle}

\renewcommand{\Pr}{\operatorname{\PP}} 
\newcommand{\Ev}{\operatorname{\EE}} 
\newcommand{\A}{\mc{A}} 
\newcommand{\ev}{e} 
\providecommand{\DC}{\operatorname{\mathscr{D}}} 
\newcommand{\st}{\mathrm{subject\ to}}
\newcommand{\Haar}{\mathrm{Haar}}

\renewcommand{\r}{_{|r}}
\renewcommand{\c}{_{\neg r}}
\newcommand{\rec}{\mathrm{rec}}
\newcommand{\Toff}{\mathrm{Toff}}
\newcommand{\dep}{\mathrm{dep}}
\newcommand{\eps}{\mathrm{eps}}

\newcommand{\SetCoordinates}[1]{
	\path (#1.west) ++ (0,\d em) coordinate (#1oli);
	\path (#1.west) ++ (0,-\d em) coordinate (#1uli);
	\path (#1.east) ++ (0,\d em) coordinate (#1ore);
	\path (#1.east) ++ (0,-\d em) coordinate (#1ure);
}
	
\definecolor{niceblue}{rgb}{0.33,0.5,0.8}%
\colorlet{mygreen}{OliveGreen!90!blue}%
\tikzset{%
	sbox/.style = {draw, rounded corners = .5ex,%
		minimum height = 1.6em,%
		minimum width = 1.8em},
	blau/.style = {top color=niceblue!12,%
		bottom color=niceblue!90},
	Bbox/.style = {sbox, blau, inner sep = 1pt},
	leg/.style = {rounded corners = 1ex,thick,postaction={decorate},
					decoration = {markings,mark=at position 0.5 with {\arrow{>}} }
					},
	Leg/.style = {rounded corners = 1ex,thick,postaction={decorate},
		decoration = {markings,mark=at position 0.2 with {\arrow{>}} }
	},
	dir/.style = {gray,thick},
}%